\documentclass{ijuc}
\usepackage[pdftex]{graphics}
\usepackage{url,xcolor}
\usepackage{mathtools}
\usepackage{bookmark}
\usepackage{theoremref}
\usepackage{enumitem}
\usepackage[most,many,breakable]{tcolorbox}
\usepackage{varwidth}
\usepackage{etoolbox}
\usepackage{capt-of}
\usepackage{nameref}
\usepackage{tikz-cd}
\usepackage{cancel}
\usepackage{pgfplots}
\pgfplotsset{compat=newest}
\usepgfplotslibrary{patchplots}
\usepackage{anyfontsize}
\usepackage{sectsty}
\tcbuselibrary{theorems,skins,hooks}
\usepackage{hyperref,cleveref}
\usepackage{subcaption}
\usepackage[disable]{todonotes}
\usepackage{anyfontsize}
\usepackage{amssymb}
\usepackage{amsmath}
\usepackage{amsfonts}
\usepackage{amsthm}
\usepackage{latexsym}
\usepackage{float}
\usepackage{comment}
\usepackage{txfonts}
\usepackage{skak}
\usepackage{scalefnt}
\usepackage{forloop}
\usepackage{fp}

\usetikzlibrary{patterns}
\tikzset{every picture/.style={line width=0.75pt}} 

\newcommand{\emptycell}[4]{
    \tikzstyle{roundnode}=[circle,draw = #3,fill=white, minimum size = 56];
    \node[roundnode] (#4) at (#1,#2) {};
    \draw[color = #3] (#1,#2 -1) -- (#1,#2 +1);
}

\newcommand{\leftmoover}[4]{
    \tikzstyle{roundnode}=[circle,draw = #3,fill=white, minimum size = 56];
    \node[roundnode] (#4) at (#1,#2) {};
    \draw[color = #3] (#1,#2 -1) -- (#1,#2 +1);
    \begin{scope}
        \clip (#1,#2-1) rectangle (#1-1,#2+1);
        \draw[color = #3][fill = #3] (#1,#2) circle(1);
    \end{scope}
}

\newcommand{\rightmoover}[4]{
    \tikzstyle{roundnode}=[circle,draw = #3,fill=white, minimum size = 56];
    \node[roundnode] (#4) at (#1,#2) {};
    \draw[color = #3] (#1,#2 -1) -- (#1,#2 +1);
    \begin{scope}
        \clip (#1,#2-1) rectangle (#1+1,#2+1);
        \draw[color = #3][fill = #3] (#1,#2) circle(1);
    \end{scope}
}

\newcommand{\rightleft}[4]{
    \tikzstyle{roundnode}=[circle,draw = #3,fill=#3, minimum size = 56];
    \node[roundnode] (#4) at (#1,#2) {};
}

\newcommand{\redcell}[3]{
    \tikzstyle{roundnode}=[circle,draw = myblack,fill=redstate, minimum size = 56];
    \node[roundnode] (#3) at (#1,#2) {};
}

\newcommand{\redcellinternal}[3]{
    \tikzstyle{roundnode}=[circle,draw = internal,fill=redstate, minimum size = 56];
    \node[roundnode] (#3) at (#1,#2) {};
}

\newcommand{\greencellgray}[3]{
    \tikzstyle{roundnode}=[circle,draw = border,fill=greenstategray, minimum size = 56];
    \node[roundnode] (#3) at (#1,#2) {};
}

\newcommand{\greencellinternal}[3]{
    \tikzstyle{roundnode}=[circle,draw = internal,fill=greenstate, minimum size = 56];
    \node[roundnode] (#3) at (#1,#2) {};
}

\newcommand{\edge}[3]{
    \draw[-stealth,color=#3,line width=1.75mm] (#1) to (#2);
}

\newcommand{\vertex}[6]{
    \ifnum #4=0
    {\emptycell{#1}{#2}{#3}{#6}}
    \fi
    \ifnum #4=1
    {\rightmoover{#1}{#2}{#3}{#6}}
    \fi
    \ifnum #4=2
    {\leftmoover{#1}{#2}{#3}{#6}}
    \fi
    \ifnum #4=3
    {\rightleft{#1}{#2}{#3}{#6}}
    \fi
    \draw (#1+1,#2) node[right]{{\Huge \color{#3} #5}};
}

\definecolor{myblack}{RGB}{0,0,0}
\definecolor{border}{RGB}{206,206,206}
\definecolor{port}{RGB}{155,155,155}
\definecolor{darkport}{RGB}{195,195,195}
\definecolor{setBorder}{RGB}{80,227,194}
\definecolor{internal}{RGB}{38,105,185}
\definecolor{greenstate}{RGB}{80,183,100}
\definecolor{greenstategray}{RGB}{158,220,170}
\definecolor{redstate}{RGB}{201,109,76}
\definecolor{very-light-gray}{gray}{0.99}
\definecolor{myb}{RGB}{45, 111, 177}
\DeclareRobustCommand{\sizeFtwo}[0]{\Large} 
\DeclareRobustCommand{\sizeFtwop}[0]{\normalsize} 
\DeclareRobustCommand{\sizeFthree}[0]{\LARGE} 
\DeclareRobustCommand{\sizeFthreep}[0]{\normalsize} 
\definecolor{internalStates}{RGB}{240,240,240}
\DeclareRobustCommand{\sizeFfour}[0]{\LARGE \scalefont{0.9}} 
\DeclareRobustCommand{\sizeFfourp}[0]{\normalsize} 
\DeclareRobustCommand{\sizeFfours}[0]{\LARGE} 
\DeclareRobustCommand{\sizeFsixp}[0]{\LARGE} 
\DeclareRobustCommand{\sizeFsixs}[0]{\Huge} 
\DeclareRobustCommand{\sizeFsix}[0]{\Huge \scalefont{2}}

\DeclareRobustCommand{\sizehalfCA}[0]{\LARGE \scalefont{1.1}} 
\DeclareRobustCommand{\sizehalfCAb}[0]{\Huge \scalefont{1.5}} 
\DeclareRobustCommand{\sizebighalfCA}[0]{\LARGE} 
\DeclareRobustCommand{\sizebighalfCAb}[0]{\LARGE} 
\DeclareRobustCommand{\sizesimu}[0]{\large} 
\DeclareRobustCommand{\sizesimub}[0]{\LARGE} 
\DeclareRobustCommand{\sizesimup}[0]{\large} 
\DeclareRobustCommand{\sizebigsimu}[0]{\Huge \scalefont{0.95}} 
\DeclareRobustCommand{\sizebigsimub}[0]{\Huge \scalefont{1.2}} 

\definecolor{myg}{RGB}{56, 140, 70}
\definecolor{myb}{RGB}{45, 111, 177}
\definecolor{myr}{RGB}{199, 68, 64}
\definecolor{mytheorembg}{HTML}{F2F2F9}
\definecolor{mypropbg}{HTML}{F6F6FA}
\definecolor{mylemmabg}{HTML}{FDFDFE}
\definecolor{mytheoremfr}{HTML}{00007B}
\definecolor{mypropfr}{HTML}{00297B}
\definecolor{mylemmafr}{HTML}{00487B}
\definecolor{myexamplebg}{HTML}{E2EBE1}
\definecolor{myexamplefr}{HTML}{164202}
\definecolor{myexampleti}{HTML}{2A7F7F}
\definecolor{mydefinitbg}{HTML}{E5E5FF}
\definecolor{mydefinitfr}{HTML}{3F3FA3}
\definecolor{notesgreen}{RGB}{0,162,0}
\definecolor{myp}{RGB}{197, 92, 212}
\definecolor{mygr}{HTML}{2C3338}
\definecolor{myred}{RGB}{127,0,0}
\definecolor{myyellow}{RGB}{169,121,69}
\definecolor{OrangeRed}{HTML}{ED135A}
\definecolor{Dandelion}{HTML}{FDBC42}
\definecolor{light-gray}{gray}{0.95}
\definecolor{very-light-gray}{gray}{0.99}
\definecolor{Emerald}{HTML}{00A99D}
\definecolor{RoyalBlue}{HTML}{0071BC}

\newtcbtheorem[number within=section]{Theorem}{Theorem}
{%
	enhanced,
	breakable,
	colback = mytheorembg,
	frame hidden,
	boxrule = 0sp,
	borderline west = {2.5pt}{0pt}{mytheoremfr},
	sharp corners,
	detach title,
	before upper = \tcbtitle\par\smallskip,
	coltitle = mytheoremfr,
	fonttitle = \bfseries\sffamily,
	description font = \mdseries,
	separator sign none,
	segmentation style={solid, mytheoremfr},
}
{th}

\newtcolorbox{Theoremcon}
{%
	enhanced
	,breakable
	,colback = mytheorembg
	,frame hidden
	,boxrule = 0sp
	,borderline west = {2.5pt}{0pt}{mytheoremfr}
	,sharp corners
	,description font = \mdseries
	,separator sign none
}

\newtcbtheorem[number within=section]{corollary}{Corollary}
{%
	enhanced,
	breakable,
	colback = mytheorembg,
	frame hidden,
	boxrule = 0sp,
	borderline west = {2.5pt}{0pt}{mytheoremfr},
	sharp corners,
	detach title,
	before upper = \tcbtitle\par\smallskip,
	coltitle = mytheoremfr,
	fonttitle = \bfseries\sffamily,
	description font = \mdseries,
	separator sign none,
	segmentation style={solid, mytheoremfr},
}
{cor}

\newcommand{\descri}[1]{\mdseries\textbf{(#1)}}
\newtcbtheorem[number within=section]{lemma}{Lemma}
{%
	enhanced,
	breakable,
	colback = mypropbg,
	frame hidden,
	boxrule = 0sp,
	borderline west = {1.5pt}{0pt}{mylemmafr},
	sharp corners,
	attach title to upper,
	before upper = \tcbtitle\par\smallskip,
	coltitle = mylemmafr,
	fonttitle = \bfseries\sffamily,
    description formatter=\descri,
	separator sign none,
    terminator sign=., 
	segmentation style={solid, mylemmafr},
}
{lem}

\newtcbtheorem[number within=section]{proposition}{Proposition}
{%
	enhanced,
	breakable,
	colback = mypropbg,
	frame hidden,
	boxrule = 0sp,
	borderline west = {2pt}{0pt}{mypropfr},
	sharp corners,
	attach title to upper,
	before upper = \tcbtitle\par\smallskip,
	coltitle = mypropfr,
	fonttitle = \bfseries\sffamily,
    description formatter=\descri,
	separator sign none,
    terminator sign=., 
	segmentation style={solid, mypropfr},
}
{prop}

\newtcbtheorem[number within=section]{remark}{Remark}
{%
	enhanced,
	breakable,
	colback = very-light-gray,
	frame hidden,
	boxrule = 0sp,
	borderline west = {1.5pt}{0pt}{myb},
	sharp corners,
	attach title to upper,
	before upper = \tcbtitle\par\smallskip,
	coltitle = black,
	fonttitle = \bfseries\sffamily,
    description formatter=\descri,
	separator sign none,
    terminator sign=., 
	segmentation style={solid, mytheoremfr},
}
{rem}

\newtcbtheorem[number within=section,use counter=mdexample]{example}{Example}
{%
    colback=green!5,
    colframe=green!35!black,
    fonttitle=\bfseries
}
{ex}

\newtheorem{definition}{Definition}

\makeatletter
\let\orgdescriptionlabel\descriptionlabel
\renewcommand*{\descriptionlabel}[1]{%
  \let\orglabel\label
  \let\label\@gobble
  \phantomsection
  \edef\@currentlabel{#1}%
  \let\label\orglabel
  \orgdescriptionlabel{#1}%
}
\makeatother

\newcommand{\RX}[2]{}
\newcommand{\MC}[1]{{\color{lowgreen}#1}}
\renewcommand{\MC}[1]{{#1}}

\newcommand{\RD}[1]{}

\newcommand{\reduced}[2]{#2}

\newcommand{\Z}{\mathbb{Z}}
\newcommand{\Past}{\mathrm{Past}}

\newcommand{\closeddomain}{{stability}}
\newcommand{\I}{\mathrm{I}}
\newcommand{\B}{\mathrm{B}}
\newcommand{\V}{\mathrm{V}}

\newcommand{\E}{\mathrm{E}}

\newcommand{\pof}{\pi.}
\renewcommand{\P}{\pof\E}

\newcommand{\restr}[2]{\mathcal{N}_{#1}(#2)}
\newcommand{\restri}[1]{{\mathcal{N}_{#1}}}
\newcommand{\crestr}[2]{{{\overline{\mathcal{N}}}_{#1}(#2)}}
\newcommand{\crestri}[1]{{{\overline{\mathcal{N}}}_{#1}}}

\renewcommand{\sp}[1]{\lfloor #1\rfloor}

\newcommand{\namesubseteq}{\,\hat{\subseteq}\,}

\newcommand{\restriA}[1]{{\mathcal{N}_{#1}}}

\newcommand{\valid}[2]{\Omega_{#1}}
\newcommand{\dangling}[1]{{\E_{#1}\setminus (\I_{#1}\!:\!\pi)^2}}
\newcommand{\graphvalidexplicit}[2]{\Gamma_{{#2}}}
\newcommand{\graphvalid}[2]{\graphvalidexplicit{#1}{#2}}

\newcommand{\upast}[1]{{\cal P}_{#1}}

\renewcommand{\sp}[1]{\lfloor #1\rfloor}

\newcommand{\sshift}{\mathcal{S}}

\newcommand{\VL}[1]{}
\newcommand{\withproofs}[1]{}

\newcommand{\calG}{{\ensuremath{\mathcal{G}}}}

\newcommand{\calV}{{\ensuremath{\mathcal{V}}}}

\newcommand{\calX}{{\ensuremath{\mathcal{X}}}}

\newcommand{\X}{\calX}

\newcommand{\complete}[0]{\calG_c}
\newcommand{\Scomplete}[0]{\sshift_c}
\newcommand{\maxi}[0]{^{\infty}}
\newcommand{\maxgraph}[0]{maximal }

\newcommand{\xCone}[0]{\mathcal{C}_x}
\newcommand{\xcone}[1]{\mathcal{C}_{#1}}
\newcommand{\diskA}[1]{\sshift_{\restriA{#1}}}

\definecolor{lowgreen}{rgb}{0.40390625,0.6109375,0.09109375}

\begin{document}

\title{Causal graph rewriting}

\author{Pablo Arrighi\inst{1}\email{pablo.arrighi@universite-paris-saclay.fr} \and
Marin Costes\inst{2}\email{marin.costes@ens-paris-saclay.fr} \and
Luidnel Maignan\inst{3}\email{luidnel.maignan@u-pec.fr}
}

\institute{Universit\'e Paris-Saclay, Inria, CNRS, LMF, 91190 Gif-sur-Yvette, France
\and
Centre for Quantum Information and Communication, \'Ecole polytechnique de Bruxelles, CP 165/59, Universit\'e libre de Bruxelles, 1050 Brussels, Belgium
\and
Univ Paris Est Creteil, LACL, 94000 Creteil, France
}

\def\received{}

\maketitle

\begin{abstract}
We introduce causal graph rewriting, a model of computation in which local rules are applied on directed acyclic graphs in an asynchronous manner. The non-determinism arising from asynchrony is disciplined by the oriented edges, which must be understood as both computational dependencies and locality constraints---and are themselves subject to the rewriting. We illustrate the model through two examples: a particle system, and a time-dilation example---reminiscent of general relativity. We study the well-definedness and properties of induced subgraphs and graph composition, which isolate and recombine the region affected by a rewrite. We then study locality with respect to these constructions, showing how a local rewrite preserves positions, borders, and context. Our main result concerns sequential composition: locality extends from single rule applications to arbitrary valid sequences, as any local rule is automatically $*$-local and $*$-extensive. We also formalise and prove the simulation of any one-dimensional cellular automaton.
\end{abstract}


\section{Introduction}


\paragraph*{Dynamical systems, synchronism, asynchronism}

{\em A dynamical system} describes the iterative evolution of a configuration, with a global update mapping the configuration at time $t$ to the configuration at time $t+1$.
When configurations have a spatial structure, this global update is synchronous: all sites advance together. Such configurations are often grid-based (e.g. for representing particles \cite{Zuse,Wolf-Gladrow}, fluids \cite{ChopardDroz,RothmanZaleski}, traffic jams \cite{NagelSchreckenberg}, demographics and regional development or consumption \cite{Bruun,WhiteEngelen}). But they can also be graph-based (e.g. for representing physical systems \cite{MeyerLove}, computer processes \cite{PapazianRemila}, biochemical agents \cite{MurrayDicksonVol2}, economical agents \cite{KozmaBarrat}, users of social networks, etc.).
One can think of cellular automata \cite{Neumann,Langton}, lattice-gas automata \cite{RothmanZaleski, Wolf-Gladrow}, parallel graph rewriting \cite{Ehrig}, causal graph dynamics~\cite{ArrighiIC,ArrighiRCGD,ArrighiBRCGD} or the so-called global transformations~\cite{DBLP:conf/gg/MaignanS15,DBLP:conf/mfcs/FernandezMS22,DBLP:conf/gg/MaignanS24} for instance.

{\em Synchronism} is often criticised however, on the basis that: 1/ In some contexts synchronisation mechanisms are considered a costly resource, which hinders the use of parallelism e.g. to achieve high performance computing~\cite{AsynchronousSimsSync}. 2/ It is often dubbed as introducing too much artefacts as far as physical simulation is concerned. Some authors argue that nature has no central clock, and thus cannot apply the same local rule everywhere at once~\cite{AsynchronousPhysical}. This is a fair point as relativistic physics clearly departs from the idea of a global time across the universe. In particular the `time covariance' symmetry entails that, in any real system, it is legitimate to evolve just a small region of space, whilst keeping the rest of it unchanged---in some precise sense.

{\em Asynchronism}---the application of the local rule at totally arbitrary places, non-deterministically---fits this picture better
and is also well studied, especially when it is more compelling to leave the evaluation strategy under-determined.
This is typically the case in rewriting theory, when the rewrite rules arise as a computationally-oriented version of an equality, e.g. $1+1\to 2$. Then, term $1+1+1$ may evolve into $2+1$ or $1+2$, non-deterministically. This feels right, because: 1/ an underlying symmetry tells us there is no reason to favour one over the other and 2/ ultimately, if what matters is the end result, we are reassured by the fact that $2+1\to 3$ and $1+2\to 3$.

\paragraph*{From cellular automata to causal graph rewriting.}

The models recalled above---cellular automata~\cite{Hedlund} with their static grid-based space and the more general causal graph dynamics~\cite{ArrighiCGD,ArrighiCayley} with their dynamical graph-based space---remain synchronous: a global clock updates every site at once. The next step along this line is to drop the global clock while still obtaining a well-determined space-time. That step is \emph{causal graph rewriting} (CGR): the asynchronous counterpart of causal graph dynamics, in which local rules act on directed acyclic graphs whose oriented edges encode both computational dependencies and locality constraints---and are themselves rewritten.

This asynchronous model, together with a space-time determinism theorem, was first presented in a short conference paper~\cite{ArrighiSpaceTimeDetWRLA}. There, local conditions on the rewrite rule are shown to suffice so that asynchronous applications still produce well-determined events in the unfolding of the graph, independently of the scheduling. The present article completes and clarifies that account, using proofs from the thesis~\cite{CostesPhD}: it develops the motivations and design choices of the model, and supplies the detailed proofs of the structural properties on which space-time determinism rests.

The mechanism that makes the step possible is inspired by the block decomposition of reversible and quantum cellular automata~\cite{ArrighiOverview,FarrellyReview,Durand-LoseBlock,ArrighiBLOCKREP}. Successive local updates become vertices of the DAG: each vertex is an event---an application of the local rule at some position and time tag---and may fire only once its incoming dependencies have been resolved. Scheduling is therefore constrained by a causal order, rather than by a global clock. Figure~\ref{fig:causal dynamical systems} locates this construction in the classical-to-reversible-to-quantum trajectory already followed by cellular automata and causal graph dynamics.

\begin{figure}
    \centering
    \begin{tabular}{ |c|c|c|c| } 
        \hline
        & \textbf{\small Cellular automata} & \textbf{\small Graph dynamics} & \textbf{\small Graph rewriting}\\
        \hline
        \textbf{Quantum} & QCA~\cite{ArrighiOverview,FarrellyReview} & QCGD~\cite{ArrighiQCGD,ArrighiQNT} & {QCGR?} \\
        \hline
        \textbf{Reversible} & RCA~\cite{KariBlock,Durand-LoseBlock} & RCGD~\cite{ArrighiRCGD} & RCGR~\cite{CostesPhD} \\
        \hline
        \textbf{Classical} & CA & CGD~\cite{ArrighiCGD} & \textcolor{internal}{CGR}~\cite{ArrighiSpaceTimeDetWRLA,CostesPhD}\\ 
        \hline
    \end{tabular}
    \caption{\emph{Causal dynamical systems.} Cellular automata are grid-based synchronous dynamical systems. Causal graph dynamics generalise this to graphs, allowing for dynamical geometry. Causal graph rewriting is their asynchronous counterpart: it was introduced in~\cite{ArrighiSpaceTimeDetWRLA}, developed in the thesis~\cite{CostesPhD}, and is the subject of the present article (highlighted in dark blue). Extending this line to the quantum case (QCGR) is left as future work.}
    \label{fig:causal dynamical systems}
\end{figure}


\paragraph*{Contributions.}
Relative to the short conference paper~\cite{ArrighiSpaceTimeDetWRLA}, this article contributes the following.

\emph{(1) Structural properties.} The formal definition of causal graph rewriting---port-labelled directed acyclic graphs, neighbourhood schemes, and local rules---was already given in~\cite{ArrighiSpaceTimeDetWRLA}. Here we focus on illustrating \emph{why} this precise definition is needed: so that the basic operations on graphs, taking induced subgraphs and composing them via $\sqcup$, remain well-behaved, and in turn yield the structural properties of local rules (position preservation, border preservation, locality) on which the rest of the development relies.

\emph{(2) Examples and expressivity.} We develop in detail the constructions only sketched in~\cite{ArrighiSpaceTimeDetWRLA}. The particle-system and time-dilation examples are given precise local rules and space-time diagrams, illustrating respectively asynchronous simulation and genuinely asynchronous behaviour. We moreover turn the cellular-automaton simulation, based on a `marching soldiers' encoding~\cite{WeakConsistencyGacs,MarchingNehaniv}, into a formal result, proving that any one-dimensional cellular automaton can be simulated within the framework (Th.~\ref{th:Simulation of synchronous cellular automata}).

\emph{(3) Robustness.} We prove that locality is preserved under sequential composition of rule applications: any local rule is automatically $*$-local and $*$-extensive, meaning that the neighbourhood and locality properties that hold for a single rule application extend canonically to arbitrary valid sequences. This is the main technical result of the paper (Th.~\ref{th: star ext and star loc}). Preserving these structural properties through composition is a prerequisite of any generic argument about the model that involves sequential rule application. This closure property is neither stated nor proved in~\cite{ArrighiSpaceTimeDetWRLA}, although it is needed in the proof of its space-time determinism theorem.



\paragraph*{Related work.}

{\em Graph rewriting.} Geometry is dynamical in our work. We thus hope it makes a useful addition to the already wide literature on Graph Rewriting~\cite{RozenbergBook,EhrigBook}.
We are aware that the dominating vocabulary to describe them is now that of Category theory~\cite{LoweAlgebraic,Taentzer,HarmerFundamentals}, in which ways of combining non-commuting rules~\cite{EchahedCombiningRules} and notions of space-time diagrams have been developed~\cite{UnfoldingAdhesivePawel,BehrStochastic,UnfoldingKonig,UnfoldingReckel}.
We instead use the vocabulary of dynamical systems, as we came to consider Graph Rewriting through a series of works generalising CA to synchronous, causal graph dynamics~\cite{ArrighiCayleyNesme}, and tilings to graph subshifts~\cite{ArrighiSubshifts}. We are confident that abstracting away the essential features of our formalism could yield interesting categorical frameworks, e.g. à la~\cite{Maignan}.

{\em Dynamical networks.} The closest works however turn out to come from varied communities. In algorithmic complexity, \cite{AsynchronousSimsSync} uses a DAG of dependencies representation to reduce the synchronisation costs of simulating a class of synchronous algorithms---we use it in the more dynamical systems context and in order to relax synchronism altogether, whilst aiming at space-time determinism. In computational Physics, \cite{GorardWofram} promotes the lattice of dependencies of local rule applications to a notion of space-time, and advocates a notion of `causal invariance' based on the unicity of this lattice, as formalised in the context of string rewriting~\cite{CausalInvarianceOnStrings}---we identify local conditions expected to achieve it. In the network reliability community, \cite{WeakConsistencyGacs} obtains a result on asynchronous simulation of cellular automata that is related to our setting---our local rules are allowed to modify the neighbouring vertices and the graph itself.


{\em Relativistic Physics.} Dropping the global clock is not only a natural computer-science extension of causal graph dynamics: it is needed to model phenomena that are fundamentally asynchronous. The DAG of dependencies is both a constraint upon the evolution and a subject of the evolution. Rewriting the DAG itself yields effects such as time dilation (Fig.~\ref{fig : Space time diagram example 2}), which have no synchronous counterpart. This is in line with seeking mathematically sound, constructive frameworks for discrete models of general relativity~\cite{Sorkin}. In General Relativity the primary object is 4D space-time, which can still be cut into successive snapshots (`space-like cuts'); the slicing may be irregular, so that one region of a slice may have evolved further than another---as modelled by asynchronism.


\paragraph*{Plan.}
Sec.~\ref{sec : preliminaries} gently introduces the main ideas through examples, motivating the use of DAGs as causal structure.
Sec.~\ref{sec:graphs} defines the graph formalism and Sec.~\ref{sec: local rules} defines neighbourhood schemes and local rules, establishing the structural properties of the model. Sec.~\ref{sec: sequences and ST diag} defines space-time diagrams.
Sec.~\ref{sec: Simulation} shows how to simulate any one-dimensional cellular automaton via a marching-soldiers encoding.
Sec.~\ref{sec:exV2} presents the time dilation example, illustrating fundamentally asynchronous behaviour.
Sec.~\ref{sec: locality for sequences} proves the main robustness result: $*$-locality and $*$-extensivity hold for all local rules.

\section{Background and Motivation}\label{sec : preliminaries}

\todo[inline]{L: Ma première sensation à la lecture de ce préliminaire est que pour un article journal, il faut être plus synthétique en faisant d'abord tout le rappel formel sur les automates cellulaires (avec le vocabulaire classique ?), puis introduire l'exemple}

\subsection{Preliminaries}
%
%

The space on which a cellular automaton acts is divided into cells. Each of these cells contains an \emph{internal state}, picked in a finite set $\Sigma$. We begin by introducing the internal states of the running example, which will be used throughout this section.

A \emph{configuration} of a cellular automaton is an $n$-dimensional grid of cells, each containing an internal state. We call $\calX = \mathbb{Z}^n$ the set of all different positions on the grid. A configuration $\sigma:\calX \to \Sigma$ is a function which maps each position to its internal state. 

To match the terminology used later, we call \emph{neighbourhood scheme} a function $\restri{}$ which, given a position $x$, selects a set of nearby vertices $\restri{x}$. We then call a neighbourhood the restriction of $\sigma$ to $\restri{x}$ and we denote it $\sigma_{\restri{x}}:\restri{x}\to \Sigma$.

%

\begin{example}{Particle system---Configuration, Neighbourhood}{}
    \begin{minipage}{\linewidth}
        We fix $\Sigma= \{0,1\}^2$. Here internal states are pairs of bits representing the presence of a left-moving particle or not, and the presence of a right-moving particle or not. 
        Consider a grid of dimension $n=1$---i.e. $\calX = \mathbb{Z}$. Each configuration $\sigma$ is an element of ${(\{0,1\}^2)}^{\mathbb{Z}}$. We then pick the radius $1$ neighbourhood scheme, namely $\restri{x} = \{x-1,x,x+1\}$. On the right we illustrate a neighbourhood $\sigma_{\restri{2}}$. It maps three positions ($1$, $2$ and $3$) to their internal states (for example: $\sigma_{\restri{2}}:1\mapsto (0,1)$).
    \end{minipage}
    \begin{minipage}{\linewidth}
        \centering
        \begin{tabular}{c}
            \begin{minipage}{0.65\linewidth}
                \includegraphics[width=\linewidth]{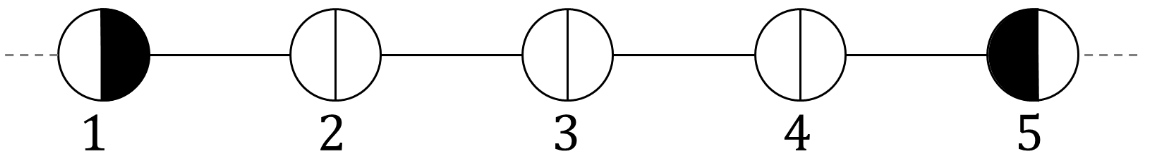}
                \captionof*{figure}{A configuration $\sigma$.} 
            \end{minipage}
            \\
            \begin{minipage}{0.65\linewidth}
                \includegraphics[width=\linewidth]{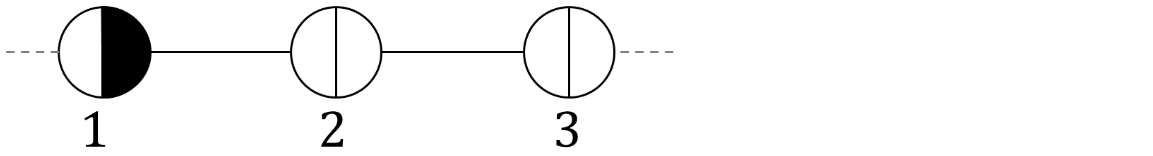}
                \captionof*{figure}{A neighbourhood $\sigma_{\restri{2}}$.} 
            \end{minipage}
        \end{tabular}
    \end{minipage}
\end{example}

A \emph{cellular automaton} is a global function $F$ that maps a configuration $\sigma^t$ at time $t$ to the configuration $\sigma^{t+1}$ at time $t+1$. This global function must be local and homogeneous---i.e. there exists a unique local update function, which we call the \emph{on-site function} $f$, mapping any neighbourhood $\sigma^t_{\restri{x}}$ to the next state $\sigma^{t+1}_x$ at position $x$. Consequently, this global function is shift-equivariant: shifting the input shifts the output by the same amount, or equivalently, $F$ commutes with shifts. 

\begin{example}{Particle system---Local rule}{PS-CA---LR}
    \begin{minipage}{\linewidth}
        The local rule transports left-moving particles to the left, and right-moving particles to the right, which means consuming the particle at a cell, to give it to another. Thus given a configuration $\sigma^t$ the on-site function is given by 
        $f:\sigma^t_{\restri{x}}\mapsto \textit{left}(\sigma_{x+1}^t),\textit{right}(\sigma_{x-1}^t)$, 
        where for any pair of bits $(i,j)$, $\textit{left}(i,j)=i$ and $\textit{right}(i,j)=j$.
    \end{minipage}
    \begin{minipage}{\linewidth}
        \centering
        \begin{tabular}{cc}
            \begin{minipage}{0.25\linewidth}
                \includegraphics[width=\linewidth]{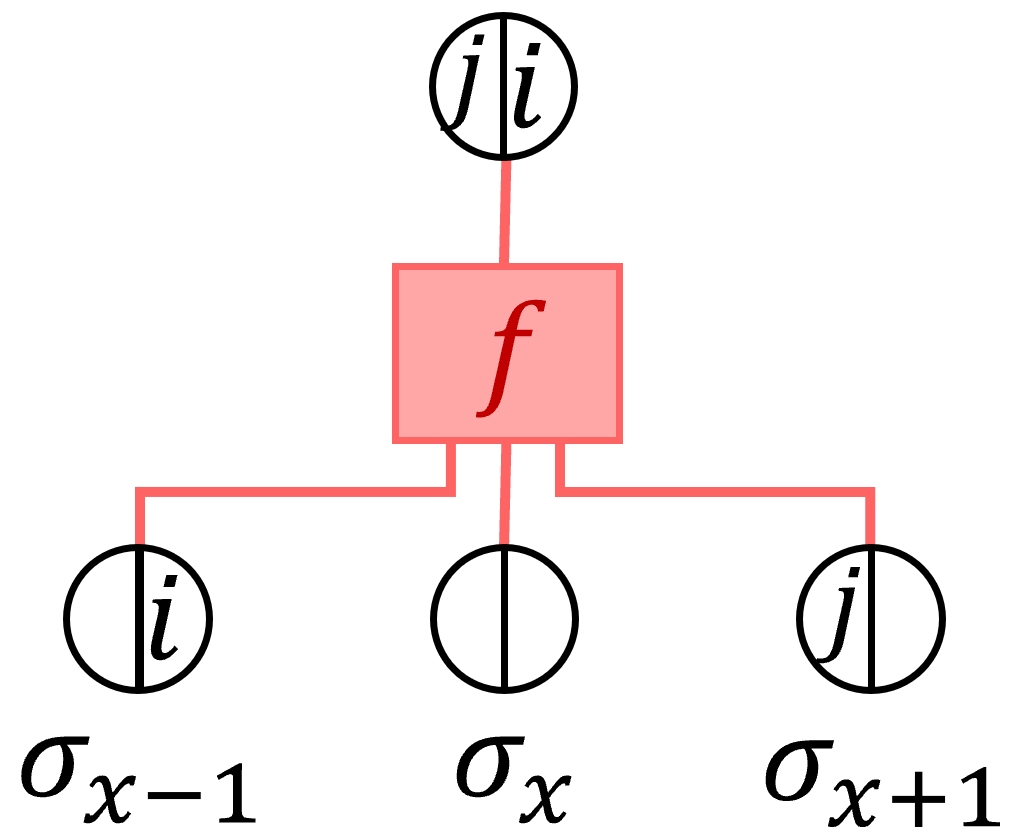}
                \captionof*{figure}{$f$ computes a state.}  
            \end{minipage}
            &
            \begin{minipage}{0.45\linewidth}
                \includegraphics[width=\linewidth]{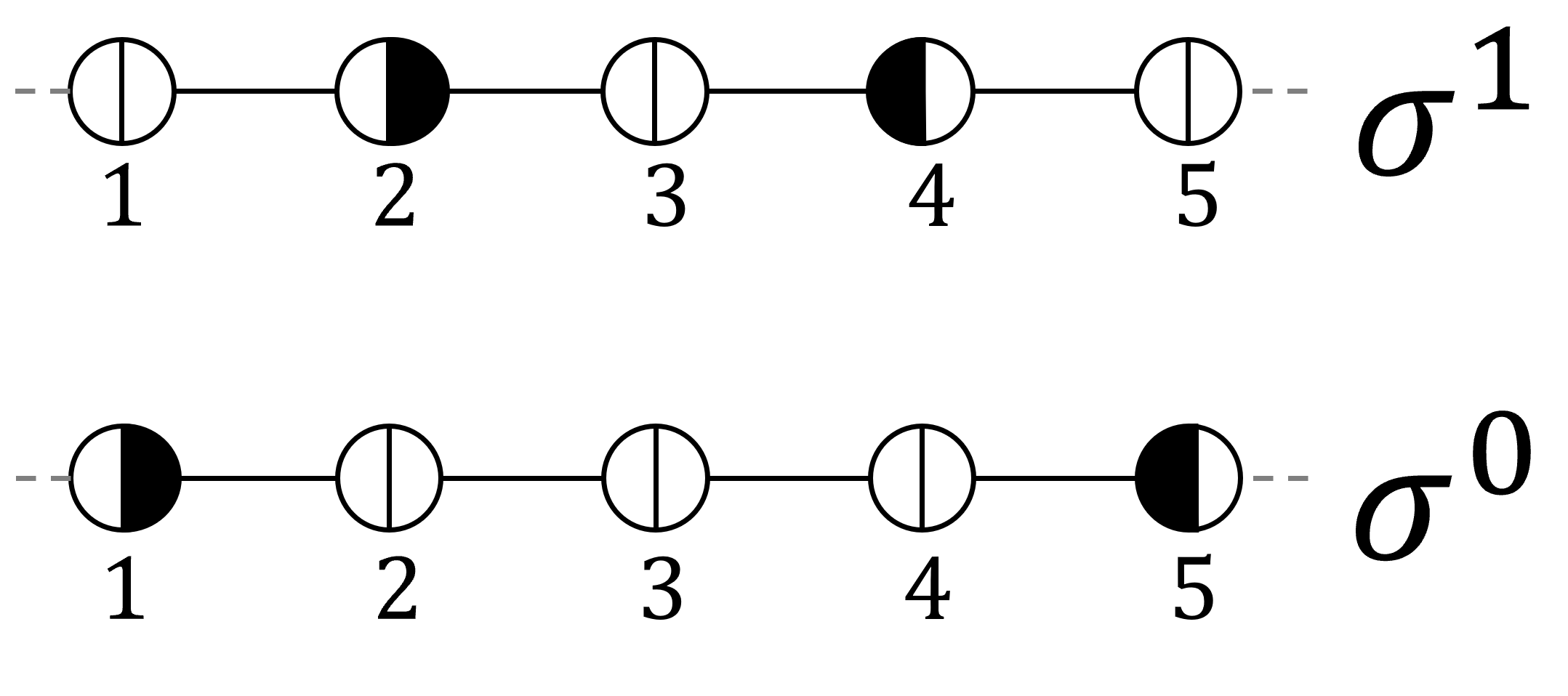}
                \captionof*{figure}{$A$ maps $\sigma^0$ to $\sigma^1$.} 
            \end{minipage}
        \end{tabular}
    \end{minipage}
\end{example}

Given a cellular automaton $F$ and an initial configuration $\sigma^0$, we can compute its temporal evolution as a sequence of configurations. We call it the \emph{space-time diagram} of $F$ and $\sigma^0$ and we denote it $\mathcal{M}_{\sigma^0} = \{\sigma^t\}_{t\in\mathbb{N}}$ where $\sigma^{t+1} = F\sigma^{t}$. Note that, by locality, each state is computed using the on-site function---i.e. $\sigma^{t+1}_x = f\sigma^{t}_{\restri{x}}$.

\begin{figure}[h]
    \centering
    \includegraphics[width=0.40\linewidth]{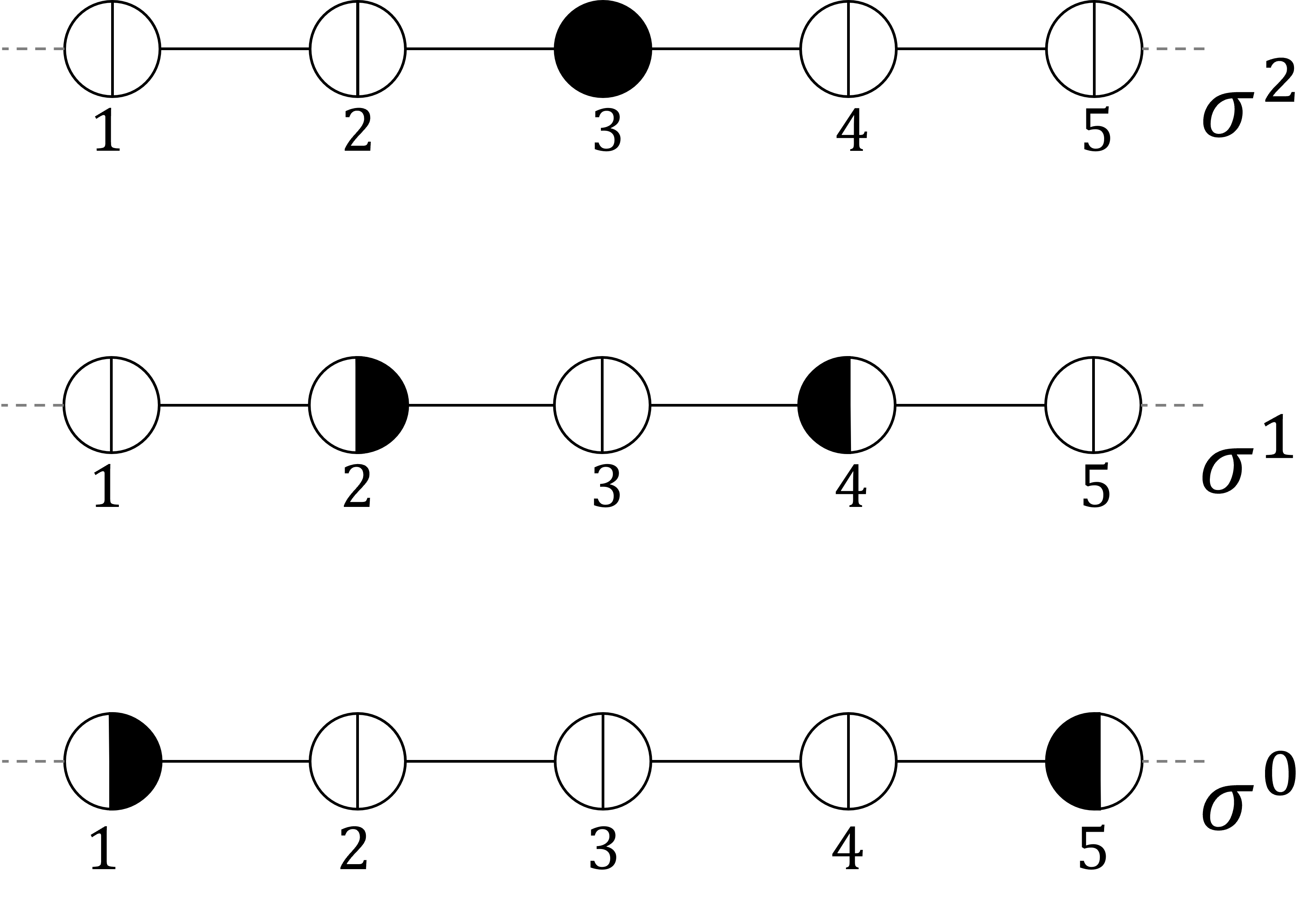}
    \caption{The space-time diagram of the particle system example.}
\end{figure}

\MC{We now turn to a formulation known as block decomposition (see Ex.~\ref{ex:Bloc decomposition}) more aligned with the objectives of this work---one that is both closer to quantum cellular automata (QCA)~\cite{ArrighiOverview,FarrellyReview} and more amenable to incorporating asynchrony.\\
The structure of QCA and reversible cellular automata (RCA), their classical counterparts, is best understood as an infinite circuit composed of local quantum gates. Indeed, it has been established in \cite{Durand-LoseBlock} that any one-dimensional or two-dimensional RCA can be expressed as a composition of finite reversible gates and partial shifts. In higher dimensions, a block representation remains possible, though it may require extending the alphabet $\Sigma$ of the original automaton \cite{ArrighiBLOCKREP}.}

\begin{example}{Particle system---Block decomposition}{Bloc decomposition}
    \begin{minipage}{\linewidth}
        This example admits a two-layer block decomposition, composed solely of swap gates. A swap gate $S$ takes as input two bits $(i,j)$ and returns the pair $(j,i)$. In the first layer, each swap gate acts between two adjacent cells, taking $\textit{right}(\sigma_x)$ and $\textit{left}(\sigma_{x+1})$ as inputs. In the second layer, we apply them again shifted by one half cell, and we obtain exactly the same as in Ex.~\ref{ex:PS-CA---LR}.
    \end{minipage}
    \begin{minipage}{\linewidth}
        \centering
        \includegraphics[width=0.55\linewidth]{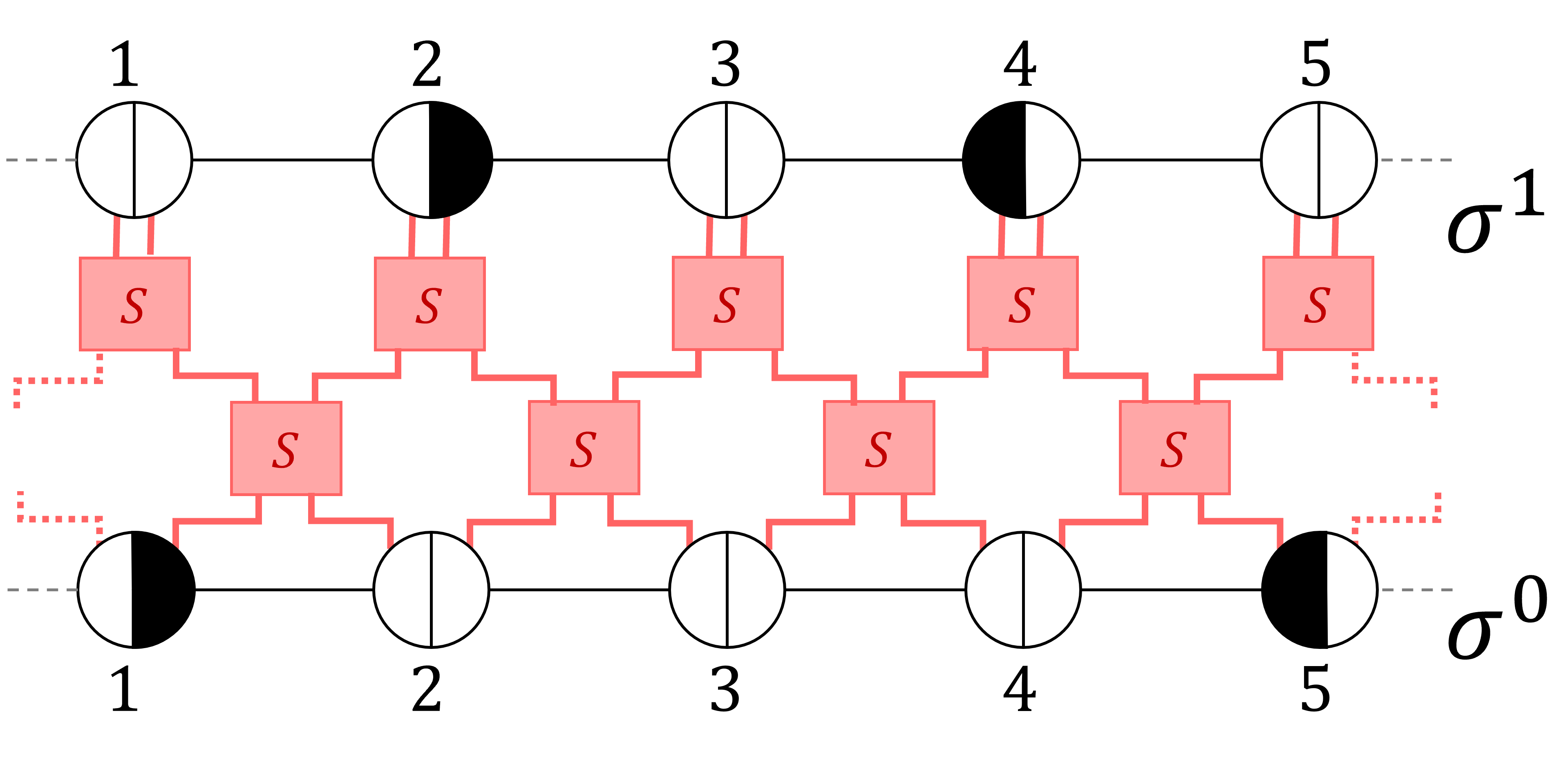}
    \end{minipage}
\end{example}

\subsection{Informal Description of the Model}

\todo{L: Cette sous-section contient clairement du contenu qui n'est pas du préliminaire sur les automates cellulaires ni sur les dynamics asynchrones.}

The block decomposition of RCA suggests an asynchronous simulation: instead of evaluating each layer in full, one may evaluate any subset of gates whose input dependencies have been satisfied. Figure~\ref{fig:space-like cut} contrasts these two evaluation strategies. Panel (a) applies all gates through two complete steps and reaches $\sigma^2$, whereas panel (b) applies only $S_{0.1^+}$, $S_{0.2^+}$, and $S_{1.2}$ and yields a partial configuration $\sigma'$ lying between synchronous cuts.\todo{L: Je modifierai Fig.~\ref{fig:space-like cut}(b) pour y enlever les informations consommées. J'appliquerai également $s_{0,4^+}$} In the following, we refer to the states of these asynchronous, partial computations, as \emph{snapshots}---by analogy to computer science---or as \emph{space-like cuts}, in analogy to physics. In order to properly keep track of which gates have been computed, let us assign names to them. Given a gate $S$, we name each of its applications $S_{t.x}$, where $t$ is a time tag and $x$ a position in space.

\begin{figure}[h]
    \begin{subfigure}{0.49\textwidth}
        \captionsetup{justification=centering}
        \includegraphics[width=\textwidth]{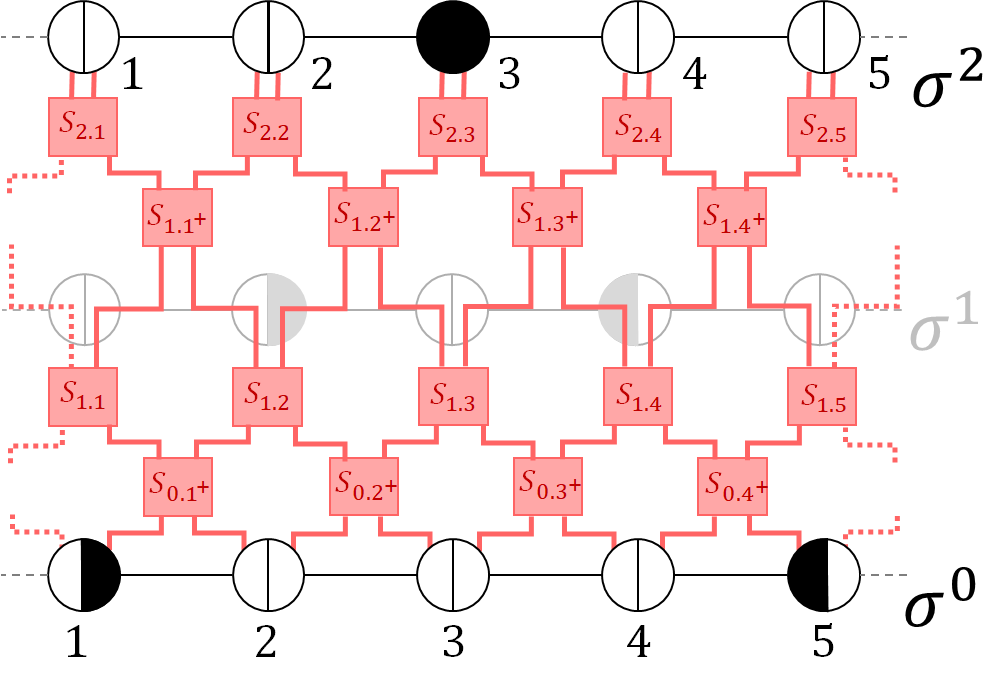}
        \caption{Applying all gates yields configurations produced by $F$.}
    \end{subfigure}
    \hfill
    \begin{subfigure}{0.49\textwidth}
        \captionsetup{justification=centering}
        \includegraphics[width=\textwidth]{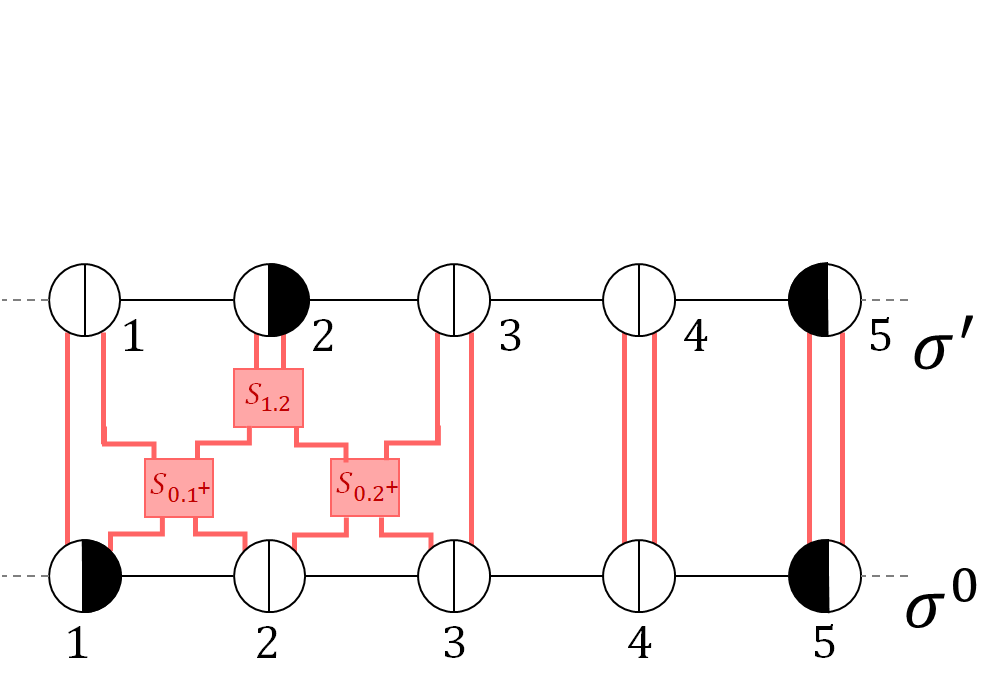}
        \caption{Applying only $S_{0.1^+}$, $S_{0.2^+}$, and $S_{1.2}$ yields the partial configuration $\sigma'$.}
    \end{subfigure}
    \caption{Asynchronous gate computation}
    \label{fig:space-like cut}
\end{figure}

When applying these gates asynchronously, we must keep track not only of the current configuration but also of gate dependencies, represented by wires. We encode both in a labelled directed acyclic port graph. Each vertex $t.x$ corresponds to a gate $S_{t.x}$, its internal state represents the wire inputs, and the edges encode input/output dependencies. Attaching edges to node ports allows each gate to distinguish its inputs. 
We explain how to describe a configuration of our running example in this manner in Ex.~\ref{ex:async_configuration}.

\begin{example}{Asynchronous particle system---Configuration}{async_configuration}
    \begin{minipage}{\linewidth}
        We represent each gate surrounding the configuration $\sigma^1$ by a vertex, and the wires between these gates by directed edges, attached either to the left or to the right port of a vertex. Then we represent the input to the wires as internal states on the bottom part of the graph. 
    \end{minipage}
    \begin{tabular}{c}
        \begin{minipage}{\linewidth}
        \centering
        \includegraphics[width=0.6\linewidth]{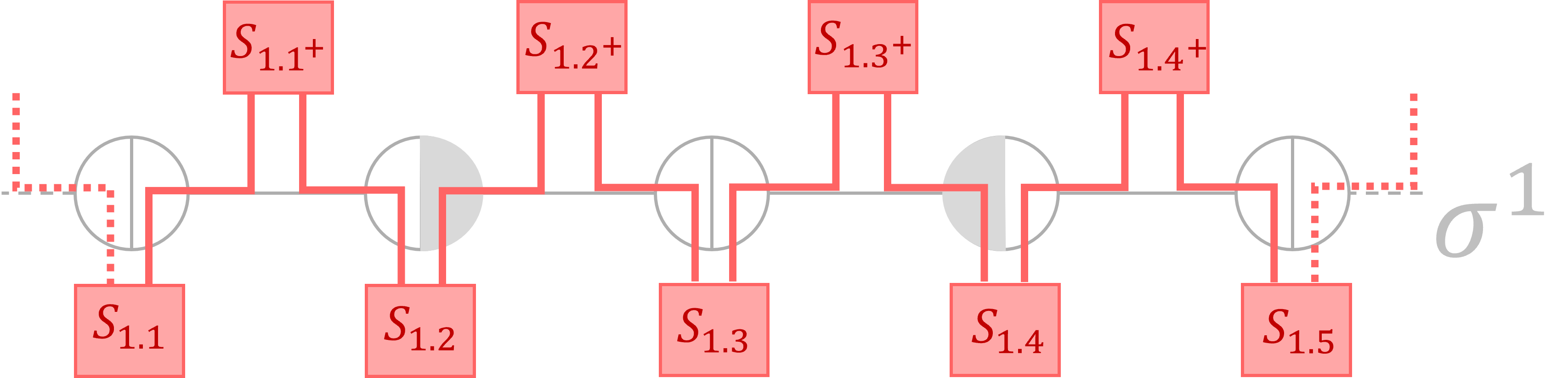}
        \captionof{figure}{The configuration $\sigma^1$ with its surrounding gates\dots}
        \end{minipage}
        \\
        \begin{minipage}{\linewidth}
        \centering
        \includegraphics[width=0.6\linewidth]{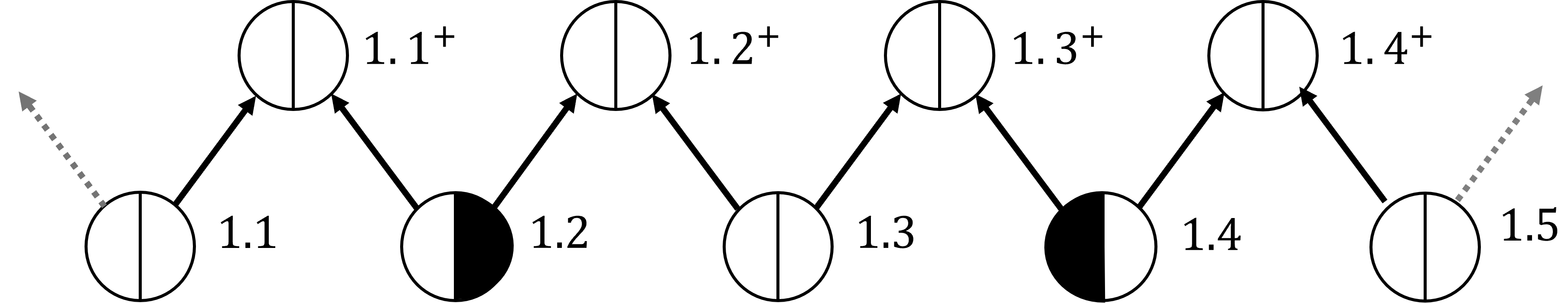}
        \captionof{figure}{\dots can be encoded as $G^1$, a labelled port DAG.}
        \end{minipage}
    \end{tabular}
\end{example}

Now we can phrase the partial evolution of the cellular automaton, as the application of a graph rewriting rule. By homogeneity this rewriting rule has to be the same everywhere in space.
Its application at $x$ is denoted $A_x$.
Given the vertex corresponding to a gate $S_{t.x}$, the rule $A_x$: $1/$ passes its internal states along the dependency edges, possibly applying the neighbouring gates $2/$ creates the vertex corresponding to the next gate at position $x$, namely $S_{(t+1).x}$, with its dependencies $3/$ drops $S_{t.x}$.

\begin{example}{Asynchronous particle system---Local rule}{causal structure}
     \begin{minipage}{0.59\linewidth}
        $A_x$ transforms the bottom pattern into the top one. By doing so any right mover $i$ in $t.x$ moves to the right and any left mover $j$ moves to the left. The vertex $t.x$ disappears from the graph, and a new one is created, called $(t+1).x$.
    \end{minipage}
    \begin{minipage}{0.4\linewidth}
        \centering
        \includegraphics[width=0.75\linewidth]{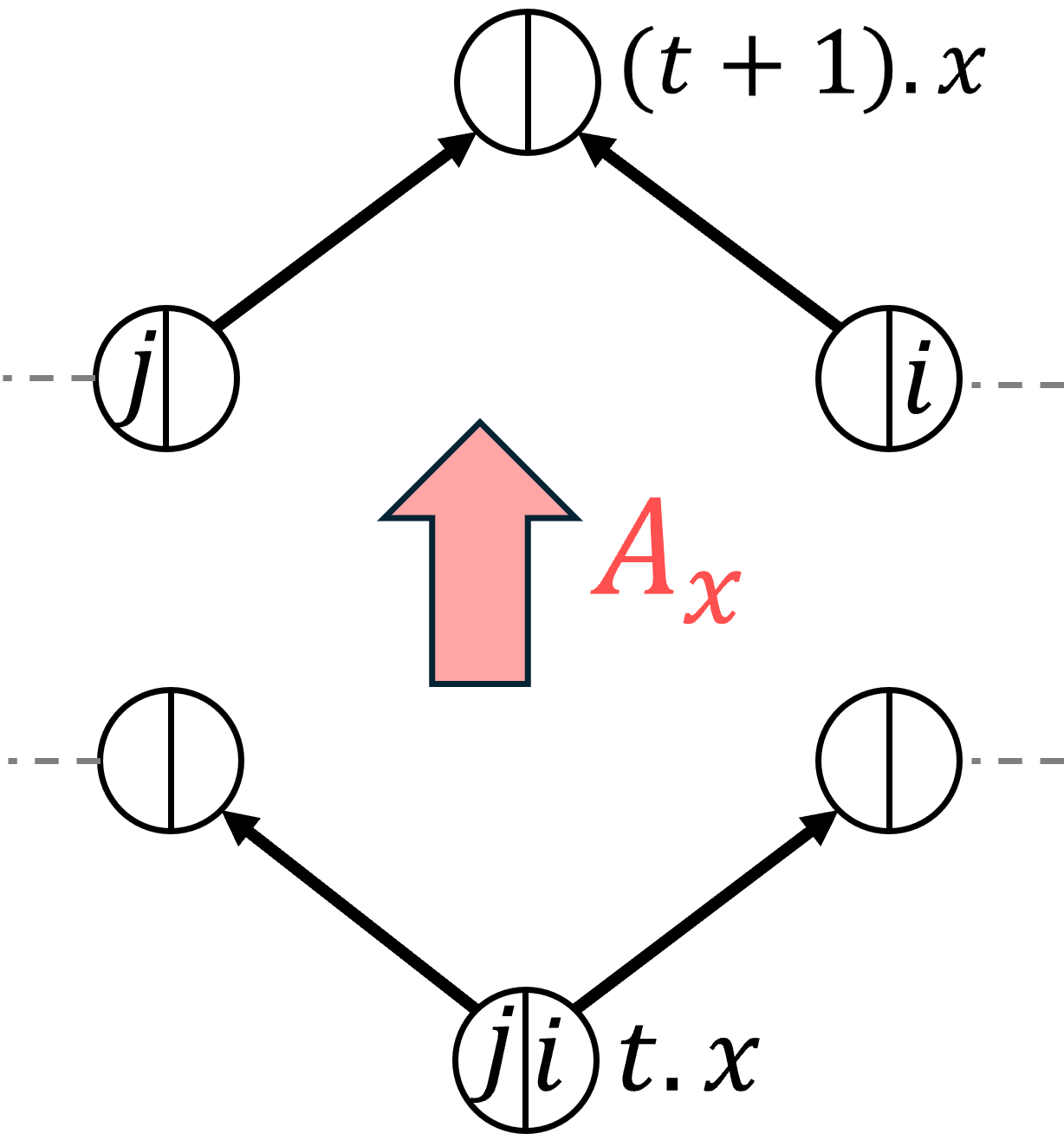}
    \end{minipage}
\end{example}

If the encoding is well-chosen, we can retrieve the states computed by any gate $S_{t.x}$ through consecutive applications of the rewriting pattern, as illustrated in Fig.~\ref{fig: async local rule}. Note that only very specific shapes of graphs correspond to a configuration of the original CA, the other graphs must be thought of as snapshots of the block decomposition.

\begin{figure}[h]
    \centering
    \begin{subfigure}{0.49\textwidth}
        \captionsetup{justification=centering}
        \includegraphics[width=\textwidth]{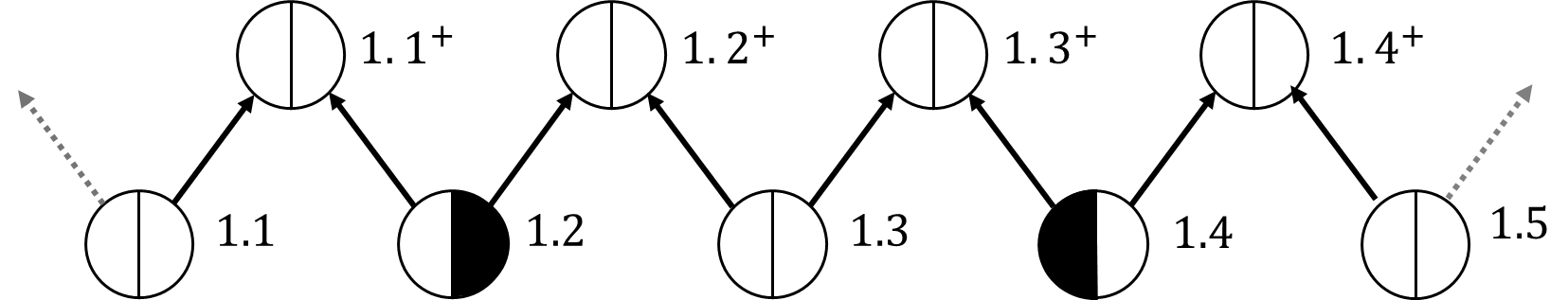}
        \caption{$G^1$.}\label{subfig: G1}
    \end{subfigure}
    \begin{subfigure}{0.49\textwidth}
        \captionsetup{justification=centering}
        \includegraphics[width=\textwidth]{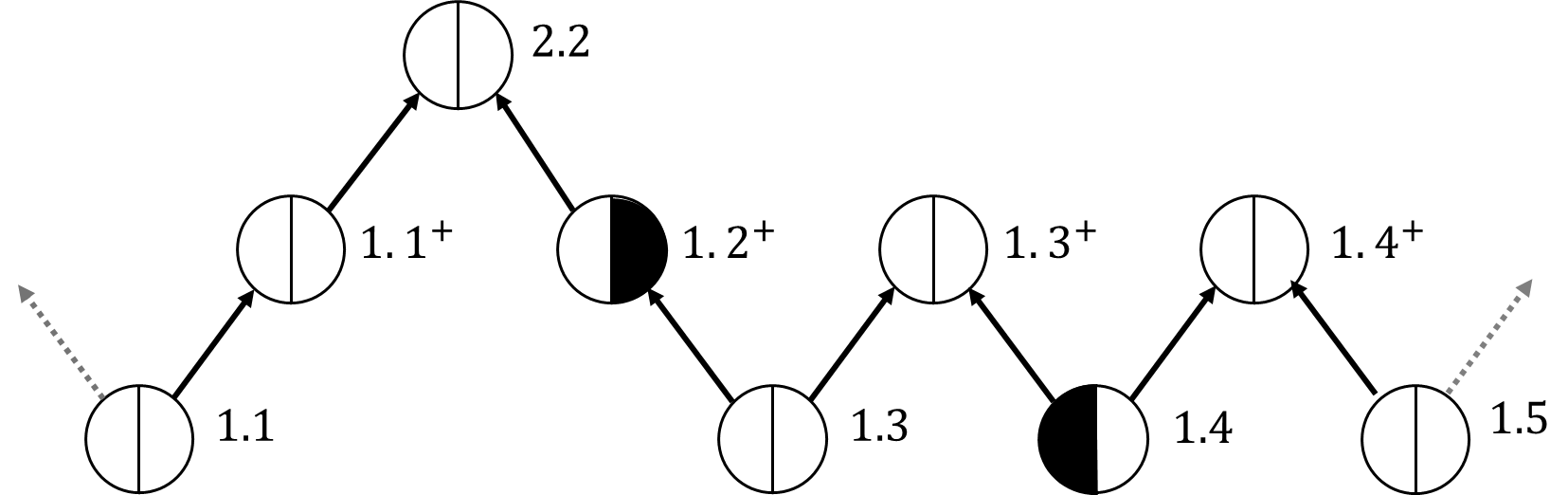}
        \caption{$A_{2}G^1$.}\label{subfig: A2G1}
    \end{subfigure}
    \begin{subfigure}{0.49\textwidth}
        \captionsetup{justification=centering}
        \includegraphics[width=\textwidth]{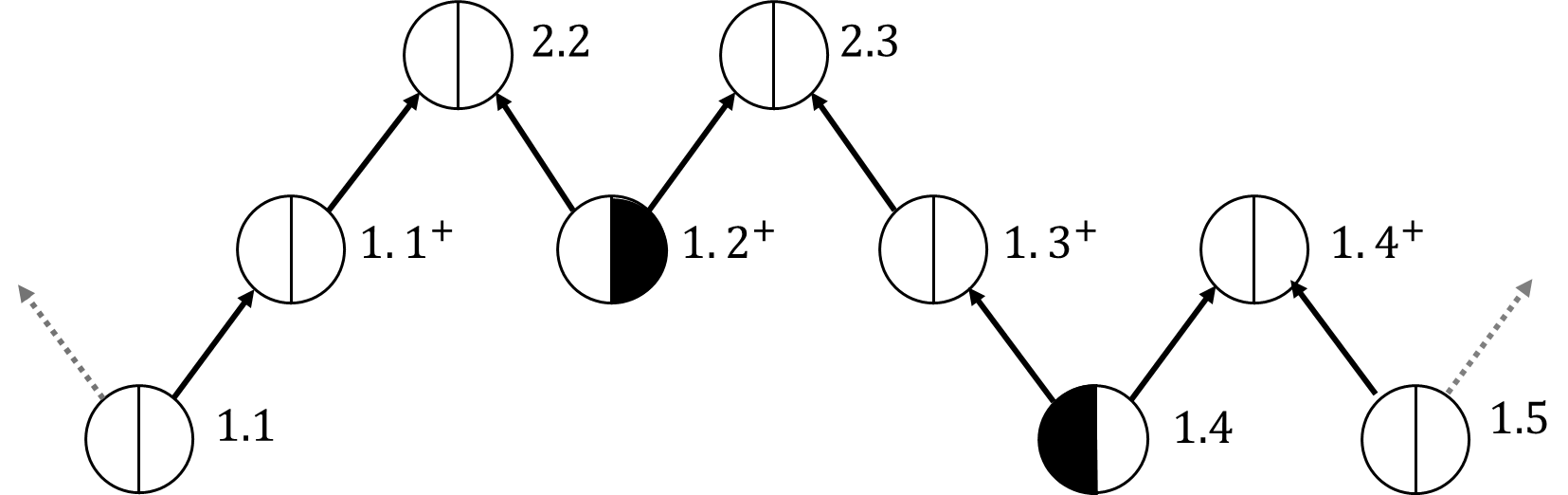}
        \caption{$A_3A_{2}G^1$.}\label{subfig: A3A2G1}
    \end{subfigure}
    \caption{\emph{Evolution of a configuration.} We start from a graph $G^1$, we then apply the rewrite pattern around $1.2$, and we apply it again around $1.3$. We end up with a new graph $A_3A_2 G^1$ where the vertex $1.2^+$ has no incoming edges anymore. Note that the gate $S_{1.2^+}$ has been computed in the process, and its output $(1,0)$ has been stored in the internal state of $1.2^+$.}
    \label{fig: async local rule}
\end{figure}

Starting from an initial graph $G$, we can derive the collection $\mathcal{M}(G) =\{ A_{x_n\dots x_1} G\mid x_n\dots x_1\in \calX^*\}$
of all possible evolutions. It gives us a new notion of space-time diagram, from which we can derive a global partial order between all the vertices that will eventually appear in a configuration as illustrated in Fig.~\ref{fig: causal ST diagram}.
This construction evokes a well-known result in theoretical physics: the geometry of space-time can be recovered from its causal structure \cite{malament_theorem}.
From a computer science perspective, the derived partial order can be interpreted as an elementary event structure \cite{Event_structure}.

\begin{figure}
    \centering
    \includegraphics[width=0.5\linewidth]{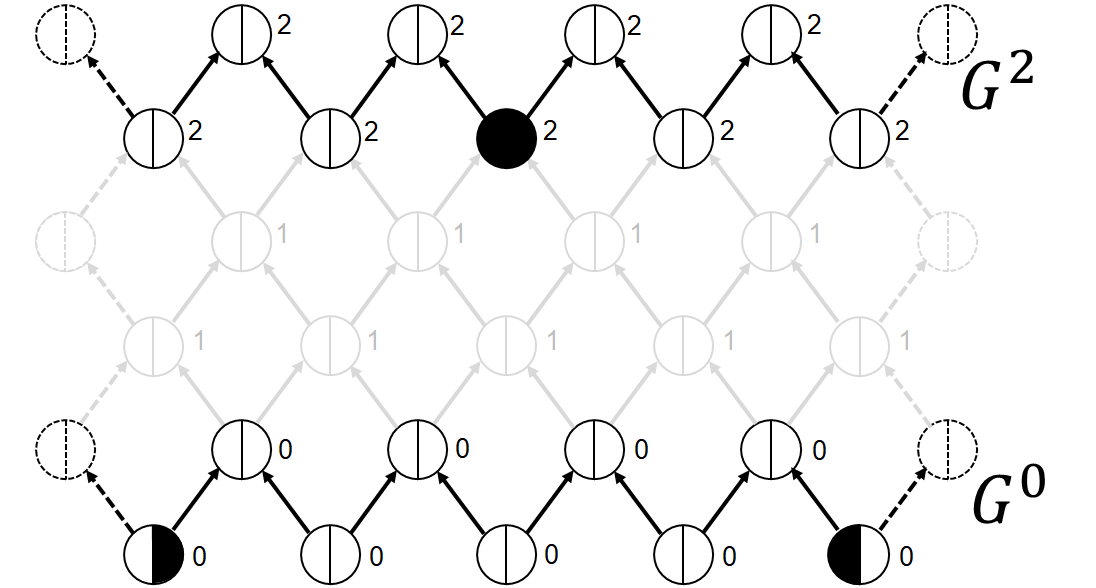}
    \caption{\emph{Representation of the causal space-time diagram.} In black we highlight just the graphs $G^0$ and $G^2$ belonging to the space-time diagram $\mathcal{M}(G)$. In grey we represent all possible rewritings obtained by applying the local rule.}
    \label{fig: causal ST diagram}
\end{figure}

Although the asynchronous simulation of cellular automata is already a well-studied topic \cite{WeakConsistencyGacs,MarchingNehaniv}, it is not the main focus of this paper.

One of the most remarkable features of general relativity is that the geometry of space-time is dynamic: it evolves depending on the distribution and movement of matter. In contrast, our earlier example involves a graph that always evolves in the same fixed way. To reflect the dynamical nature of space-time more faithfully, we aim to allow the shape of the rewrite pattern itself to vary, depending on the local topology of the graph and the internal states of its nodes. 

This analogy with general relativity serves as a conceptual motivation: our goal is not to model physical gravity, but to explore how local computational rules can give rise to dynamic, state-dependent topologies.

In this paper, we define a general framework for asynchronous graph evolution, in which arbitrary local rewriting patterns are allowed as long as they respect a constraint of locality. We then prove that this property of locality is preserved through sequential composition of these local rule applications.

\section{The model}
\subsection{Graphs, Subgraphs and Composition}\label{sec:graphs}

\todo[inline]{L: Je trouvais ça bien de parler explicitement des inpirations et des travaux précédents. C'est déjà mentionné dans l'intro, mais je pense utile de dire quels ingrédients viennent d'où plus exactement. Pour l'instant, on rentre trop brutalement dans le formalisme par rapport à mon feeling naturel... ou bien ça peut être des remarques éventuellement.}

We begin by introducing the type of graphs considered in this work: directed acyclic labelled port graphs. A \emph{port graph} is one in which edges are attached to the ports of nodes rather than to the nodes themselves. A \emph{labelled graph} is one in which each node $u$ carries an internal state $\sigma(u)$ drawn from a finite set $\Sigma=\{0,1,\ldots\}$.\\
Labelled port graphs are standard in distributed computing~\cite{PapazianRemila,Chalopin,KrivineKappa}: ports distinguish neighbouring processes, while labels represent their internal states. Directed acyclic graphs are also commonly used to capture dependencies between processes~\cite{CausalSets,AsynchronousSimsSync}. Such graphs should not be surprising in light of the discussion in Sec.~\ref{sec : preliminaries}.\\
Each vertex $t.x$ may be understood as an event, featuring a computational process at position $x$ and time tag $t$. Moreover, partial views of these graphs may have \emph{borders}: ``dangling'' edges that connect internal vertices to border vertices, whose internal states are unknown. This leads us to the following definition, illustrated in Ex.~\ref{ex:first graph}. 

\begin{definition}[Positions, ports, states, names]
Let $\calX$ be an infinite countable set of \emph{positions}.
Let $\pi$ be a finite set of \emph{ports} and $\Sigma$ be a finite set of \emph{states}.
Let us denote $\calV := \{\, t.x \mid (t,x) \in \Z \times \calX \,\}$ and call its elements \emph{names}. We say that $t.x\in \calV$ is at position $x$. For any subsets $U \subseteq \calV$ and $X\subseteq \mathcal{X}$, let us define $\sp{U} = \{\, x\in \calX \mid \exists t\in \mathbb{Z}, t.x\in U \,\}$. 
Let us also denote $\overline{X} := \mathcal{X} \setminus X$, and 
$t'.u$ for $(t'+t).x \in \calV$.
\end{definition}

\begin{definition}[Graphs]\label{def : graphs}
A \emph{graph} $G$ is given by a tuple $(\I_G, \B_G, \E_G, \sigma_G)$ where:
\begin{itemize}
    \item $\I_G \subseteq \calV$ is the set of \emph{internal vertices of $G$},
    \item $\B_G \subseteq \calV \setminus \I_G$ is the set of \emph{border vertices of $G$},
    \item $\E_G \subseteq (\V_G\!:\!\pi )^{2} \setminus ( \B_{G}\!:\!\pi )^2$ is the set of \emph{(oriented) edges}, and
    \item $\sigma_G :\I_G\rightarrow \Sigma$ maps each internal vertex to its state.
\end{itemize}
where we denote by $\V_G := \I_G\cup \B_G$ the set of all vertices of the graph, and by $(V\!:\!\pi):=\{v\!:\!p\mid v\in V,\,p\in\pi\}$ the set of ports of some set of vertices $V$. Moreover the tuple has to be such that:
\begin{description}[itemsep = 0.25ex, topsep = 0.5ex]
    \item[acyclicity\label{graph : acyclicity}] 
    the graph has no cycles.
    \item[border-attachment\label{graph : border-attachment}] $\forall u \in \B_G,\, \exists (v\!:\!a,v'\!:\!a') \in \E_G,\, u \in \{\,v,v'\,\}$,\\ i.e. border vertices are connected to at least one edge.
    \item[port-saturation\label{graph : port-saturation}] $\forall (u\!:\!a,v\!:\!b),(u'\!:\!a',v'\!:\!b') \in {\E_G},\, u\!:\!a\neq v'\!:\!b' \,\wedge\, (u\!:\!a=u'\!:\!a' \Leftrightarrow v\!:\!b=v'\!:\!b'$)\\
    i.e. ports are used only once---so as to distinguish each neighbour.
    \item[unicity of positions\label{graph : unicity of positions}] $\forall t.x, t'.x' \in {\V_{G}},\, x = x' \Rightarrow t = t'$,\\
    i.e. positions appear only once. \RX{Non overlapping positions}{}
    \item[finite chains\label{graph : finite chains}] the graph only contains paths of finite length.
\end{description}
We denote by $\Past(G)\subseteq \I_G$ the vertices of $\I_G$ with no incoming edges, 
and 
by $\mathcal{G}$ the set of all graphs. 
\end{definition}

Intuitively, vertices represent computational processes, and each edge indicates that the target process is waiting for the source process. The vertices in $\Past(G)$ correspond to processes that are no longer waiting for results from others, and are therefore ready to be executed.

As a first illustration, consider the following example, inspired by the discussion in Sec.~\ref{sec : preliminaries}.

\todo[inline]{L: Ici, j'insisterai sur les différences avec la définition des graphs dans le CGD et cette définition-ci, toujours dans l'idée de discuter la façon dont le design du modèle a eu lieu.}

\begin{example}{Particle system---A first graph}{first graph}
    \begin{minipage}{\linewidth}
        We fix $\calX = \mathbb{Z}$, $\Sigma = \{0,1\}^2$ 
        and $\pi=\{l,r\}$.
        The graph $G=(\I_G, \B_G, \E_G, \sigma_G)$ pictured below is defined as follows. $\I_G =\{1.2,2.3,1.4,1.5,1.6\}$ represents its internal vertices in black, and $\B_G =\{1.1,1.7\}$ its border vertices in grey. $\E_G$ contains all the edges, those between internal vertices (in black) as for example $(1.5\!:\!l,1.4\!:\!r)$, but also ``border'' edges linking a border vertex to an internal one (in grey) as $(1.1\!:\!r,1.2\!:\!l)$. 
        The internal states are then described by $\sigma_G$ which maps : $1.2,1.4,1.5\mapsto (0,0)$, $1.6\mapsto(1,0)$, and $1.3\mapsto(0,1)$.
    \end{minipage}
    \begin{minipage}{\linewidth}
        \centering
        \includegraphics[width=0.45\linewidth]{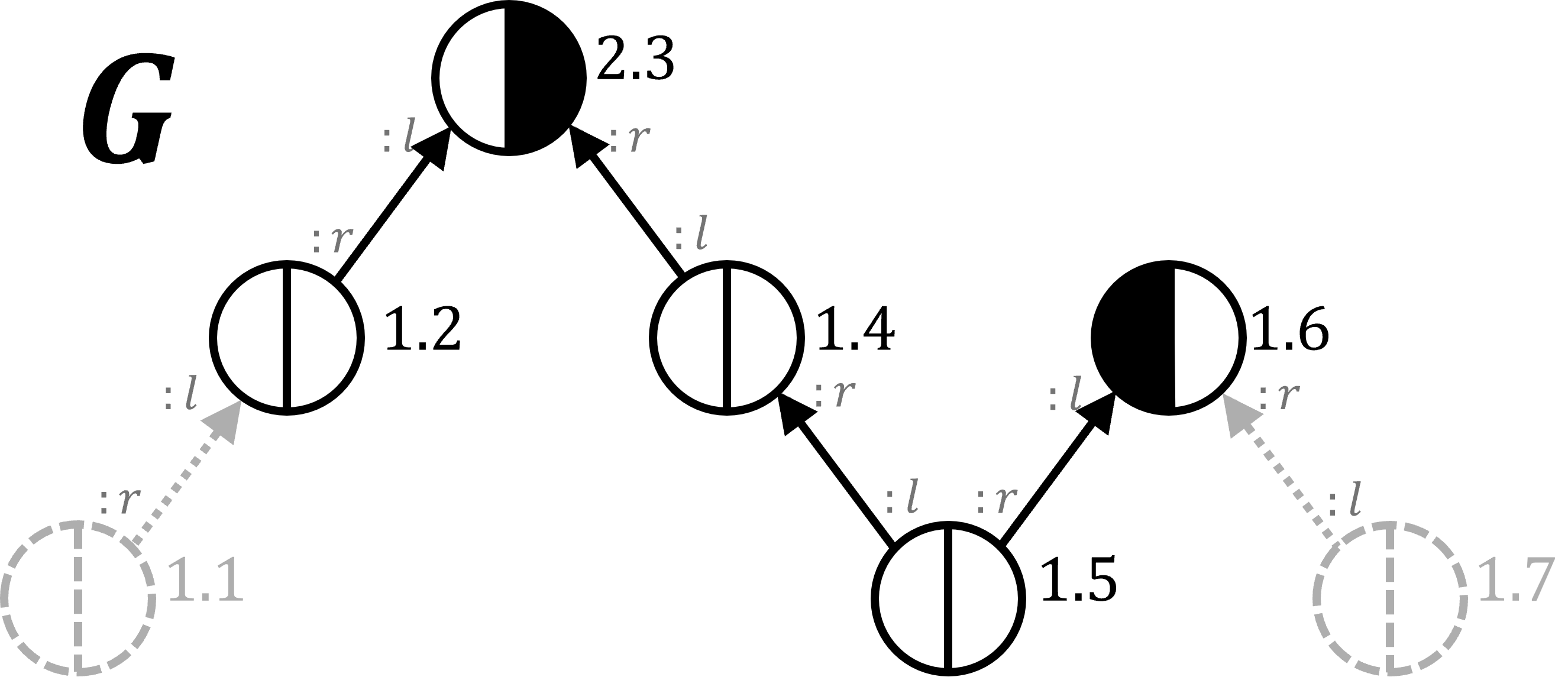}
    \end{minipage}
\end{example}

We now examine in more detail the conditions in our definition of a graph. Each position $x$ appears only once, to avoid name conflicts (see Fig.~\ref{fig : name inconsistency 1}). Each vertex port $:\!a$ can be used at most once per vertex (see Fig.~\ref{fig : port saturation}). Border vertices exclusively represent dangling edges: there are no edges between border vertices (see Fig.~\ref{fig : border attachement}), and each is always attached to some internal vertex (see Rmk.~\ref{rem:border vertices}). Equivalently, every edge is attached to at least one internal vertex. This last requirement is essential: it guarantees that the subgraph relation $\sqsubseteq$ defined below is reflexive. Indeed, Rmk.~\ref{rem:induction on internal vertices} then yields $G_{\I_G}=G$, hence $G\sqsubseteq G$. 
Finally, we require that every chain in the graph be finite, a condition that will be useful for certain inductive proofs.

\begin{figure}[h]
    \begin{subfigure}{0.3\textwidth}
        \captionsetup{justification=centering}
        \centering
        \includegraphics[width=0.7\textwidth]{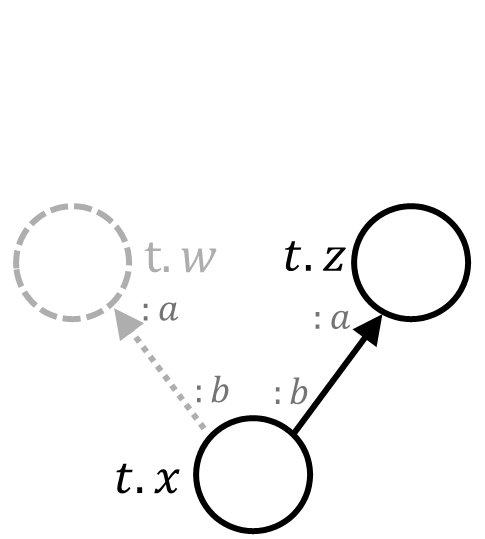}
        \caption{Ports are used once per vertex.}
        \label{fig : port saturation}
    \end{subfigure}
    \begin{subfigure}{0.3\textwidth}
        \captionsetup{justification=centering}
        \centering
        \includegraphics[width=0.75\textwidth]{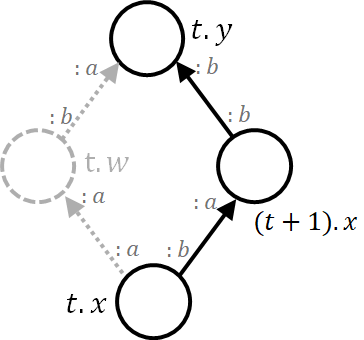}
        \caption{Positions must be distinct.}
        \label{fig : name inconsistency 1}
    \end{subfigure}
    \begin{subfigure}{0.3\textwidth}
        \captionsetup{justification=centering}
        \centering
        \includegraphics[width=0.65\textwidth]{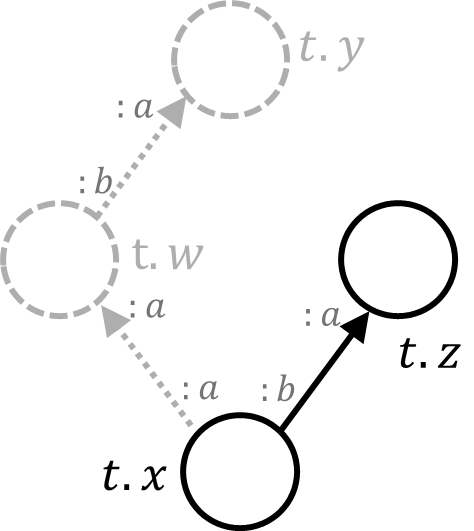}
        \caption{Border vertices cannot be connected.}
        \label{fig : border attachement}
    \end{subfigure}
    \caption{\emph{Example of non-graphs.} 
    $(a)$ is not a graph because the port $b$ of vertex $t.x$ is used twice. 
    $(b)$ is not a graph because it contains twice the same position $x$ (in $t.x$ and $(t+1).x$). 
    We forbid it because if the vertex $t.x$ were to become $(t+1).x$, a name conflict would arise.
    $(c)$ is also not a graph, because it misuses border vertices: no edges are allowed between border vertices, yet here an edge $(t.w\!:\!b, t.y\!:\!a)$ is present.}
    \label{fig: not a graph}
\end{figure}

\begin{remark}{Border vertices are attached to internal vertices}{border vertices}
    Since we allow no edges between border vertices (see Fig.~\ref{fig : border attachement}), the \ref{graph : border-attachment} condition implies that for all $v\in \B_G$, there exists an edge between $v$ and some $u\in \I_G$. 
\end{remark}

Given a graph $G$, we wish to define a local rule that acts only on a set of positions $X \subseteq \calX$. For this purpose, we introduce the induced graph $G_X$: the subgraph of $G$ that contains exactly the information relevant to the vertices $\{t.x\mid x\in X\ t\in \mathbb{Z}\}$, namely their internal states and their connectivity.

\begin{definition}[Induced graphs, subgraphs]\label{def : induced graph}
    Given a subset $U \subseteq \mathcal{V}$, we define the induced subgraph $G_U$ as the graph whose internal vertices $\I_{G_U}$ are given by $\I_G\cap U$, and whose edges $\E_{G_U}$ are all those edges of $G$ which touch a vertex in $\I_{G_U}$. \MC{Thus, its border vertices $\B_{G_U}$ are the vertices of $\V_G\setminus \I_{G_U}$ which lie at distance one of $\I_{G_U}$ in $G$, and the internal state of a vertex $v\in \I_{G_U}$ is given by $\sigma_{G_U}(v)=\sigma_G(v)$. }
    This defines a partial order relation: given any graph $H$ if there exists $U \subseteq \mathcal{V}$ such that $H=G_U$ we say that $H$ is a subgraph of $G$ and denote this relation by $H\sqsubseteq G$. For a subset $X\subseteq \mathcal{X}$, we write $G_X$ for the induced subgraph $G_{\{\, t.x \mid t\in \Z,x\in X \,\}}$.
\end{definition}

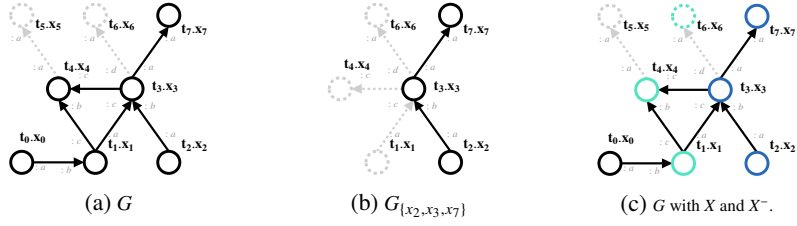
\begin{figure*}[t]
\hfil
\begin{subfigure}{0.250\textwidth}
    \captionsetup{justification=centering}
    \centering
    \resizebox{\textwidth}{!}{\tikzset{every picture/.style={line width=0.75pt}} 

\begin{tikzpicture}[x=0.75pt,y=0.75pt,yscale=-1,xscale=1]

\draw [line width=2.25]    (220,185) -- (155,185) ;
\draw [shift={(150,185)}, rotate = 360] [fill=myblack  ][line width=0.08]  [draw opacity=0] (14.29,-6.86) -- (0,0) -- (14.29,6.86) -- cycle    ;
\draw [line width=2.25]    (100,285) -- (165,285) ;
\draw [shift={(170,285)}, rotate = 180] [fill=myblack  ][line width=0.08]  [draw opacity=0] (14.29,-6.86) -- (0,0) -- (14.29,6.86) -- cycle    ;
\draw [color=border  ,draw opacity=1 ][line width=2.25]  [dash pattern={on 2.53pt off 3.02pt}]  (135,170) -- (87.91,104.07) ;
\draw [shift={(85,100)}, rotate = 54.46] [fill=border  ,fill opacity=1 ][line width=0.08]  [draw opacity=0] (14.29,-6.86) -- (0,0) -- (14.29,6.86) -- cycle    ;
\draw [line width=2.25]    (185,270) -- (137.91,204.07) ;
\draw [shift={(135,200)}, rotate = 54.46] [fill=myblack  ][line width=0.08]  [draw opacity=0] (14.29,-6.86) -- (0,0) -- (14.29,6.86) -- cycle    ;
\draw [color=border  ,draw opacity=1 ][line width=2.25]  [dash pattern={on 2.53pt off 3.02pt}]  (235,170) -- (187.91,104.07) ;
\draw [shift={(185,100)}, rotate = 54.46] [fill=border  ,fill opacity=1 ][line width=0.08]  [draw opacity=0] (14.29,-6.86) -- (0,0) -- (14.29,6.86) -- cycle    ;
\draw [line width=2.25]    (285,270) -- (237.91,204.07) ;
\draw [shift={(235,200)}, rotate = 54.46] [fill=myblack  ][line width=0.08]  [draw opacity=0] (14.29,-6.86) -- (0,0) -- (14.29,6.86) -- cycle    ;
\draw [line width=2.25]    (185,270) -- (232.09,204.07) ;
\draw [shift={(235,200)}, rotate = 125.54] [fill=myblack  ][line width=0.08]  [draw opacity=0] (14.29,-6.86) -- (0,0) -- (14.29,6.86) -- cycle    ;
\draw [line width=2.25]    (235,170) -- (282.09,104.07) ;
\draw [shift={(285,100)}, rotate = 125.54] [fill=myblack  ][line width=0.08]  [draw opacity=0] (14.29,-6.86) -- (0,0) -- (14.29,6.86) -- cycle    ;
\draw  [color=myblack  ,draw opacity=1 ][fill={rgb, 255:red, 255; green, 255; blue, 255 }  ,fill opacity=1 ][line width=3]  (170,285) .. controls (170,276.72) and (176.72,270) .. (185,270) .. controls (193.28,270) and (200,276.72) .. (200,285) .. controls (200,293.28) and (193.28,300) .. (185,300) .. controls (176.72,300) and (170,293.28) .. (170,285) -- cycle ;
\draw  [color=myblack  ,draw opacity=1 ][fill={rgb, 255:red, 255; green, 255; blue, 255 }  ,fill opacity=1 ][line width=3]  (70,285) .. controls (70,276.72) and (76.72,270) .. (85,270) .. controls (93.28,270) and (100,276.72) .. (100,285) .. controls (100,293.28) and (93.28,300) .. (85,300) .. controls (76.72,300) and (70,293.28) .. (70,285) -- cycle ;
\draw  [color=myblack  ,draw opacity=1 ][fill={rgb, 255:red, 255; green, 255; blue, 255 }  ,fill opacity=1 ][line width=3]  (270,285) .. controls (270,276.72) and (276.72,270) .. (285,270) .. controls (293.28,270) and (300,276.72) .. (300,285) .. controls (300,293.28) and (293.28,300) .. (285,300) .. controls (276.72,300) and (270,293.28) .. (270,285) -- cycle ;
\draw  [color=myblack  ,draw opacity=1 ][fill={rgb, 255:red, 255; green, 255; blue, 255 }  ,fill opacity=1 ][line width=3]  (120,185) .. controls (120,176.72) and (126.72,170) .. (135,170) .. controls (143.28,170) and (150,176.72) .. (150,185) .. controls (150,193.28) and (143.28,200) .. (135,200) .. controls (126.72,200) and (120,193.28) .. (120,185) -- cycle ;
\draw  [color=myblack  ,draw opacity=1 ][fill={rgb, 255:red, 255; green, 255; blue, 255 }  ,fill opacity=1 ][line width=3]  (220,185) .. controls (220,176.72) and (226.72,170) .. (235,170) .. controls (243.28,170) and (250,176.72) .. (250,185) .. controls (250,193.28) and (243.28,200) .. (235,200) .. controls (226.72,200) and (220,193.28) .. (220,185) -- cycle ;
\draw  [color=border  ,draw opacity=1 ][fill={rgb, 255:red, 255; green, 255; blue, 255 }  ,fill opacity=1 ][dash pattern={on 3.38pt off 3.27pt}][line width=3]  (170,85) .. controls (170,76.72) and (176.72,70) .. (185,70) .. controls (193.28,70) and (200,76.72) .. (200,85) .. controls (200,93.28) and (193.28,100) .. (185,100) .. controls (176.72,100) and (170,93.28) .. (170,85) -- cycle ;
\draw  [color=border  ,draw opacity=1 ][fill={rgb, 255:red, 255; green, 255; blue, 255 }  ,fill opacity=1 ][dash pattern={on 3.38pt off 3.27pt}][line width=3]  (70,85) .. controls (70,76.72) and (76.72,70) .. (85,70) .. controls (93.28,70) and (100,76.72) .. (100,85) .. controls (100,93.28) and (93.28,100) .. (85,100) .. controls (76.72,100) and (70,93.28) .. (70,85) -- cycle ;
\draw  [color=myblack  ,draw opacity=1 ][fill={rgb, 255:red, 255; green, 255; blue, 255 }  ,fill opacity=1 ][line width=3]  (270,85) .. controls (270,76.72) and (276.72,70) .. (285,70) .. controls (293.28,70) and (300,76.72) .. (300,85) .. controls (300,93.28) and (293.28,100) .. (285,100) .. controls (276.72,100) and (270,93.28) .. (270,85) -- cycle ;

\draw (81,242.4) node [anchor=north west][inner sep=0.75pt]  [font=\sizeFtwo]  {$\mathbf{t_{0} .x_{0}}$};
\draw (201,252.4) node [anchor=north west][inner sep=0.75pt]  [font=\sizeFtwo]  {$\mathbf{t_{1} .x_{1}}$};
\draw (301,252.4) node [anchor=north west][inner sep=0.75pt]  [font=\sizeFtwo]  {$\mathbf{t_{2} .x_{2}}$};
\draw (141,145.4) node [anchor=north west][inner sep=0.75pt]  [font=\sizeFtwo]  {$\mathbf{t_{4} .x_{4}}$};
\draw (261,172.4) node [anchor=north west][inner sep=0.75pt]  [font=\sizeFtwo]  {$\mathbf{t_{3} .x_{3}}$};
\draw (102,88.4) node [anchor=north west][inner sep=0.75pt]  [font=\sizeFtwo]  {$\mathbf{t_{5} .x_{5}}$};
\draw (202,88.4) node [anchor=north west][inner sep=0.75pt]  [font=\sizeFtwo]  {$\mathbf{t_{6} .x_{6}}$};
\draw (301,95.4) node [anchor=north west][inner sep=0.75pt]  [font=\sizeFtwo]  {$\mathbf{t_{7} .x_{7}}$};
\draw (279,242.4) node [anchor=north west][inner sep=0.75pt]  [font=\sizeFtwop,color={rgb, 255:red, 155; green, 155; blue, 155 }  ,opacity=1 ]  {$:a$};
\draw (201,242.4) node [anchor=north west][inner sep=0.75pt]  [font=\sizeFtwop,color={rgb, 255:red, 155; green, 155; blue, 155 }  ,opacity=1 ]  {$:a$};
\draw (102,288.4) node [anchor=north west][inner sep=0.75pt]  [font=\sizeFtwop,color={rgb, 255:red, 155; green, 155; blue, 155 }  ,opacity=1 ]  {$:a$};
\draw (99,152.4) node [anchor=north west][inner sep=0.75pt]  [font=\sizeFtwop,color={rgb, 255:red, 155; green, 155; blue, 155 }  ,opacity=1 ]  {$:a$};
\draw (69,112.4) node [anchor=north west][inner sep=0.75pt]  [font=\sizeFtwop,color={rgb, 255:red, 155; green, 155; blue, 155 }  ,opacity=1 ]  {$:a$};
\draw (169,112.4) node [anchor=north west][inner sep=0.75pt]  [font=\sizeFtwop,color={rgb, 255:red, 155; green, 155; blue, 155 }  ,opacity=1 ]  {$:a$};
\draw (279,112.4) node [anchor=north west][inner sep=0.75pt]  [font=\sizeFtwop,color={rgb, 255:red, 155; green, 155; blue, 155 }  ,opacity=1 ]  {$:a$};
\draw (249,152.4) node [anchor=north west][inner sep=0.75pt]  [font=\sizeFtwop,color={rgb, 255:red, 155; green, 155; blue, 155 }  ,opacity=1 ]  {$:a$};
\draw (251,202.4) node [anchor=north west][inner sep=0.75pt]  [font=\sizeFtwop,color={rgb, 255:red, 155; green, 155; blue, 155 }  ,opacity=1 ]  {$:b$};
\draw (199,202.4) node [anchor=north west][inner sep=0.75pt]  [font=\sizeFtwop,color={rgb, 255:red, 155; green, 155; blue, 155 }  ,opacity=1 ]  {$:c$};
\draw (137,288.4) node [anchor=north west][inner sep=0.75pt]  [font=\sizeFtwop,color={rgb, 255:red, 155; green, 155; blue, 155 }  ,opacity=1 ]  {$:b$};
\draw (151,202.4) node [anchor=north west][inner sep=0.75pt]  [font=\sizeFtwop,color={rgb, 255:red, 155; green, 155; blue, 155 }  ,opacity=1 ]  {$:b$};
\draw (161,162.4) node [anchor=north west][inner sep=0.75pt]  [font=\sizeFtwop,color={rgb, 255:red, 155; green, 155; blue, 155 }  ,opacity=1 ]  {$:c$};
\draw (198,152.4) node [anchor=north west][inner sep=0.75pt]  [font=\sizeFtwop,color={rgb, 255:red, 155; green, 155; blue, 155 }  ,opacity=1 ]  {$:d$};
\draw (151,252.4) node [anchor=north west][inner sep=0.75pt]  [font=\sizeFtwop,color={rgb, 255:red, 155; green, 155; blue, 155 }  ,opacity=1 ]  {$:c$};

\end{tikzpicture}}
    \caption{$G$}
    \label{fig : Induced subgraph 1}
\end{subfigure}
\hfil
\hfil
\begin{subfigure}{0.410\textwidth}
    \captionsetup{justification=centering}
    \centering
    \resizebox{.5\textwidth}{!}{\tikzset{every picture/.style={line width=0.75pt}} 

\begin{tikzpicture}[x=0.75pt,y=0.75pt,yscale=-1,xscale=1]

\draw [color=border  ,draw opacity=1 ][line width=2.25]  [dash pattern={on 2.53pt off 3.02pt}]  (222,185) -- (157,185) ;
\draw [shift={(152,185)}, rotate = 360] [fill=border  ,fill opacity=1 ][line width=0.08]  [draw opacity=0] (14.29,-6.86) -- (0,0) -- (14.29,6.86) -- cycle    ;
\draw [color=border  ,draw opacity=1 ][line width=2.25]  [dash pattern={on 2.53pt off 3.02pt}]  (237,170) -- (189.91,104.07) ;
\draw [shift={(187,100)}, rotate = 54.46] [fill=border  ,fill opacity=1 ][line width=0.08]  [draw opacity=0] (14.29,-6.86) -- (0,0) -- (14.29,6.86) -- cycle    ;
\draw [line width=2.25]    (287,270) -- (239.91,204.07) ;
\draw [shift={(237,200)}, rotate = 54.46] [fill=myblack  ][line width=0.08]  [draw opacity=0] (14.29,-6.86) -- (0,0) -- (14.29,6.86) -- cycle    ;
\draw [color=border  ,draw opacity=1 ][line width=2.25]  [dash pattern={on 2.53pt off 3.02pt}]  (187,270) -- (234.09,204.07) ;
\draw [shift={(237,200)}, rotate = 125.54] [fill=border  ,fill opacity=1 ][line width=0.08]  [draw opacity=0] (14.29,-6.86) -- (0,0) -- (14.29,6.86) -- cycle    ;
\draw [line width=2.25]    (237,170) -- (284.09,104.07) ;
\draw [shift={(287,100)}, rotate = 125.54] [fill=myblack  ][line width=0.08]  [draw opacity=0] (14.29,-6.86) -- (0,0) -- (14.29,6.86) -- cycle    ;
\draw  [color=border  ,draw opacity=1 ][fill={rgb, 255:red, 255; green, 255; blue, 255 }  ,fill opacity=1 ][dash pattern={on 3.38pt off 3.27pt}][line width=3]  (172,285) .. controls (172,276.72) and (178.72,270) .. (187,270) .. controls (195.28,270) and (202,276.72) .. (202,285) .. controls (202,293.28) and (195.28,300) .. (187,300) .. controls (178.72,300) and (172,293.28) .. (172,285) -- cycle ;
\draw  [color=myblack  ,draw opacity=1 ][fill={rgb, 255:red, 255; green, 255; blue, 255 }  ,fill opacity=1 ][line width=3]  (272,285) .. controls (272,276.72) and (278.72,270) .. (287,270) .. controls (295.28,270) and (302,276.72) .. (302,285) .. controls (302,293.28) and (295.28,300) .. (287,300) .. controls (278.72,300) and (272,293.28) .. (272,285) -- cycle ;
\draw  [color=border  ,draw opacity=1 ][fill={rgb, 255:red, 255; green, 255; blue, 255 }  ,fill opacity=1 ][dash pattern={on 3.38pt off 3.27pt}][line width=3]  (122,185) .. controls (122,176.72) and (128.72,170) .. (137,170) .. controls (145.28,170) and (152,176.72) .. (152,185) .. controls (152,193.28) and (145.28,200) .. (137,200) .. controls (128.72,200) and (122,193.28) .. (122,185) -- cycle ;
\draw  [color=myblack  ,draw opacity=1 ][fill={rgb, 255:red, 255; green, 255; blue, 255 }  ,fill opacity=1 ][line width=3]  (222,185) .. controls (222,176.72) and (228.72,170) .. (237,170) .. controls (245.28,170) and (252,176.72) .. (252,185) .. controls (252,193.28) and (245.28,200) .. (237,200) .. controls (228.72,200) and (222,193.28) .. (222,185) -- cycle ;
\draw  [color=border  ,draw opacity=1 ][fill={rgb, 255:red, 255; green, 255; blue, 255 }  ,fill opacity=1 ][dash pattern={on 3.38pt off 3.27pt}][line width=3]  (172,85) .. controls (172,76.72) and (178.72,70) .. (187,70) .. controls (195.28,70) and (202,76.72) .. (202,85) .. controls (202,93.28) and (195.28,100) .. (187,100) .. controls (178.72,100) and (172,93.28) .. (172,85) -- cycle ;
\draw  [color=myblack  ,draw opacity=1 ][fill={rgb, 255:red, 255; green, 255; blue, 255 }  ,fill opacity=1 ][line width=3]  (272,85) .. controls (272,76.72) and (278.72,70) .. (287,70) .. controls (295.28,70) and (302,76.72) .. (302,85) .. controls (302,93.28) and (295.28,100) .. (287,100) .. controls (278.72,100) and (272,93.28) .. (272,85) -- cycle ;

\draw (203,252.4) node [anchor=north west][inner sep=0.75pt]  [font=\sizeFtwo]  {$\mathbf{t_{1} .x_{1}}$};
\draw (303,252.4) node [anchor=north west][inner sep=0.75pt]  [font=\sizeFtwo]  {$\mathbf{t_{2} .x_{2}}$};
\draw (141,142.4) node [anchor=north west][inner sep=0.75pt]  [font=\sizeFtwo]  {$\mathbf{t_{4} .x_{4}}$};
\draw (263,172.4) node [anchor=north west][inner sep=0.75pt]  [font=\sizeFtwo]  {$\mathbf{t_{3} .x_{3}}$};
\draw (204,88.4) node [anchor=north west][inner sep=0.75pt]  [font=\sizeFtwo]  {$\mathbf{t_{6} .x_{6}}$};
\draw (303,95.4) node [anchor=north west][inner sep=0.75pt]  [font=\sizeFtwo]  {$\mathbf{t_{7} .x_{7}}$};
\draw (281,242.4) node [anchor=north west][inner sep=0.75pt]  [font=\sizeFtwop,color={rgb, 255:red, 155; green, 155; blue, 155 }  ,opacity=1 ]  {$:a$};
\draw (203,242.4) node [anchor=north west][inner sep=0.75pt]  [font=\sizeFtwop,color={rgb, 255:red, 155; green, 155; blue, 155 }  ,opacity=1 ]  {$:a$};
\draw (171,112.4) node [anchor=north west][inner sep=0.75pt]  [font=\sizeFtwop,color={rgb, 255:red, 155; green, 155; blue, 155 }  ,opacity=1 ]  {$:a$};
\draw (281,112.4) node [anchor=north west][inner sep=0.75pt]  [font=\sizeFtwop,color={rgb, 255:red, 155; green, 155; blue, 155 }  ,opacity=1 ]  {$:a$};
\draw (251,152.4) node [anchor=north west][inner sep=0.75pt]  [font=\sizeFtwop,color={rgb, 255:red, 155; green, 155; blue, 155 }  ,opacity=1 ]  {$:a$};
\draw (253,202.4) node [anchor=north west][inner sep=0.75pt]  [font=\sizeFtwop,color={rgb, 255:red, 155; green, 155; blue, 155 }  ,opacity=1 ]  {$:b$};
\draw (201,202.4) node [anchor=north west][inner sep=0.75pt]  [font=\sizeFtwop,color={rgb, 255:red, 155; green, 155; blue, 155 }  ,opacity=1 ]  {$:c$};
\draw (163,162.4) node [anchor=north west][inner sep=0.75pt]  [font=\sizeFtwop,color={rgb, 255:red, 155; green, 155; blue, 155 }  ,opacity=1 ]  {$:c$};
\draw (200,152.4) node [anchor=north west][inner sep=0.75pt]  [font=\sizeFtwop,color={rgb, 255:red, 155; green, 155; blue, 155 }  ,opacity=1 ]  {$:d$};

\end{tikzpicture}}
    \caption{$G_{\{x_2,x_3,x_7\}}$}
    \label{fig : Induced subgraph 2}
\end{subfigure}
\hfil
\begin{subfigure}{0.250\textwidth}
    \captionsetup{justification=centering}
    \centering
    \resizebox{\textwidth}{!}{\tikzset{every picture/.style={line width=0.75pt}} 

\begin{tikzpicture}[x=0.75pt,y=0.75pt,yscale=-1,xscale=1]

\draw [line width=2.25]    (220,185) -- (155,185) ;
\draw [shift={(150,185)}, rotate = 360] [fill=myblack  ][line width=0.08]  [draw opacity=0] (14.29,-6.86) -- (0,0) -- (14.29,6.86) -- cycle    ;
\draw [line width=2.25]    (100,285) -- (165,285) ;
\draw [shift={(170,285)}, rotate = 180] [fill=myblack  ][line width=0.08]  [draw opacity=0] (14.29,-6.86) -- (0,0) -- (14.29,6.86) -- cycle    ;
\draw [color=border  ,draw opacity=1 ][line width=2.25]  [dash pattern={on 2.53pt off 3.02pt}]  (135,170) -- (87.91,104.07) ;
\draw [shift={(85,100)}, rotate = 54.46] [fill=border  ,fill opacity=1 ][line width=0.08]  [draw opacity=0] (14.29,-6.86) -- (0,0) -- (14.29,6.86) -- cycle    ;
\draw [line width=2.25]    (185,270) -- (137.91,204.07) ;
\draw [shift={(135,200)}, rotate = 54.46] [fill=myblack  ][line width=0.08]  [draw opacity=0] (14.29,-6.86) -- (0,0) -- (14.29,6.86) -- cycle    ;
\draw [color=border  ,draw opacity=1 ][line width=2.25]  [dash pattern={on 2.53pt off 3.02pt}]  (235,170) -- (187.91,104.07) ;
\draw [shift={(185,100)}, rotate = 54.46] [fill=border  ,fill opacity=1 ][line width=0.08]  [draw opacity=0] (14.29,-6.86) -- (0,0) -- (14.29,6.86) -- cycle    ;
\draw [line width=2.25]    (285,270) -- (237.91,204.07) ;
\draw [shift={(235,200)}, rotate = 54.46] [fill=myblack  ][line width=0.08]  [draw opacity=0] (14.29,-6.86) -- (0,0) -- (14.29,6.86) -- cycle    ;
\draw [line width=2.25]    (185,270) -- (232.09,204.07) ;
\draw [shift={(235,200)}, rotate = 125.54] [fill=myblack  ][line width=0.08]  [draw opacity=0] (14.29,-6.86) -- (0,0) -- (14.29,6.86) -- cycle    ;
\draw [line width=2.25]    (235,170) -- (282.09,104.07) ;
\draw [shift={(285,100)}, rotate = 125.54] [fill=myblack  ][line width=0.08]  [draw opacity=0] (14.29,-6.86) -- (0,0) -- (14.29,6.86) -- cycle    ;
\draw  [color=setBorder  ,draw opacity=1 ][fill={rgb, 255:red, 255; green, 255; blue, 255 }  ,fill opacity=1 ][line width=3]  (170,285) .. controls (170,276.72) and (176.72,270) .. (185,270) .. controls (193.28,270) and (200,276.72) .. (200,285) .. controls (200,293.28) and (193.28,300) .. (185,300) .. controls (176.72,300) and (170,293.28) .. (170,285) -- cycle ;
\draw  [color=myblack  ,draw opacity=1 ][fill={rgb, 255:red, 255; green, 255; blue, 255 }  ,fill opacity=1 ][line width=3]  (70,285) .. controls (70,276.72) and (76.72,270) .. (85,270) .. controls (93.28,270) and (100,276.72) .. (100,285) .. controls (100,293.28) and (93.28,300) .. (85,300) .. controls (76.72,300) and (70,293.28) .. (70,285) -- cycle ;
\draw  [color=internal  ,draw opacity=1 ][fill={rgb, 255:red, 255; green, 255; blue, 255 }  ,fill opacity=1 ][line width=3]  (270,285) .. controls (270,276.72) and (276.72,270) .. (285,270) .. controls (293.28,270) and (300,276.72) .. (300,285) .. controls (300,293.28) and (293.28,300) .. (285,300) .. controls (276.72,300) and (270,293.28) .. (270,285) -- cycle ;
\draw  [color=setBorder  ,draw opacity=1 ][fill={rgb, 255:red, 255; green, 255; blue, 255 }  ,fill opacity=1 ][line width=3]  (120,185) .. controls (120,176.72) and (126.72,170) .. (135,170) .. controls (143.28,170) and (150,176.72) .. (150,185) .. controls (150,193.28) and (143.28,200) .. (135,200) .. controls (126.72,200) and (120,193.28) .. (120,185) -- cycle ;
\draw  [color=internal  ,draw opacity=1 ][fill={rgb, 255:red, 255; green, 255; blue, 255 }  ,fill opacity=1 ][line width=3]  (220,185) .. controls (220,176.72) and (226.72,170) .. (235,170) .. controls (243.28,170) and (250,176.72) .. (250,185) .. controls (250,193.28) and (243.28,200) .. (235,200) .. controls (226.72,200) and (220,193.28) .. (220,185) -- cycle ;
\draw  [color=setBorder  ,draw opacity=1 ][fill={rgb, 255:red, 255; green, 255; blue, 255 }  ,fill opacity=1 ][dash pattern={on 3.38pt off 3.27pt}][line width=3]  (170,85) .. controls (170,76.72) and (176.72,70) .. (185,70) .. controls (193.28,70) and (200,76.72) .. (200,85) .. controls (200,93.28) and (193.28,100) .. (185,100) .. controls (176.72,100) and (170,93.28) .. (170,85) -- cycle ;
\draw  [color=border  ,draw opacity=1 ][fill={rgb, 255:red, 255; green, 255; blue, 255 }  ,fill opacity=1 ][dash pattern={on 3.38pt off 3.27pt}][line width=3]  (70,85) .. controls (70,76.72) and (76.72,70) .. (85,70) .. controls (93.28,70) and (100,76.72) .. (100,85) .. controls (100,93.28) and (93.28,100) .. (85,100) .. controls (76.72,100) and (70,93.28) .. (70,85) -- cycle ;
\draw  [color=internal  ,draw opacity=1 ][fill={rgb, 255:red, 255; green, 255; blue, 255 }  ,fill opacity=1 ][line width=3]  (270,85) .. controls (270,76.72) and (276.72,70) .. (285,70) .. controls (293.28,70) and (300,76.72) .. (300,85) .. controls (300,93.28) and (293.28,100) .. (285,100) .. controls (276.72,100) and (270,93.28) .. (270,85) -- cycle ;

\draw (81,242.4) node [anchor=north west][inner sep=0.75pt]  [font=\sizeFtwo]  {$\mathbf{t_{0} .x_{0}}$};
\draw (201,252.4) node [anchor=north west][inner sep=0.75pt]  [font=\sizeFtwo]  {$\mathbf{t_{1} .x_{1}}$};
\draw (301,252.4) node [anchor=north west][inner sep=0.75pt]  [font=\sizeFtwo]  {$\mathbf{t_{2} .x_{2}}$};
\draw (141,145.4) node [anchor=north west][inner sep=0.75pt]  [font=\sizeFtwo]  {$\mathbf{t_{4} .x_{4}}$};
\draw (261,172.4) node [anchor=north west][inner sep=0.75pt]  [font=\sizeFtwo]  {$\mathbf{t_{3} .x_{3}}$};
\draw (102,88.4) node [anchor=north west][inner sep=0.75pt]  [font=\sizeFtwo]  {$\mathbf{t_{5} .x_{5}}$};
\draw (202,88.4) node [anchor=north west][inner sep=0.75pt]  [font=\sizeFtwo]  {$\mathbf{t_{6} .x_{6}}$};
\draw (301,95.4) node [anchor=north west][inner sep=0.75pt]  [font=\sizeFtwo]  {$\mathbf{t_{7} .x_{7}}$};
\draw (279,242.4) node [anchor=north west][inner sep=0.75pt]  [font=\sizeFtwop,color={rgb, 255:red, 155; green, 155; blue, 155 }  ,opacity=1 ]  {$:a$};
\draw (201,242.4) node [anchor=north west][inner sep=0.75pt]  [font=\sizeFtwop,color={rgb, 255:red, 155; green, 155; blue, 155 }  ,opacity=1 ]  {$:a$};
\draw (102,288.4) node [anchor=north west][inner sep=0.75pt]  [font=\sizeFtwop,color={rgb, 255:red, 155; green, 155; blue, 155 }  ,opacity=1 ]  {$:a$};
\draw (99,152.4) node [anchor=north west][inner sep=0.75pt]  [font=\sizeFtwop,color={rgb, 255:red, 155; green, 155; blue, 155 }  ,opacity=1 ]  {$:a$};
\draw (69,112.4) node [anchor=north west][inner sep=0.75pt]  [font=\sizeFtwop,color={rgb, 255:red, 155; green, 155; blue, 155 }  ,opacity=1 ]  {$:a$};
\draw (169,112.4) node [anchor=north west][inner sep=0.75pt]  [font=\sizeFtwop,color={rgb, 255:red, 155; green, 155; blue, 155 }  ,opacity=1 ]  {$:a$};
\draw (279,112.4) node [anchor=north west][inner sep=0.75pt]  [font=\sizeFtwop,color={rgb, 255:red, 155; green, 155; blue, 155 }  ,opacity=1 ]  {$:a$};
\draw (249,152.4) node [anchor=north west][inner sep=0.75pt]  [font=\sizeFtwop,color={rgb, 255:red, 155; green, 155; blue, 155 }  ,opacity=1 ]  {$:a$};
\draw (251,202.4) node [anchor=north west][inner sep=0.75pt]  [font=\sizeFtwop,color={rgb, 255:red, 155; green, 155; blue, 155 }  ,opacity=1 ]  {$:b$};
\draw (199,202.4) node [anchor=north west][inner sep=0.75pt]  [font=\sizeFtwop,color={rgb, 255:red, 155; green, 155; blue, 155 }  ,opacity=1 ]  {$:c$};
\draw (137,288.4) node [anchor=north west][inner sep=0.75pt]  [font=\sizeFtwop,color={rgb, 255:red, 155; green, 155; blue, 155 }  ,opacity=1 ]  {$:b$};
\draw (151,202.4) node [anchor=north west][inner sep=0.75pt]  [font=\sizeFtwop,color={rgb, 255:red, 155; green, 155; blue, 155 }  ,opacity=1 ]  {$:b$};
\draw (161,162.4) node [anchor=north west][inner sep=0.75pt]  [font=\sizeFtwop,color={rgb, 255:red, 155; green, 155; blue, 155 }  ,opacity=1 ]  {$:c$};
\draw (198,152.4) node [anchor=north west][inner sep=0.75pt]  [font=\sizeFtwop,color={rgb, 255:red, 155; green, 155; blue, 155 }  ,opacity=1 ]  {$:d$};
\draw (151,252.4) node [anchor=north west][inner sep=0.75pt]  [font=\sizeFtwop,color={rgb, 255:red, 155; green, 155; blue, 155 }  ,opacity=1 ]  {$:c$};

\end{tikzpicture}}
    \caption{\scalefont{0.8}$G$ with $X$ and $X^-$.}
    \label{fig : Induced subgraph 3}
\end{subfigure}
\caption{
{\em Induced subgraph and borders (a)(b):} 
We consider a graph $G$ and its induced subgraph $G_{\{x_2,x_3,x_7\}}$. Graphs have borders, as shown pointed by the dashed lines for $(a)$ $G$ and $(b)$ $G_{\{x_2,x_3,x_7\}}$. 
{\em Interior of a set (c):} We consider a set of positions $X = \{x_1,x_2,x_3,x_4,x_6,x_7\}$. Positions in this set are either interior (denoted $X^-$ and shown in dark blue) or in the boundary ($X\setminus X^-$ in cyan).}
\label{fig : Induced subgraph}
\end{figure*}

The notion of an induced graph 
retains only the internal states of vertices in $X$. But it also preserves all information about their connectivity, including connections to vertices at positions of ${\overline{X}}$\MC{, thus inducing on all internal vertices of a graph acts trivially (see Rmk.~\ref{rem:induction on internal vertices}). }
As illustrated in Fig.~\ref{fig : Induced subgraph}, any vertex not at a position of ${X}$, but connected to it becomes a border vertex. 

\begin{remark}{Induction on internal vertices is trivial}{induction on internal vertices}
    Note how Rmk.~\ref{rem:border vertices} implies that, for any graph $G$, we have :
    \begin{equation}
        G_{X} = G \equiv \I_G \subseteq X
    \end{equation}
    In particular $G_{\I_G}=G$, so the relation $\sqsubseteq$ is reflexive.
\end{remark}

We can now split a graph $G$ into two subgraphs, $G_X$ and $G_{\overline{X}}$. Later on, we will often apply an operation only to a subgraph of the graph. For this purpose, we need a notion of composition: having applied the local operation $A$ on $G_X$ we then recompose it with the rest of the graph $G_{\overline{X}}$, yielding a new graph $AG_X \sqcup G_{\overline{X}}$. When no operation is applied to the left-hand side, \MC{thanks to the \hyperref[graph : border-attachment]{border-attachment}} we simply recover the original graph $G=G_X \sqcup G_{\overline{X}}$, as illustrated in Ex.~\ref{ex:splitting a graph}.

\begin{definition}[Composition of graphs]
    Given $G,H\in \calG$ we introduce the operation $G\sqcup H$ which is only defined if $G$ and $H$ can both be obtained as induced subgraphs of the same graph, in which case $G\sqcup H$ is the smallest graph (for the order $\sqsubseteq$) such that $G\sqsubseteq G\sqcup H$ and $H\sqsubseteq G\sqcup H$.
\end{definition}

\begin{remark}{Composition as union}{composition of graphs}
    Whenever $G\sqcup H$ is defined, its vertices and edges are obtained by taking the corresponding unions in $G$ and $H$:
    \begin{align*}
        \I_{G\sqcup H} &= \I_G\cup \I_H, &
        \V_{G\sqcup H} &= \V_G\cup \V_H, &
        \E_{G\sqcup H} &= \E_G\cup \E_H.
    \end{align*}
    Its internal-state map is inherited from its components: $\sigma_{G\sqcup H}(u)=\sigma_G(u)$ for $u\in \I_G$, and $\sigma_{G\sqcup H}(u)=\sigma_H(u)$ for $u\in \I_H$. 
    In particular, $\sigma_G$ and $\sigma_H$ necessarily agree on $\I_G\cap \I_H$.
\end{remark}

\begin{example}{Splitting a graph in two subgraphs}{splitting a graph}
    \begin{minipage}{0.54\linewidth}
        Given a set of positions $X$, we can always split $G$ into two subgraphs, $G_X$ and $G_{\overline{X}}$, and recover the original graph via composition, i.e. $G_X \sqcup G_{\overline{X}} = G$. As an illustration, consider the graph $G$ from Ex.~\ref{ex:first graph} and the set of positions $X={3,6}$. In this case, $\overline{X}=\calX \setminus {3,6}$, but $G_{\overline{X}} = G_{{2,4,5}}$, since only internal vertices matter when taking an induced subgraph.
    \end{minipage}
    \begin{minipage}{0.45\linewidth}
        \centering
        \includegraphics[width=\linewidth]{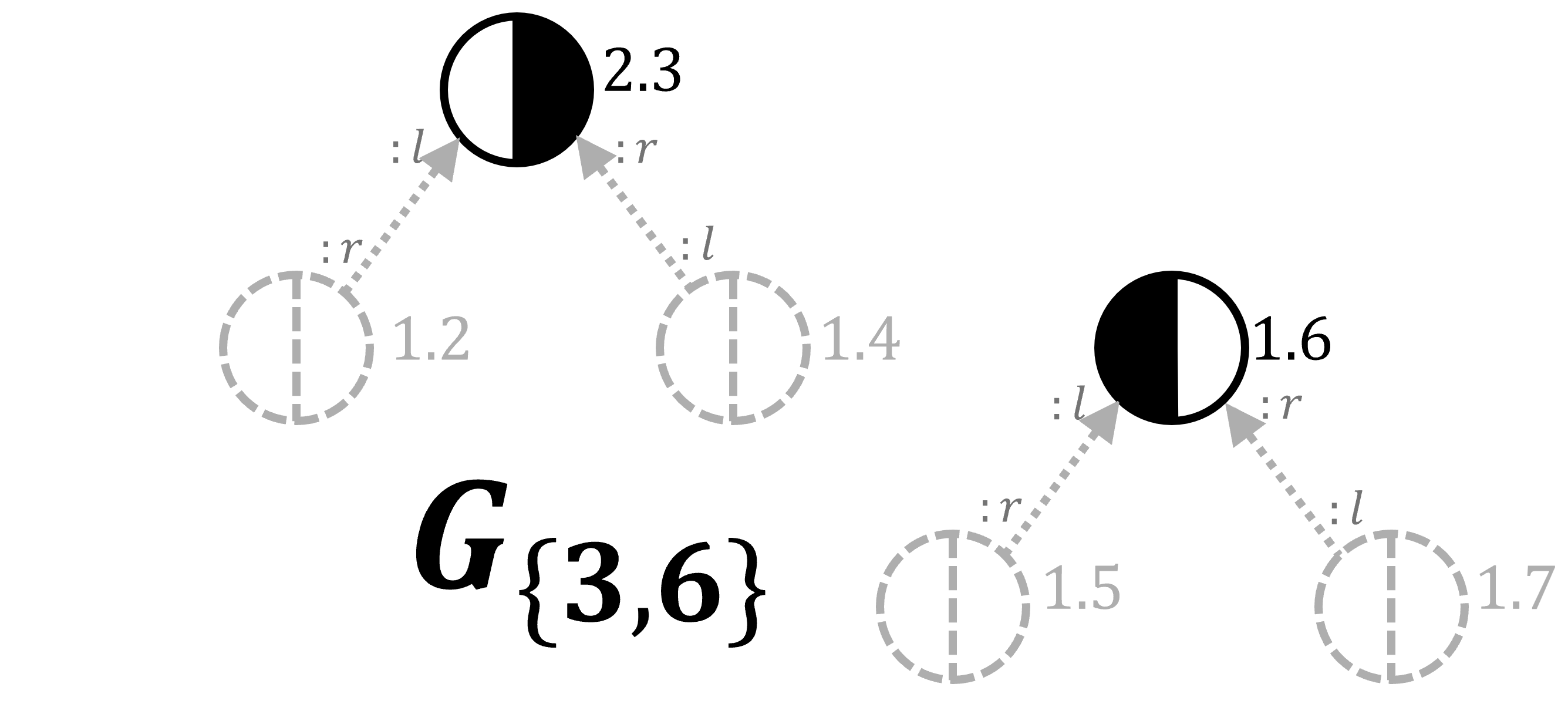}
        \centering
        \includegraphics[width=\linewidth]{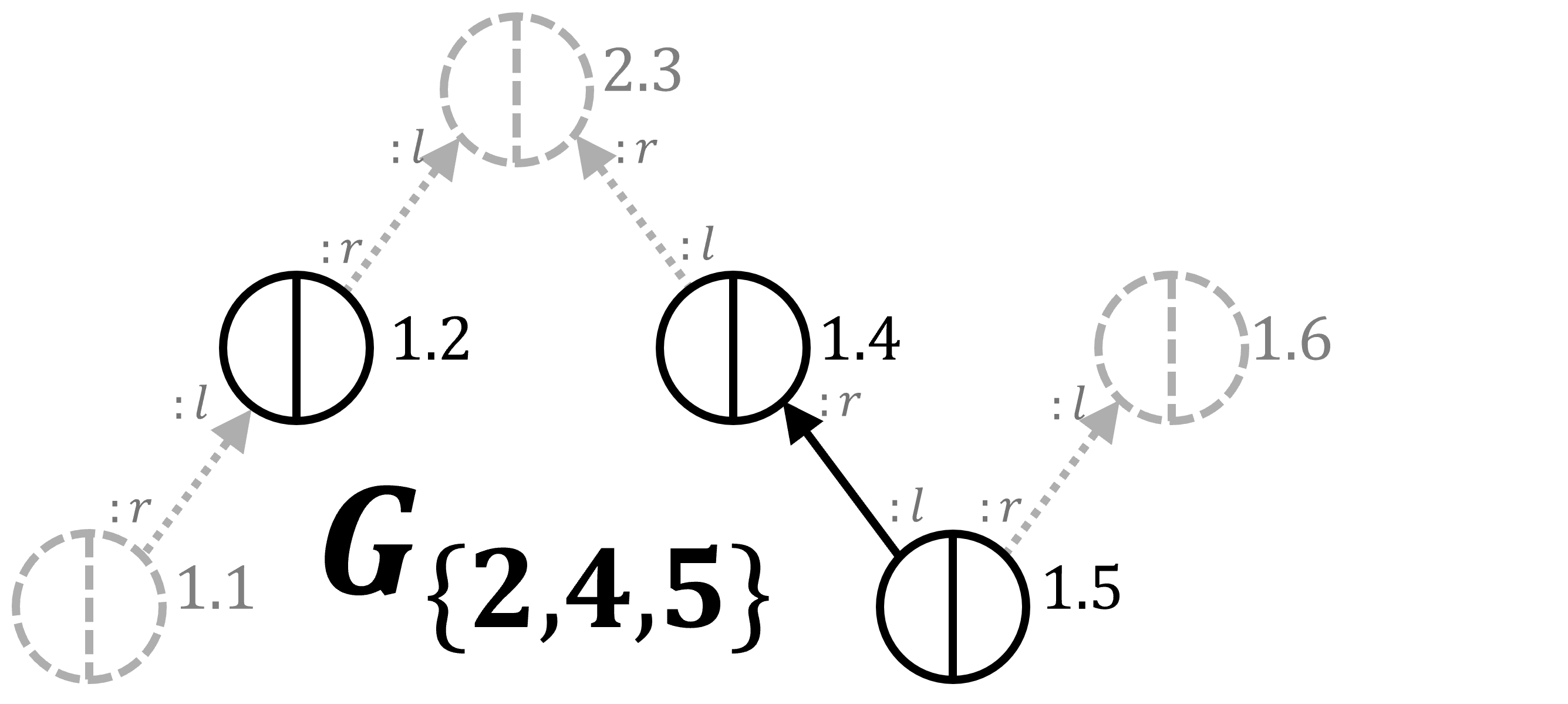}
    \end{minipage}
\end{example}

The preceding remark describes the result of composition whenever it is defined. The non-trivial question is whether two given graphs are compatible---that is, whether the union of their data forms a graph in which both occur as induced subgraphs. This question will arise when a subgraph $G_X$ is transformed into $AG_X$ and then recomposed with its context $G_{\overline{X}}$. The following lemma characterises compatibility in the case relevant here, where the two graphs have disjoint sets of internal vertices.


\begin{lemma}{Characterization of compatible disjoint subgraphs}{compatible subgraphs}
    Consider two graphs $G$ and $H$ that share no internal vertices---i.e. $\I_G\cap \I_H=\emptyset$. $H\sqcup G$ is well defined if and only if :
    \begin{enumerate}
        \item positions are unique---i.e. $\forall t.x\in \V_G, \forall t'\!.x'\in \V_H,\ x= x'\implies t=t'$;
        \item all border edges are coherent, meaning that border edges of $G$ linked to internal vertices of $H$ must also exist in $H$---i.e. for all $e\in \dangling{G}$ in between $u\in \I_G$ and $v\in \B_G\cap \I_H$, we have $e\in \dangling{H}$ and $u\in \B_H$---and symmetrically border edges of $H$ pointing towards internal vertices of $G$ must also exist in $G$.
        \item the pseudo graph $U$ defined by $\V_U = \V_H\cup \V_G$ and $\E_U = \E_H\cup \E_G$ has no cycle, and only contain finite chains.
        \item all common border vertices $v\in \B_H\cap \B_G$ have their ports not oversaturated---i.e. edges connected to $v$ in $\E_H$ use different ports than edges connected to $v$ in $\E_G$.
    \end{enumerate}
\end{lemma}
\begin{proof}
    Suppose $\I_H\cap \I_G=\emptyset$ and $H\sqcup G$ well defined. 1/ Positions cannot appear multiple times, otherwise either $H\sqcup G$ would not be a graph because of \hyperref[graph : unicity of positions]{unicity of positions}, or they would contradict $\I_H\cap \I_G=\emptyset$. 2/ Now consider a border edge $e$ of $G$, in between $u\in \I_G$ and $v\in \I_H\cap \B_G$. Then $e\in \E_{H\sqcup G}$. Since $H= (H\sqcup G)_{U}$ and $v\in \I_H$, we necessarily also have $e\in \E_H$. Since $\I_H\cap \I_G = \emptyset$, this implies $u\in \B_H$, thus $e\in \dangling{H}$. 3/ and 4/ $H\sqcup G$ being larger than both $H$ and $G$, implies that it contains all their edges and vertices, thus $3/$ is true and \hyperref[graph : port-saturation]{port-saturation} of $H\sqcup G$ implies $4/$.\\
    Suppose the four conditions above. Then the pseudo graph $K$, with $\I_K = \I_H\cup \I_G$, $\sigma_K = \sigma_G$ on $\I_G$ and $\sigma_K = \sigma_H$ on $\I_H$, is a graph such that $K_{\I_H} = H$ and $K_{\I_G}=G$. Indeed \hyperref[graph : acyclicity]{acyclicity} and \hyperref[graph : finite chains]{finite chains} come directly from $3/$, \hyperref[graph : border-attachment]{border-attachment} comes from border-attachement of $G$ and $H$, \hyperref[graph : unicity of positions]{unicity of positions} comes from $1/$, and \hyperref[graph : port-saturation]{port-saturation} comes from $4/$ and $2/$.
\end{proof}

The second condition of Lem.~\ref{lem:compatible subgraphs} implies that the transformation $AG_X$ must preserve certain border edges between $G_X$ and $G_{\overline{X}}$. Consequently, any vertex of $G_X$ connected to such an edge must also be preserved. These vertices therefore play a special role when defining local operations, so we refer to them as \emph{boundary vertices}. Figure.~\ref{fig : Induced subgraph 3} illustrates them.

\begin{definition}[Boundary, Interior]\label{def : boundary}
    We denote $X^-$ the set $\{x\in \sp{\I_G}  \;|\; \sp{\V_{G_x}} \subseteq X\}$, and call these the interior positions of $X$ in $G$. 
    Its other vertices, namely $X\setminus X^-$, are referred to as the boundary of $X$ in $G$, see Fig.~\ref{fig : Induced subgraph}.
\end{definition}

The vertices of our graphs are named. In the quantum setting, anonymous graphs~\cite{ArrighiNamesInQG} lead to superluminal signalling. 
Still, in the end the names of the vertices are just intended to describe the geometry, to be able to single out events, and nothing else. Consequently, the constructions introduced below must be independent of the particular names chosen for the vertices. We formalise changes of naming as follows.
\begin{definition}[Renaming]\label{def:renaminginvariance}
Consider a function $R:\calX\rightarrow \calV$ mapping each position $x$ to a name $R(x) = t.y$. It is a renaming whenever its action on positions $\sp{R(.)}:x\mapsto \sp{R(x)}$ is a bijection on $\calX$.
It is extended to act on any $t.y\in \calV$ as follows: $R(t.y)=t.R(y)$. 
It is extended to act upon graphs by renaming their nodes, written $R(G)$.
\end{definition}

\begin{remark}{Renaming graphs}{renaming on graphs}
    Since the induced map on positions $\sp{R(\cdot)}$ is bijective, renaming preserves all the graph conditions. Hence $G\in\calG$ implies $R(G)\in\calG$. It also preserves induced-subgraph inclusion:
    \[
        G\sqsubseteq H \implies R(G)\sqsubseteq R(H).
    \]
    Finally, whenever $G\sqcup H$ is defined, so is $R(G)\sqcup R(H)$, and
    \[
        R(G)\sqcup R(H)=R(G\sqcup H).
    \]
\end{remark}

Although the preceding constructions are defined on the full set $\calG$, applications generally involve a more restricted family of graphs $\sshift\subseteq\calG$. This family must remain stable under the operations used throughout the paper. In particular, taking an induced subgraph or renaming a graph should not leave $\sshift$. We also require closure under disjoint union, which allows us to adjoin a renamed copy with fresh vertices. We do not, however, require closure under arbitrary compatible compositions. Combined with closure under induced subgraphs, such a requirement would be extremely restrictive: arbitrarily small fragments of two graphs in $\sshift$ could be recomposed and would necessarily remain in $\sshift$, thereby preventing $\sshift$ from expressing most constraints extending beyond radius one.

\begin{definition}[Closed subset of graphs]\label{def : closure}
Consider $\sshift\subseteq\mathcal{G}$. This $\sshift$ is said to be closed under
\begin{description}[itemsep = 0.25ex, topsep = 0.5ex]
    \item[renaming\label{S : renaming}] if $G\in \sshift$ implies $R(G)\in \sshift$ where $R$ is a \hyperref[def:renaminginvariance]{renaming} (see Def.~\ref{def:renaminginvariance}).
    \item[inclusion\label{S : inclusion}] if $G\in \sshift$ implies $\forall H\sqsubseteq G,\ H\in \sshift$.
    \item[disjoint-union\label{S : disjoint-union}] if $G, H\in \sshift$ and $\V_G\cap \V_H=\varnothing$ implies $G\sqcup H\in \sshift$.
\end{description}
\end{definition}
Given $L\subseteq\calG$, we denote by $\sshift_L$ the smallest closed subset of graphs containing $L$, and say that $L$ generates $\sshift_L$. In the remainder of the paper, $\sshift$ denotes a set satisfying these three closure properties. In particular, all results apply to $\calG$, which is itself closed.


We now give a precise example of such an $\sshift$, which will serve as the domain for the main example of dynamic presented in the next sections. 

\begin{example}{Particle system---Defining $\sshift$}{PS---sshift}
    We fix $\pi = \{l,r\}$ and $\Sigma=\{0,1\}^2$. Internal states are pairs of bits representing the presence of a left-moving particle or not, and of a right-moving particle or not. Ports specify spatial directions (left or right). Let $L$ be the set of infinite ``line'' graphs, i.e. each edge goes either from left port $:\!l$ to $:\!r$ or from right port $:\!r$ to $:\!l$, thus forming an infinite line without borders. We define $\sshift$ as the closed subset of graphs generated by $L$.
\end{example}

\subsection{Neighbourhood and local rules}\label{sec: local rules}

The decomposition $G=G_X\sqcup G_{\overline{X}}$ (Ex.~\ref{ex:splitting a graph}) provides the basic mechanism for defining local transformations: isolate a region of the graph, transform it, and recompose it with the unchanged context.

We first define a \emph{neighbourhood scheme} $\restri{}$, which associates with a position $x$ and a graph $G$ a set of nearby positions $\restr{x}{G}$. We then define a local rule $A_{(-)}$ by requiring that its action at $x$ replace only the induced subgraph $G_{\restr{x}{G}}$:
\[
    A_xG=(A_xG_{\restr{x}{G}})\sqcup G_{\overline{\restr{x}{G}}}.
\]
Thus, the transformation is determined by the selected neighbourhood, while the rest of the graph remains unchanged. This construction is close in spirit to double-pushout graph rewriting~\cite{HarmerFundamentals}, although here it is expressed directly through the partial composition $\sqcup$.


\subsubsection{Neighbourhood scheme}

A neighbourhood scheme is a function yielding the position of the nodes that are considered ``closed'' to a position $x$. A neighbourhood scheme $\restri{}$ does not return an arbitrary set of positions, but those of a \emph{cone} of $x$. By a cone of $x$, we mean a connected subgraph of $G$ obtained by starting from a past node at position $x$ and exploring forward along the directed edges. Intuitively, cones represent the regions of the graph that $x$ could potentially influence. The neighbourhood scheme then chooses one such cone, denoted $G_{\restr{x}{G}}$. This idea is reminiscent of relativity theory, where an event can only signal within its future light-cone. It is also reminiscent of DAGs of dependencies between processes: a result is sent to those processes which depend upon it.


\begin{definition}[Cones]
    A \emph{(forward) cone} of $x$ is a graph $C\in\sshift$ such that $x\in\sp{\Past(C)}$ and, for every $v\in \I_C$, there exists a directed path from the vertex at position $x$ to $v$ that lies entirely in $\I_C$, a.k.a {\em accessibility}. We denote by $\xCone$ the set of cones of $x$, and by $\xcone{}$ the set of all cones.\\ 
\end{definition}

\begin{figure*}[h]
\centering
\begin{subfigure}{0.3\textwidth}
    \captionsetup{justification=centering}
    \centering
    \includegraphics[width=0.55\textwidth]{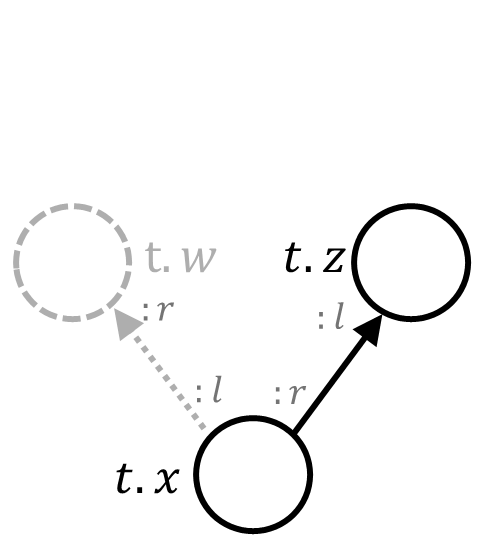}
    \caption{$\restri{x}$ can sometimes be undefined.}
    \label{fig : NS undefined}
\end{subfigure}
\begin{subfigure}{0.3\textwidth}
    \captionsetup{justification=centering}
    \centering
    \includegraphics[width=0.55\textwidth]{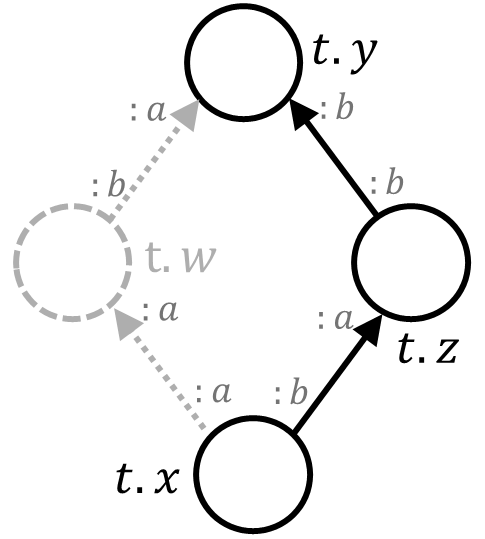}
    \caption{When it is defined it returns a cone.}
    \label{fig : cone 1}
\end{subfigure}
\begin{subfigure}{0.3\textwidth}
    \captionsetup{justification=centering}
    \centering
    \includegraphics[width=0.55\textwidth]{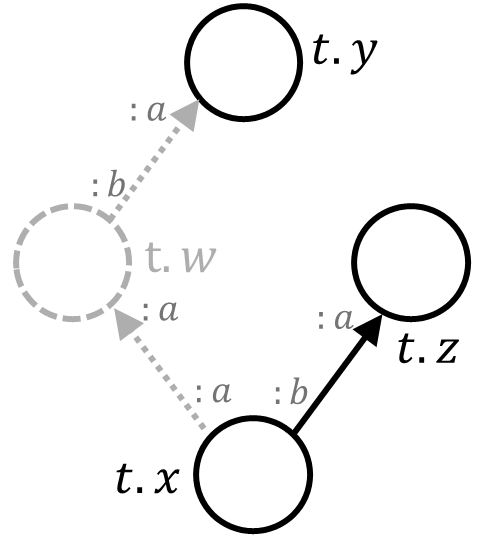}
    \caption{So this graph is never a neighbourhood.}
    \label{fig : cone 2}
\end{subfigure}
\caption{\emph{Properties of neighbourhood.} \emph{(a)} In this graph, the neighbourhood scheme of Ex.~\ref{ex:PS---NS} is undefined, as it lacks the required information: $t.x$ is too close to the border vertex $t.w$. \emph{(b)} This graph is a cone of $x$. Thus it can be the output $G_{\restri{x}}$ of some neighbourhood scheme $\restri{}$. \emph{(c)} This graph is not a cone, since the vertex $t.y$ violates the accessibility condition. Thus it is never the output $G_{\restri{x}}$ of any neighbourhood scheme $\restri{}$.}
\label{fig : examples of cones}
\end{figure*}

Thus, the neighbourhood scheme $\restri{x}$ constructs $G_{\restri{x}}$, a cone of $x$, on a domain $\graphvalidexplicit{\restri{}}{x}$ of graphs containing enough information to determine such a cone.\\
To guarantee that $\restri{x}$ is undefined only when the graph genuinely lacks information, we require \emph{stability}: once $\restri{x}$ is defined on a graph $G$, it must also be defined on all larger graphs. Lastly, we impose a condition called \emph{extensivity}, which guarantees that the neighbourhood of $x$ is indeed determined locally. We will discuss this condition in more detail later.

\begin{definition}[Neighbourhood scheme]\label{def:neighbourhood}
A \emph{(forward) neighbourhood scheme} $\restri{}$ is a function which maps any position $x\in \mathcal{X}$ to a function $\restri{x} : \graphvalidexplicit{\restri{}}{x} \to P(\mathcal{X})$ on some domain $\graphvalidexplicit{\restri{}}{x} \subseteq \sshift$ such that for any  $H \in \sshift$ and $G\in \graphvalidexplicit{\restri{}}{x}$ we have:
\begin{description}[itemsep = 0ex, topsep = 0.25ex, parsep = 0ex]
    \RX{}{\item{unambiguity:} $\forall G\in \upast{u},\ |\{x\in\X \mid x\namesubseteq\sp{u}\textrm{ and }G\in\graphvalid{\restri{}}{x}\}|\leq 1$.}
    \item[cone\label{neighbourhood : cone}] $\restr{x}{G}\subseteq\sp{\I_G}$ and $G_{\restri{x}}$ is in $\xCone$,
    \item[\closeddomain\label{neighbourhood : stability}] $G_{\restri{x}}\sqsubseteq H$ implies $H\in\graphvalid{\restri{}}{x}$.
    \item[extensivity\label{neighbourhood : extensivity}] $G_{\restri{x}} \sqsubseteq H\sqsubseteq G \text{ implies } G_{\restri{x}} = H_{\restri{x}}.$
    \item[renaming-invariance\label{neighbourhood : renaming-invariance}] for all renaming $R$, we have $R(G)\in \graphvalidexplicit{\restri{}}{\sp{R(x)}}$ and $\restri{\sp{R(x)}}(R(G))=\sp{R(\restri{x}(G))}$.
\end{description}
where we introduced $G_{\restri{x}}$ as a shorthand for $G_{\restr{x}{G}}$.\\
We denote $\diskA{x}:= \{ G_{\restri{x}}\mid G\in\graphvalidexplicit{\restri{}}{x} \}$ the \emph{set of disks} of $x$.
\end{definition}

\begin{example}{Particle system---Neighbourhood scheme}{PS---NS}
    \begin{minipage}{\linewidth}
        We start by fixing $\graphvalidexplicit{\restri{}}{x}$ to be the set of graphs in $\sshift$ such that $x\in \Past(G)$ and both neighbours of $t.x$ are internal vertices. Then for all $G\in \graphvalidexplicit{\restri{}}{x}$, we fix $\restr{x}{G}=\{y,x,z\}$ where $y$ and $z$ are the positions of the two neighbours of $t.x$. On the example graph $G$ on the right, $\restri{3}(G)=\{2,3,4\}$. The vertex $1.3$ belongs to the interior (Def.~\ref{def : boundary}) of the graph (pictured in dark blue), whilst $1.2$ and $1.4$ are boundary vertices (pictured in cyan).
    \end{minipage}
    \begin{minipage}{\linewidth}
        \centering
        \begin{tabular}{cc}
            \begin{minipage}{0.49\linewidth}
                \centering
                \includegraphics[width=\linewidth]{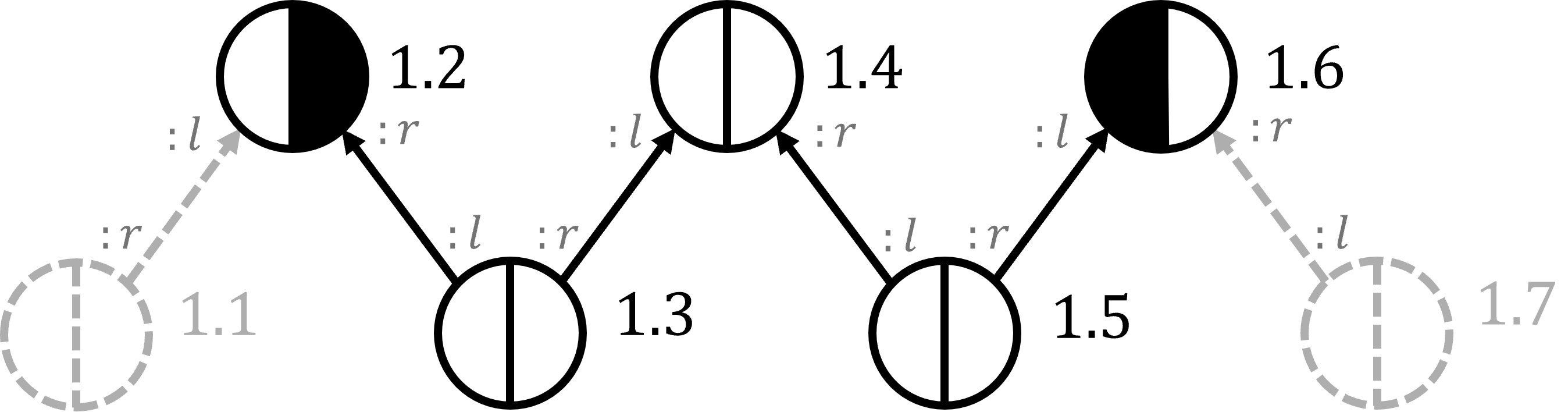}
                \captionof*{figure}{A graph $G\in \sshift$.} 
            \end{minipage}
            &
            \begin{minipage}{0.49\linewidth}
                \centering
                \includegraphics[width=0.67\linewidth]{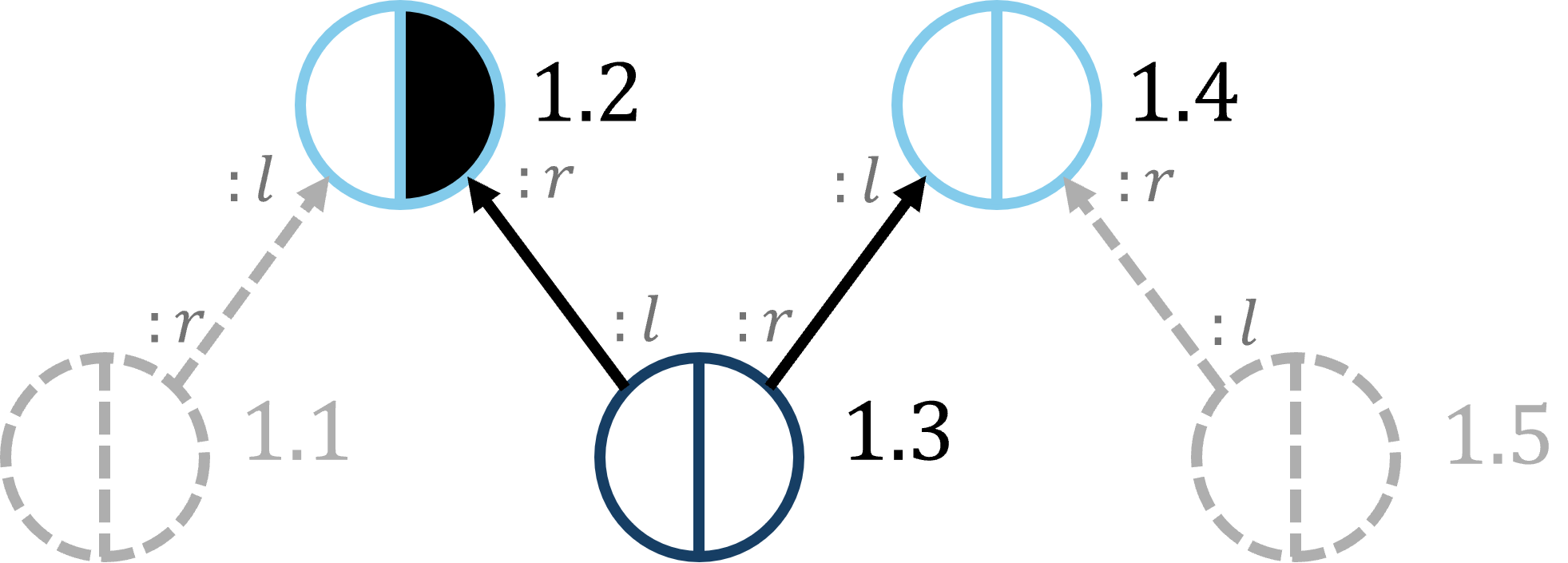}
                \captionof*{figure}{A neighbourhood $G_{\restri{3}}$.} 
            \end{minipage}
        \end{tabular}
    \end{minipage}
\end{example}

From the neighbourhood scheme we also derive a set of disks $\diskA{x}$. These disks are the neighbourhoods that can be extracted from any $G \in \sshift$, and they will serve as the basic building blocks on which the local rule is defined. Moreover, by extensivity, disks are unchanged under the action of the neighbourhood scheme, as explained in Rmk.~\ref{rem:idempotence}.

\begin{remark}{Extensivity implies idempotence}{idempotence}
    For any disk $G_{\restri{x}}\in \diskA{x}$, we have $(G_{\restri{x}})_{\restri{x}} = G_{\restri{x}}$. Indeed by \hyperref[neighbourhood : extensivity]{extensivity} we have $\restr{x}{G_{\restri{x}}}=\restr{x}{G}$ since $G_{\restri{x}}\sqsubseteq G_{\restri{x}}\sqsubseteq G$. By Rmk.~\ref{rem:induction on internal vertices}, this implies $(G_{\restri{x}})_{\restri{x}} = G_{\restri{x}}$.
\end{remark}

Notice that our definition of neighbourhood schemes is quite permissive. It even allows for unbounded neighbourhood (if the directed edges are interpreted as lightlike this still respects bounded speed of propagation of information), even though for all practical purposes one is likely to use a bounded neighbourhood scheme. Our results will apply so long as the abstract conditions are met. For some of them we need extensivity, which is there to forbid that global criteria be used to decide whether some closeby vertex belongs to the neighbourhood or not (see Ex.~\ref{ex:non extensive}). It demands that the neighbourhood $\restr{x}{G}$ as computed by the function $\restri{}$ over a graph $G$, be the same as that computed over a big enough subgraph of $G$. It can be understood as a form of graph-locality of the neighbourhood scheme $\restri{}$ itself. Typically, if $\restr{x}{G}$ is computed step by step starting from $x$ until a suitable disk is found, and no preferred disk is present, then it is extensive (see Ex.~\ref{ex:infinite neighbourhood}). 

\begin{example}{A non extensive neighbourhood}{non extensive}
    \begin{minipage}{\linewidth}
        Consider the non-local function $\restri{x}$ defined as follows: if it is applied to a graph $G$ that does not contain a black node, it returns just $x$. If it is applied to a graph $H$ that does contain a black node, it instead returns $x$ together with its two neighbouring positions ($y$ and $w$). This function is non-local because the size of the neighbourhood depends on the state of the unconnected node $t.z$. The issue arises because $\restri{}$ is not extensive. Indeed, consider $H_{\restri{x}}$. Since it does not contain any black node, we should have $\restr{x}{H_{\restri{x}}} = x$, which contradicts idempotency (Rmk.~\ref{rem:idempotence}).
    \end{minipage}\\
    \begin{minipage}{\linewidth}
        \begin{tabular}{cc}
            \begin{minipage}{0.49\linewidth}
                \centering
                \captionsetup{justification=centering}
                \includegraphics[width=0.6\linewidth]{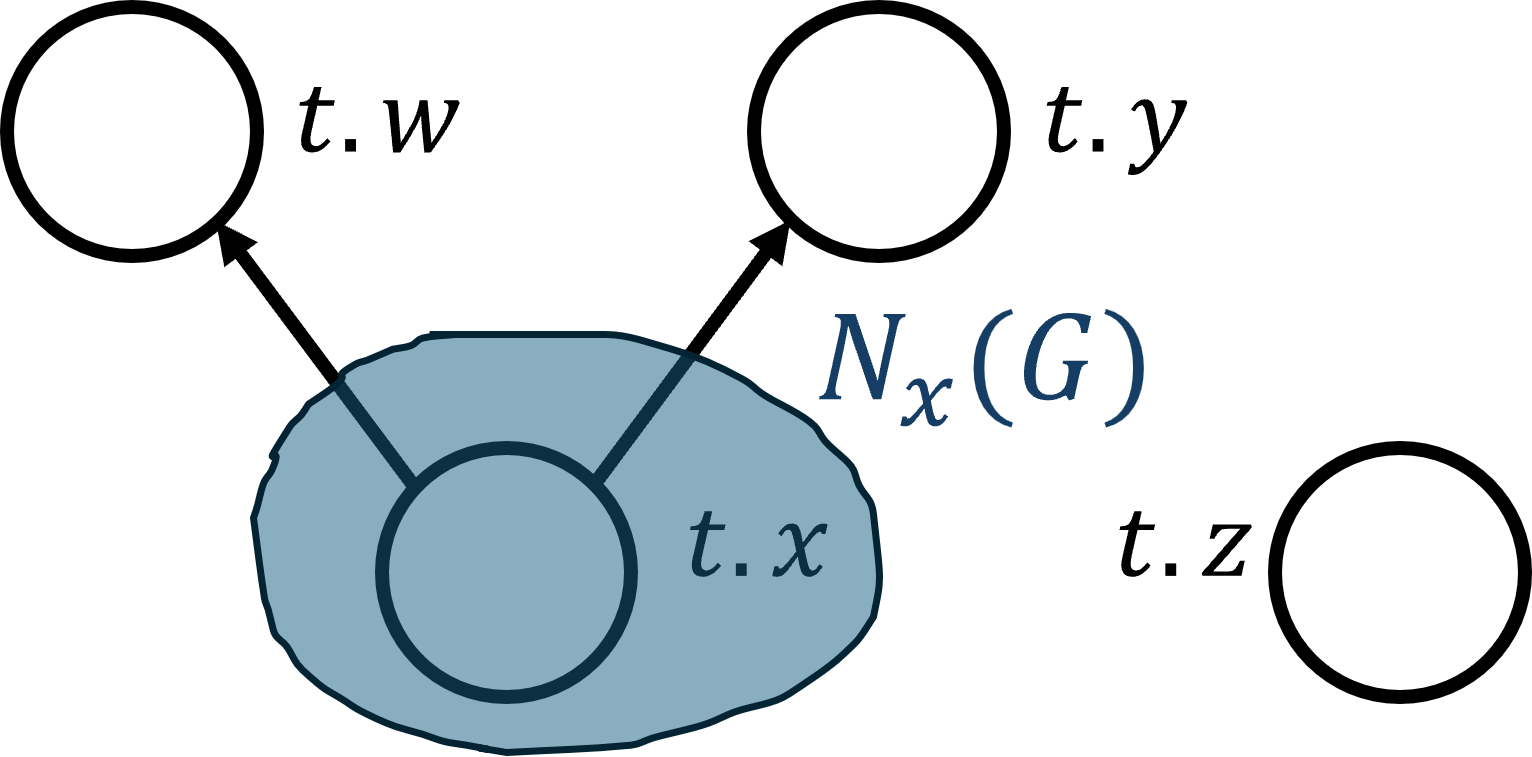}
                \captionof{figure}{$G$ and $\restr{x}{G}$.}      
            \end{minipage}
            &
            \begin{minipage}{0.49\linewidth}
                \centering
                \captionsetup{justification=centering}
                \includegraphics[width=0.6\linewidth]{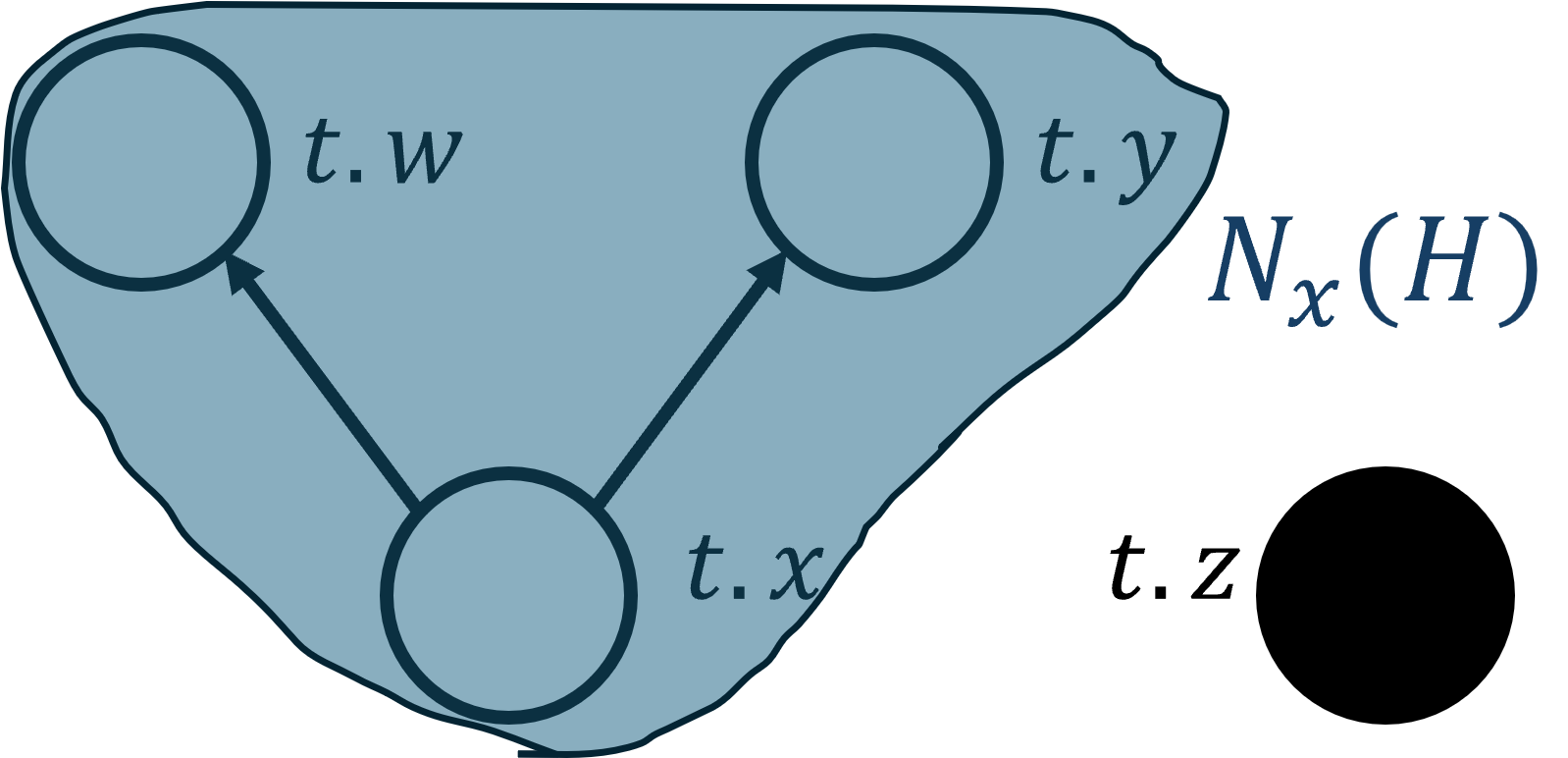}
                \captionof{figure}{$H$ and $\restr{x}{H}$}
            \end{minipage}
        \end{tabular}
    \end{minipage}
\end{example}

\begin{example}{An extensive unbounded neighbourhood}{infinite neighbourhood}
    \begin{minipage}{\linewidth}
        Here we work on $\sshift$ defined in Ex.~\ref{ex:PS---sshift}.
        Given a graph $G$ and $x \in \Past(G)$, the function $\restri{x}(G)$ returns the set of positions of all vertices reached by the following exploration:
        start from $x$ and move forward along directed edges, layer by layer (distance $0$, $1$, $2$, …), stopping when a “dead end” is reached — that is, when the next edge points in the opposite direction (see $e=(1.8\!:\!l, 2.7\!:\!r)$), or when it has reached a border (see $e'=(2.3\!:\!l, 2.2\!:\!r)$).
        This function is extensive, it outputs the same positions in any graph $H$ such that $G_{\restri{5}}\sqsubseteq H \sqsubseteq G$, because the exploration in $H$ stops at the same vertices on both sides. On the right side, it is because $G_{\restri{5}}\sqsubseteq H $, thus we still have the edge $e$ in the opposite direction. On the left side $H \sqsubseteq G$ enforces $2.2\in \B_H$, thus $e'$ is still a border edge.
    \end{minipage}\\
    \begin{minipage}{\linewidth}
        \begin{tabular}{cc}
            \begin{minipage}{0.49\linewidth}
                \centering
                \captionsetup{justification=centering}
                \includegraphics[width=\linewidth]{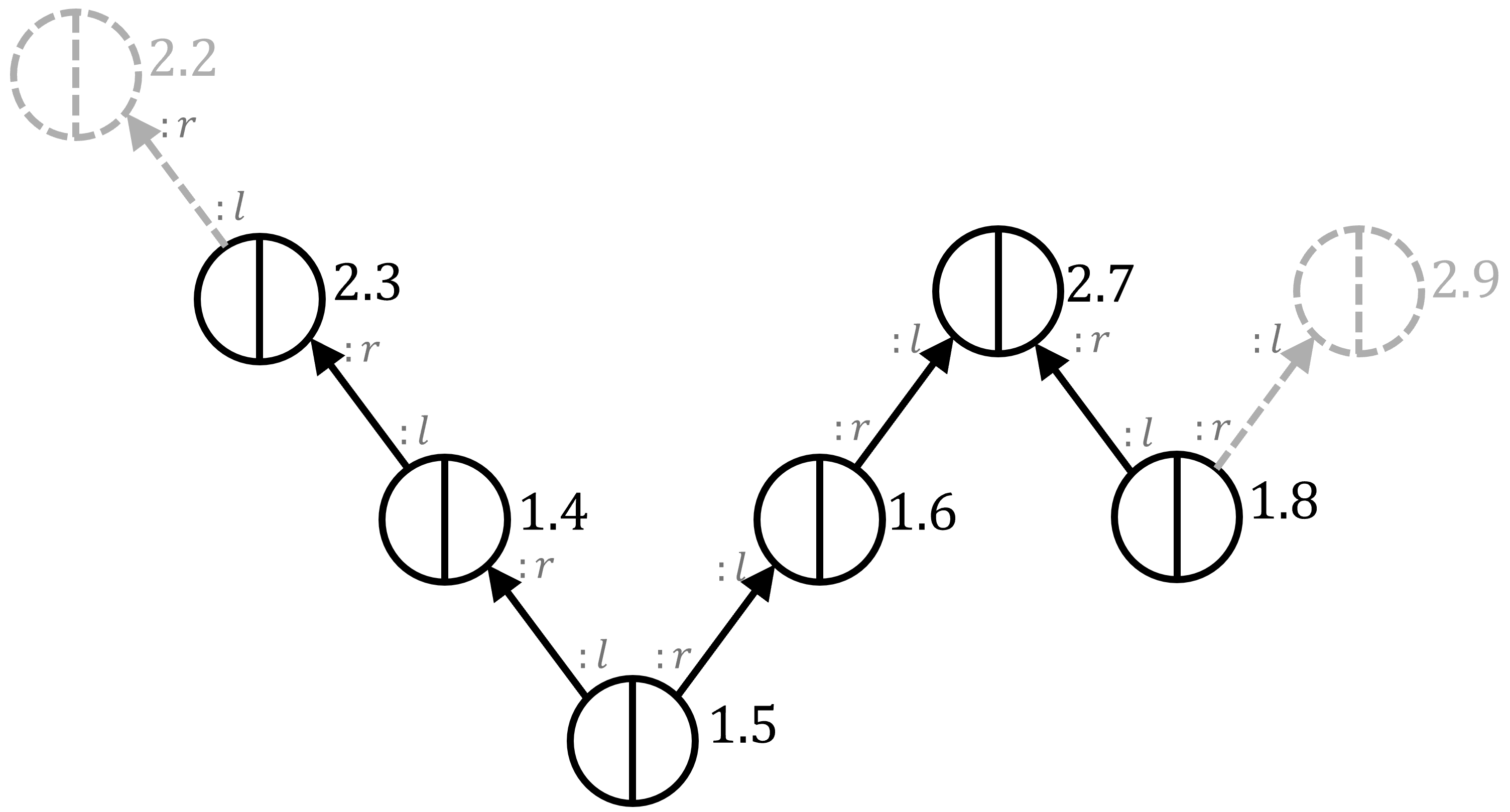}
                \captionof{figure}{$G$.}         
            \end{minipage}
            &
            \begin{minipage}{0.49\linewidth}
                \centering
                \captionsetup{justification=centering}
                \includegraphics[width=\linewidth]{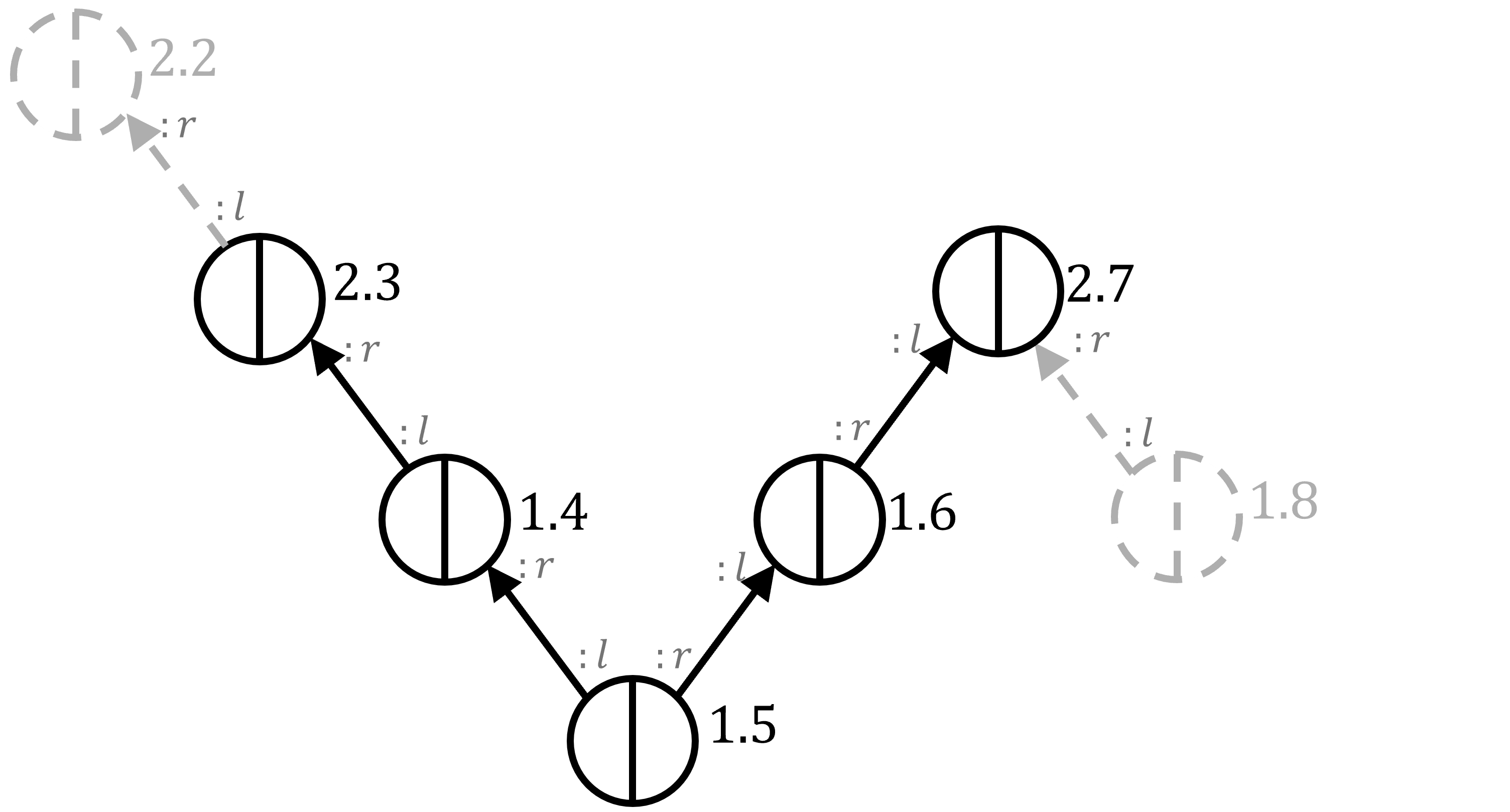}
                \captionof{figure}{$G_{\restri{5}}$.}
            \end{minipage}
        \end{tabular}
    \end{minipage}
\end{example}


\subsubsection{Local rule}\label{sec:local rule}

Next we define the notion of a local rule $A_{(-)}$. Before we proceed, remember that our graphs can never contain both vertices $u=t.x$ and $u'=(t+1).x$. It follows that the local rule $A_x$ will act unambiguously on the unique vertex of the form $u=t.x$ that is present in the graph. Keep in mind also that our directed edges stand for dependencies between events, i.e. if $u=t.x$ points to $v=t'.y$, then $v$ is considered ahead in time of $u$, frozen awaiting information from $u$. The action of some local rule $A_y$ on $v$ is therefore trivial, preventing $y$ to be computed too far ahead and thereby providing a weak synchronisation mechanism. Moreover, the action of $A_x$ on $u$ is non-trivial only if $\restri{x}$ is defined on the graph $G$ considered, capturing the fact that enough information is present. This forces $u$ to be at a past position. Such a vertex can be thought of as lagging behind in time, and no longer awaiting any information---it has reached its local normal form. The action of $A_x$ is to dispose of it possibly modifying $v$ and other dependencies, and creating vertex $u'=t'.x$ in some provisional state.

\reduced{
\begin{definition}[Local rule]\label{def : locality}
    Given some forward neighbourhood scheme $\restri{}$, an \emph{$\restri{}$-local operator} $A_{(-)}$ is a renaming-invariant operator from $\mathcal{X}\times\sshift$ to $\sshift$ such that $\forall x\in\X$, 
    $\forall G \in \sshift,$
    $$A_x G := 
    \begin{cases}
    (A_x G_\restri{x}) \sqcup G_\crestri{x}&\textrm{if $G\in \graphvalid{\restri{}}{x}$}\\
    G&\textrm{otherwise}
    \end{cases}
    $$
    where the action of $A_x$ on $G_{\restri{x}}$ must be ``context-preserving''---i.e. it preserves the border $\B_{G}= \B_{A_xG}$ and its border edges $e\in \dangling{G} \implies e \in \E_{A_x G}$.
\end{definition}
}{
\begin{definition}[Local rule]\label{def : locality}
    Given some forward neighbourhood scheme $\restri{}$, an \emph{$\restri{}$-local operator} $A_{(-)}$ is an operator from $\mathcal{X}\times\sshift$ to $\sshift$ such that $\forall x\in\X$ and $\forall G \in \sshift$, 
    $A_{\sp{R(x)}}(R(G))=R(A_x G)$ for every renaming $R$, and
    $$A_x G := 
    \begin{cases}
    (A_x G_\restri{x}) \sqcup G_\crestri{x}&\textrm{if $G\in \graphvalid{\restri{}}{x}$}\\
    G&\textrm{otherwise}
    \end{cases}
    $$
    where the action of $A_x$ on $G_{\restri{x}}$ must preserve the border---i.e. $\B_{G_{\restri{x}}}= \B_{A_xG_{\restri{x}}}$.
\end{definition}
}
\begin{remark}{$A_{(-)}$ is determined by its action on $\diskA{x}$}{local rule and disks}
    The definition above may seem inductive, as $A$ appears on both sides of the equation. However, by idempotence (Rmk.~\ref{rem:idempotence}) of the neighbourhood scheme, this induction stops after the first step. Thus the action of $A_x$ is fully determined by its effect on the set of disks $\diskA{x}$. Hence, we will always define local rules only on this set of graphs.
\end{remark}
\begin{example}{Particle system---Local rule}{PS---LR}
        $A_x$ acts by consuming the internal state $i=\sigma^r_G(v)$ of vertex $v$ (and symmetrically with $j=\sigma^l_G(w)$), thereby moving those particles at $x$ if they are present. It also updates ports 
        , flips the arrows pointing to $x$, and increments its timetag, in order to move the vertex from past $u=t.x$ to future $u'= 1.u = (t+1).x$. Dashed edges and vertices do not influence the local rule.
    \begin{tabular}{cc}
        \begin{minipage}{0.49\linewidth}
            \centering
            \captionsetup{justification=centering}
            \resizebox{0.5\linewidth}{!}{\tikzset{every picture/.style={line width=0.75pt}} 

\begin{tikzpicture}[x=0.75pt,y=0.75pt,yscale=-1,xscale=1]

\draw [color={rgb, 255:red, 74; green, 74; blue, 74 }  ,draw opacity=1 ][line width=2.25]  [dash pattern={on 2.53pt off 3.02pt}]  (276,176) -- (321.96,104.74) ;
\draw [shift={(324.67,100.53)}, rotate = 122.82] [fill={rgb, 255:red, 74; green, 74; blue, 74 }  ,fill opacity=1 ][line width=0.08]  [draw opacity=0] (14.29,-6.86) -- (0,0) -- (14.29,6.86) -- cycle    ;
\draw [color={rgb, 255:red, 74; green, 74; blue, 74 }  ,draw opacity=1 ][line width=2.25]  [dash pattern={on 2.53pt off 3.02pt}]  (176,176) -- (126.52,104.64) ;
\draw [shift={(123.67,100.53)}, rotate = 55.26] [fill={rgb, 255:red, 74; green, 74; blue, 74 }  ,fill opacity=1 ][line width=0.08]  [draw opacity=0] (14.29,-6.86) -- (0,0) -- (14.29,6.86) -- cycle    ;
\draw [color={rgb, 255:red, 74; green, 74; blue, 74 }  ,draw opacity=1 ][line width=2.25]  [dash pattern={on 2.53pt off 3.02pt}]  (116,284) -- (162.19,202.88) ;
\draw [shift={(164.67,198.53)}, rotate = 119.66] [fill={rgb, 255:red, 74; green, 74; blue, 74 }  ,fill opacity=1 ][line width=0.08]  [draw opacity=0] (14.29,-6.86) -- (0,0) -- (14.29,6.86) -- cycle    ;
\draw [color={rgb, 255:red, 0; green, 0; blue, 0 }  ,draw opacity=1 ][line width=2.25]    (226,274) -- (189.91,202) ;
\draw [shift={(187.67,197.53)}, rotate = 63.38] [fill={rgb, 255:red, 0; green, 0; blue, 0 }  ,fill opacity=1 ][line width=0.08]  [draw opacity=0] (14.29,-6.86) -- (0,0) -- (14.29,6.86) -- cycle    ;
\draw [color={rgb, 255:red, 74; green, 74; blue, 74 }  ,draw opacity=1 ][line width=2.25]  [dash pattern={on 2.53pt off 3.02pt}]  (335,285) -- (291.05,203.93) ;
\draw [shift={(288.67,199.53)}, rotate = 61.54] [fill={rgb, 255:red, 74; green, 74; blue, 74 }  ,fill opacity=1 ][line width=0.08]  [draw opacity=0] (14.29,-6.86) -- (0,0) -- (14.29,6.86) -- cycle    ;
\draw [line width=2.25]    (226,274) -- (264.27,204.09) ;
\draw [shift={(266.67,199.7)}, rotate = 118.69] [fill={rgb, 255:red, 0; green, 0; blue, 0 }  ][line width=0.08]  [draw opacity=0] (14.29,-6.86) -- (0,0) -- (14.29,6.86) -- cycle    ;
\draw  [color={rgb, 255:red, 74; green, 74; blue, 74 }  ,draw opacity=1 ][fill={rgb, 255:red, 255; green, 255; blue, 255 }  ,fill opacity=1 ][dash pattern={on 3.38pt off 3.27pt}][line width=3]  (320,285) .. controls (320,276.72) and (326.72,270) .. (335,270) .. controls (343.28,270) and (350,276.72) .. (350,285) .. controls (350,293.28) and (343.28,300) .. (335,300) .. controls (326.72,300) and (320,293.28) .. (320,285) -- cycle ;
\draw  [color={rgb, 255:red, 74; green, 74; blue, 74 }  ,draw opacity=1 ][fill={rgb, 255:red, 255; green, 255; blue, 255 }  ,fill opacity=1 ][dash pattern={on 3.38pt off 3.27pt}][line width=3]  (102,284) .. controls (102,276.27) and (108.27,270) .. (116,270) .. controls (123.73,270) and (130,276.27) .. (130,284) .. controls (130,291.73) and (123.73,298) .. (116,298) .. controls (108.27,298) and (102,291.73) .. (102,284) -- cycle ;
\draw  [color={rgb, 255:red, 38; green, 105; blue, 185 }  ,draw opacity=1 ][fill={rgb, 255:red, 255; green, 255; blue, 255 }  ,fill opacity=1 ][line width=3]  (200,274) .. controls (200,259.64) and (211.64,248) .. (226,248) .. controls (240.36,248) and (252,259.64) .. (252,274) .. controls (252,288.36) and (240.36,300) .. (226,300) .. controls (211.64,300) and (200,288.36) .. (200,274) -- cycle ;
\draw  [color={rgb, 255:red, 80; green, 227; blue, 194 }  ,draw opacity=1 ][fill={rgb, 255:red, 255; green, 255; blue, 255 }  ,fill opacity=1 ][line width=3]  (152,176) .. controls (152,162.75) and (162.75,152) .. (176,152) .. controls (189.25,152) and (200,162.75) .. (200,176) .. controls (200,189.25) and (189.25,200) .. (176,200) .. controls (162.75,200) and (152,189.25) .. (152,176) -- cycle ;
\draw  [color={rgb, 255:red, 80; green, 227; blue, 194 }  ,draw opacity=1 ][fill={rgb, 255:red, 255; green, 255; blue, 255 }  ,fill opacity=1 ][line width=3]  (252,176) .. controls (252,162.75) and (262.75,152) .. (276,152) .. controls (289.25,152) and (300,162.75) .. (300,176) .. controls (300,189.25) and (289.25,200) .. (276,200) .. controls (262.75,200) and (252,189.25) .. (252,176) -- cycle ;
\draw  [color={rgb, 255:red, 74; green, 74; blue, 74 }  ,draw opacity=1 ][fill={rgb, 255:red, 255; green, 255; blue, 255 }  ,fill opacity=1 ][dash pattern={on 3.38pt off 3.27pt}][line width=3]  (100,85) .. controls (100,76.72) and (106.72,70) .. (115,70) .. controls (123.28,70) and (130,76.72) .. (130,85) .. controls (130,93.28) and (123.28,100) .. (115,100) .. controls (106.72,100) and (100,93.28) .. (100,85) -- cycle ;
\draw  [color={rgb, 255:red, 74; green, 74; blue, 74 }  ,draw opacity=1 ][fill={rgb, 255:red, 255; green, 255; blue, 255 }  ,fill opacity=1 ][dash pattern={on 3.38pt off 3.27pt}][line width=3]  (320,85) .. controls (320,76.72) and (326.72,70) .. (335,70) .. controls (343.28,70) and (350,76.72) .. (350,85) .. controls (350,93.28) and (343.28,100) .. (335,100) .. controls (326.72,100) and (320,93.28) .. (320,85) -- cycle ;
\draw [color={rgb, 255:red, 80; green, 227; blue, 194 }  ,draw opacity=1 ][line width=3]    (176,152) -- (176,200) ;
\draw [color={rgb, 255:red, 80; green, 227; blue, 194 }  ,draw opacity=1 ][line width=3]    (276,152) -- (276,200) ;
\draw [color={rgb, 255:red, 74; green, 107; blue, 226 }  ,draw opacity=1 ][line width=3]    (226,250) -- (226,298) ;

\draw (250,288.4) node [anchor=north west][inner sep=0.75pt]  [font=\sizeFfour]  {$\mathbf{u=t.x}$};
\draw (130,165) node [anchor=north west][inner sep=0.75pt]  [font=\sizeFfour]  {$\mathbf{v}$};
\draw (300,165) node [anchor=north west][inner sep=0.75pt]  [font=\sizeFfour]  {$\mathbf{w}$};
\draw (200,201.4) node [anchor=north west][inner sep=0.75pt]  [font=\sizeFfourp,color={rgb, 255:red, 155; green, 155; blue, 155 }  ,opacity=1 ]  {$:r$};
\draw (230,200.4) node [anchor=north west][inner sep=0.75pt]  [font=\sizeFfourp,color={rgb, 255:red, 155; green, 155; blue, 155 }  ,opacity=1 ]  {$:l$};
\draw (246,232.4) node [anchor=north west][inner sep=0.75pt]  [font=\sizeFfourp,color={rgb, 255:red, 155; green, 155; blue, 155 }  ,opacity=1 ]  {$:r$};
\draw (184,232.4) node [anchor=north west][inner sep=0.75pt]  [font=\sizeFfourp,color={rgb, 255:red, 155; green, 155; blue, 155 }  ,opacity=1 ]  {$:l$};
\draw (181,162) node [anchor=north west][inner sep=0.75pt]   [align=left] {\textbf{{\sizeFfours i}}};
\draw (261,162) node [anchor=north west][inner sep=0.75pt]   [align=left] {\textbf{{\sizeFfours j}}};

\end{tikzpicture}}
            \captionof{figure}{$G_{\restri{x_1}}$.}
            \label{fig : Dynamic example 1; G}            
        \end{minipage}
        &
        \begin{minipage}{0.49\linewidth}
            \centering
            \resizebox{0.5\linewidth}{!}{\tikzset{every picture/.style={line width=0.75pt}} 

\begin{tikzpicture}[x=0.75pt,y=0.75pt,yscale=-1,xscale=1]

\draw [color={rgb, 255:red, 74; green, 74; blue, 74 }  ,draw opacity=1 ][line width=2.25]  [dash pattern={on 2.53pt off 3.02pt}]  (276,176) -- (323.91,103.6) ;
\draw [shift={(326.67,99.43)}, rotate = 123.49] [fill={rgb, 255:red, 74; green, 74; blue, 74 }  ,fill opacity=1 ][line width=0.08]  [draw opacity=0] (14.29,-6.86) -- (0,0) -- (14.29,6.86) -- cycle    ;
\draw [color={rgb, 255:red, 74; green, 74; blue, 74 }  ,draw opacity=1 ][line width=2.25]  [dash pattern={on 2.53pt off 3.02pt}]  (176,176) -- (118.13,103.9) ;
\draw [shift={(115,100)}, rotate = 51.25] [fill={rgb, 255:red, 74; green, 74; blue, 74 }  ,fill opacity=1 ][line width=0.08]  [draw opacity=0] (14.29,-6.86) -- (0,0) -- (14.29,6.86) -- cycle    ;
\draw [color={rgb, 255:red, 74; green, 74; blue, 74 }  ,draw opacity=1 ][line width=2.25]  [dash pattern={on 2.53pt off 3.02pt}]  (116,284) -- (162.14,205.08) ;
\draw [shift={(164.67,200.77)}, rotate = 120.31] [fill={rgb, 255:red, 74; green, 74; blue, 74 }  ,fill opacity=1 ][line width=0.08]  [draw opacity=0] (14.29,-6.86) -- (0,0) -- (14.29,6.86) -- cycle    ;
\draw [color={rgb, 255:red, 0; green, 0; blue, 0 }  ,draw opacity=1 ][line width=2.25]    (276,176) -- (239.95,105.88) ;
\draw [shift={(237.67,101.43)}, rotate = 62.79] [fill={rgb, 255:red, 0; green, 0; blue, 0 }  ,fill opacity=1 ][line width=0.08]  [draw opacity=0] (14.29,-6.86) -- (0,0) -- (14.29,6.86) -- cycle    ;
\draw [color={rgb, 255:red, 74; green, 74; blue, 74 }  ,draw opacity=1 ][line width=2.25]  [dash pattern={on 2.53pt off 3.02pt}]  (335,285) -- (292.34,204.42) ;
\draw [shift={(290,200)}, rotate = 62.1] [fill={rgb, 255:red, 74; green, 74; blue, 74 }  ,fill opacity=1 ][line width=0.08]  [draw opacity=0] (14.29,-6.86) -- (0,0) -- (14.29,6.86) -- cycle    ;
\draw [line width=2.25]    (176,176) -- (212.32,107.85) ;
\draw [shift={(214.67,103.43)}, rotate = 118.05] [fill={rgb, 255:red, 0; green, 0; blue, 0 }  ][line width=0.08]  [draw opacity=0] (14.29,-6.86) -- (0,0) -- (14.29,6.86) -- cycle    ;
\draw  [color={rgb, 255:red, 74; green, 74; blue, 74 }  ,draw opacity=1 ][fill={rgb, 255:red, 255; green, 255; blue, 255 }  ,fill opacity=1 ][dash pattern={on 3.38pt off 3.27pt}][line width=3]  (320,285) .. controls (320,276.72) and (326.72,270) .. (335,270) .. controls (343.28,270) and (350,276.72) .. (350,285) .. controls (350,293.28) and (343.28,300) .. (335,300) .. controls (326.72,300) and (320,293.28) .. (320,285) -- cycle ;
\draw  [color={rgb, 255:red, 74; green, 74; blue, 74 }  ,draw opacity=1 ][fill={rgb, 255:red, 255; green, 255; blue, 255 }  ,fill opacity=1 ][dash pattern={on 3.38pt off 3.27pt}][line width=3]  (102,284) .. controls (102,276.27) and (108.27,270) .. (116,270) .. controls (123.73,270) and (130,276.27) .. (130,284) .. controls (130,291.73) and (123.73,298) .. (116,298) .. controls (108.27,298) and (102,291.73) .. (102,284) -- cycle ;
\draw  [color={rgb, 255:red, 38; green, 105; blue, 185 }  ,draw opacity=1 ][fill={rgb, 255:red, 255; green, 255; blue, 255 }  ,fill opacity=1 ][line width=3]  (200,76) .. controls (200,61.64) and (211.64,50) .. (226,50) .. controls (240.36,50) and (252,61.64) .. (252,76) .. controls (252,90.36) and (240.36,102) .. (226,102) .. controls (211.64,102) and (200,90.36) .. (200,76) -- cycle ;
\draw  [color={rgb, 255:red, 80; green, 227; blue, 194 }  ,draw opacity=1 ][fill={rgb, 255:red, 255; green, 255; blue, 255 }  ,fill opacity=1 ][line width=3]  (152,176) .. controls (152,162.75) and (162.75,152) .. (176,152) .. controls (189.25,152) and (200,162.75) .. (200,176) .. controls (200,189.25) and (189.25,200) .. (176,200) .. controls (162.75,200) and (152,189.25) .. (152,176) -- cycle ;
\draw  [color={rgb, 255:red, 80; green, 227; blue, 194 }  ,draw opacity=1 ][fill={rgb, 255:red, 255; green, 255; blue, 255 }  ,fill opacity=1 ][line width=3]  (252,176) .. controls (252,162.75) and (262.75,152) .. (276,152) .. controls (289.25,152) and (300,162.75) .. (300,176) .. controls (300,189.25) and (289.25,200) .. (276,200) .. controls (262.75,200) and (252,189.25) .. (252,176) -- cycle ;
\draw  [color={rgb, 255:red, 74; green, 74; blue, 74 }  ,draw opacity=1 ][fill={rgb, 255:red, 255; green, 255; blue, 255 }  ,fill opacity=1 ][dash pattern={on 3.38pt off 3.27pt}][line width=3]  (100,85) .. controls (100,76.72) and (106.72,70) .. (115,70) .. controls (123.28,70) and (130,76.72) .. (130,85) .. controls (130,93.28) and (123.28,100) .. (115,100) .. controls (106.72,100) and (100,93.28) .. (100,85) -- cycle ;
\draw  [color={rgb, 255:red, 74; green, 74; blue, 74 }  ,draw opacity=1 ][fill={rgb, 255:red, 255; green, 255; blue, 255 }  ,fill opacity=1 ][dash pattern={on 3.38pt off 3.27pt}][line width=3]  (320,85) .. controls (320,76.72) and (326.72,70) .. (335,70) .. controls (343.28,70) and (350,76.72) .. (350,85) .. controls (350,93.28) and (343.28,100) .. (335,100) .. controls (326.72,100) and (320,93.28) .. (320,85) -- cycle ;
\draw [color={rgb, 255:red, 80; green, 227; blue, 194 }  ,draw opacity=1 ][line width=3]    (176,152) -- (176,200) ;
\draw [color={rgb, 255:red, 80; green, 227; blue, 194 }  ,draw opacity=1 ][line width=3]    (276,152) -- (276,200) ;
\draw [color={rgb, 255:red, 74; green, 107; blue, 226 }  ,draw opacity=1 ][line width=3]    (226,50) -- (226,98) ;

\draw (241,92.4) node [anchor=north west][inner sep=0.75pt]  [font=\sizeFfour]  {$\mathbf{( t+1) .x}$};
\draw (130,165) node [anchor=north west][inner sep=0.75pt]  [font=\sizeFfour]  {$\mathbf{v}$};
\draw (300,165) node [anchor=north west][inner sep=0.75pt]  [font=\sizeFfour]  {$\mathbf{w}$};
\draw (251,112.4) node [anchor=north west][inner sep=0.75pt]  [font=\sizeFfourp,color={rgb, 255:red, 155; green, 155; blue, 155 }  ,opacity=1 ]  {$:r$};
\draw (179,101.4) node [anchor=north west][inner sep=0.75pt]  [font=\sizeFfourp,color={rgb, 255:red, 155; green, 155; blue, 155 }  ,opacity=1 ]  {$:l$};
\draw (195,141.4) node [anchor=north west][inner sep=0.75pt]  [font=\sizeFfourp,color={rgb, 255:red, 155; green, 155; blue, 155 }  ,opacity=1 ]  {$:r$};
\draw (236,141.4) node [anchor=north west][inner sep=0.75pt]  [font=\sizeFfourp,color={rgb, 255:red, 155; green, 155; blue, 155 }  ,opacity=1 ]  {$:l$};
\draw (231,62) node [anchor=north west][inner sep=0.75pt]   [align=left] {\textbf{{\sizeFfours i}}};
\draw (211,62) node [anchor=north west][inner sep=0.75pt]   [align=left] {\textbf{{\sizeFfours j}}};

\end{tikzpicture}}
            \captionof{figure}{$A_{x_1}G_{\restri{x_1}}$.}
            \label{fig : Dynamic example 1; A_x G}
        \end{minipage}
    \end{tabular}
\end{example}
In the definition of a local rule, we require the local action on disks to be border-preserving---that is, to preserve border \reduced{edges}{vertices}. This is a standard assumption in the graph rewriting community \cite{HarmerFundamentals}, and it rules out pathological interface behaviours. As we make precise below, once we restrict to ``large enough'' graphs, this also enforces preservation of boundary positions (Def.~\ref{def : boundary}).\\
By requiring that a local rule is a function $A_x : \sshift \to \sshift$, we implicitly assume that $(A_x G_{\restri{x}}) \sqcup G_{\crestri{x}}$ is well defined for all $G \in \graphvalid{\restri{}}{x}$. In particular, the action of $A_x$ on $G_{\restri{x}}$ cannot create positions already present in $G_{\crestri{x}}$. Combined with renaming-invariance, this ensures that no new positions appear in $G_{\restri{x}}$, as stated in Lem.~\ref{lem:easy name preservation}. Hence $\I_{A_x G_{\restri{x}}} \cap \I_{G_{\crestri{x}}} = \emptyset$, and we may apply Lem.~\ref{lem:compatible subgraphs} to better understand what the well-definedness of the $\sqcup$ operation actually enforces.\\
To derive stronger constraints, we then restrict to a broad family of well-behaved graphs that are ``large enough'', namely graphs $G$ such that there exists $H\in\sshift$ with $\V_G\subseteq \I_H$ and $H_{\restri{x}}=G$. In that case, border vertices of $G$ lie in the context part of $H$, so the four conditions of Lem.~\ref{lem:compatible subgraphs} become informative at the interface. Condition (1) implies no more than name preservation (Lem.~\ref{lem:easy name preservation}). Condition (2) \reduced{is redundant with the stronger assumption of context-preservation.}{ yields boundary-position and dangling-edge preservation (Lem.~\ref{lem:border edges are preserved}).
} Condition (3) implies that the global action $A_x$ on $G$ must preserve acyclicity. Condition (4) adds nothing, since vertices belonging both to $\B_{A_x G_{\restri{x}}}= \B_{G_{\restri{x}}}$ and $\B_{G_{\crestri{x}}}$ already use different ports because of the \hyperref[graph : port-saturation]{port-saturation} of $G_{\restri{x}}\sqcup G_{\crestri{x}}$.

\reduced{
\begin{remark}{Only names of interior vertices are modified}{boundaries are stable}
    $A_x$ does not change the name of boundary vertices---i.e. for all $t.y\in \I_{G}$ we have $y\in \restri{x}\setminus \restri{x}^-$ implies $t.y\in \I_{A_x G}$. Indeed by definition such a boundary vertex $t.y$ is the extremity of a border edge $e$. Since by context-preservation $e\in \E_{A_xG}$, we also have $t.y\in \I_{A_x G}$.
\end{remark}}
{
\begin{lemma}{Position preservation}{easy name preservation}
  Let $A_x:{\sshift}\to {\sshift}$ be a local rule, then $\sp{\I_{A_x G}} \subseteq \sp{\I_{G}}$.
\end{lemma}
\begin{proof}
    We proceed by contradiction. Consider $G\in  \graphvalid{\restri{}}{x}$, and $y\in \sp{\I_{A_x G}}\setminus\sp{\I_{G}}$. 
    Construct a graph $G'\in \sshift$ such that $ \I_G\cup \{t.y\} \subseteq \I_{G'}$ and $G\sqsubseteq G'$. Such a graph can exist since $y\notin\sp{\V_{G}}=\sp{\I_G}\cup \sp{\B_G}$ (by border preservation we have $y\notin \sp{\B_G}$), but we must make sure that $G\in \sshift$. This can be done by first constructing $R(G)$ a copy of $G$ which contains only fresh names, including $t.y$. This $R(G)\in \sshift$ since $\sshift$ is closed under \hyperref[S : renaming]{renaming}. Second, since $\sshift$ is closed under \mbox{\hyperref[S : disjoint-union]{disjoint union}}, we can take $G'= G\sqcup R(G)$. Note that by \hyperref[neighbourhood : stability]{stability} of $\restri{}$ we have $G'\in \graphvalid{\restri{}}{x}$. Moreover we have $G'_{\restri{x}}\sqsubseteq G$ since the neighbourhood cannot select some unconnected nodes in $R(G)$, so by extensivity, $G'_\restri{x}=G_\restri{x}$, from which it follows that $A_xG'=(A_x G_\restri{x}) \sqcup G'_\crestri{x}$. But this union is undefined as $y\in \sp{\I_{A_x G_\restri{x}}}$ and $t.y\in G'_\crestri{x}$ violates the \hyperref[graph : unicity of positions]{unicity of positions} condition.  This proves $\sp{\I_{A_xG}} \subseteq \sp{\I_{G}}$.
\end{proof}

\begin{lemma}{Boundary preservation}{border edges are preserved}
    For any graph $G$ such that $\exists H\in\sshift,\ \V_G\subseteq \I_H,\ H_{\restri{x}}=G$, $A_x$ preserves boundary vertices and border edges---i.e.
    \begin{align*}
        \forall u\in \I_G,\sp{u}\in\restri{x}\setminus \restri{x}^- \implies u\in \I_{A_xG} && \dangling{G} = \dangling{A_x G}
    \end{align*}
\end{lemma}
\begin{proof}
    Fix $H\in\sshift$ such that $\V_G\subseteq \I_H$ and $H_{\restri{x}}=G$. Then $A_xH=(A_xG)\sqcup H_{\crestri{x}}$ by locality of $A_x$.\\
    \textbf{$\forall u\in \I_G,\sp{u}\in\restri{x}\setminus \restri{x}^- \implies u\in \I_{A_xG}$.} Consider $t.y\in \I_G$ such that $y\in\restri{x}\setminus \restri{x}^-$. There exists a border edge $e$ linking $t.y\in \I_G$ to some border vertex $t'.z\in \B_G\subseteq \I_H$. Since $t'.z\in \I_{H_{\overline{\restri{x}}}}$, we have $e\in \E_{H_{\overline{\restri{x}}}}$. By locality of $A_x$ on $H$, we have $e\in \E_{A_xH}$, hence $t.y\in \I_{A_xG}$.\\
    \textbf{$\dangling{G} \subseteq \dangling{A_x G}$.} Consider $e\in \dangling{G}$. This edge links a vertex $t.z\in \I_G$ to a vertex $u\in \B_G$. Thus $z\in \restri{x}\setminus \restri{x}^-$, which implies $z\in \sp{\I_{A_xG}}$. Since $u\in \B_G\subseteq \I_H$, we have $e\in \E_{H_{\overline{\restri{x}}}}$. As the square cup between $A_xG$ and $H_{\overline{\restri{x}}}$ is well defined, the first two points of Lem.~\ref{lem:compatible subgraphs} imply $t.z\in \I_{A_xG}$ and $e\in \E_{A_xG}$.\\
    \textbf{$ \dangling{A_x G}\subseteq \dangling{G}$.} Consider $e\in \dangling{A_x G}$. This edge links a vertex $t.z\in \I_{A_x G}$ to a vertex $u\in \B_{A_x G}$. Since $\B_{A_xG}=\B_G\subseteq \I_H$, we have $e\in \E_{H_{\overline{\restri{x}}}}$, and by position preservation (Lem.~\ref{lem:easy name preservation}) we have $z\in \sp{\I_G}$. As the square cup between $G$ and $H_{\overline{\restri{x}}}$ is well defined, the first two points of Lem.~\ref{lem:compatible subgraphs} imply $t.z\in \I_G$ and $e\in \E_G$.
\end{proof}

\begin{remark}{Only names of interior vertices are modified}{boundaries are stable}
    For any graph $G$ such that $\exists H\in\sshift,\ \V_G\subseteq \I_H,\ H_{\restri{x}}=G$, we know by Lem.~\ref{lem:border edges are preserved} that $A_x$ preserves boundary positions---i.e. $\forall u\in \I_G,\sp{u}\in\restri{x}\setminus \restri{x}^- \implies u\in \I_{A_xG}$. 
    By border and position preservation (Lem.~\ref{lem:easy name preservation}) this implies that all new names are at positions of $\restri{x}^-$---i.e. ${\sp{\I_{A_x}\setminus \I_G} }\subseteq \restri{x}^-$.
\end{remark}}

\reduced{
\begin{lemma}{Position preservation}{easy name preservation}
  Let $A_x:{\sshift}\to {\sshift}$ be a local rule, then $\sp{\I_{A_x G}} \subseteq \sp{\I_{G}}$ and all new names are at positions of $\restri{x}^-$---i.e. ${\sp{\I_{A_x}\setminus \I_G} }\subseteq \restri{x}^-$.
\end{lemma}
\begin{proof}
    We proceed by contradiction. Consider $G\in  \graphvalid{\restri{}}{x}$, and $y\in \sp{\I_{A_x G}}\setminus\sp{\I_{G}}$. 
    Construct a graph $G'\in \sshift$ such that $ \I_G\cup \{t.y\} \subseteq \I_{G'}$ and $G\sqsubseteq G'$. Such a graph can exist since $y\notin\sp{\V_{G}}$, but we must make sure that $G\in \sshift$. This can be done by first constructing $R(G)$ a copy of $G$ which contains only fresh names, including $t.y$. This $R(G)\in \sshift$ since $\sshift$ is closed under \hyperref[S : renaming]{renaming}. Second, since $\sshift$ is closed under \hyperref[S : disjoint-union]{disjoint union}, we can take $G'= G\sqcup R(G)$. Note that by \hyperref[neighbourhood : stability]{stability} of $\restri{}$ we have $G'\in \graphvalid{\restri{}}{x}$. Moreover we have $G'_{\restri{x}}\sqsubseteq G$ since the neighbourhood cannot select some unconnected nodes in $R(G)$, so by extensivity, $G'_\restri{x}=G_\restri{x}$, from which it follows that $A_xG'=(A_x G_\restri{x}) \sqcup G'_\crestri{x}$. But this union is undefined as $y\in \sp{\I_{A_x G_\restri{x}}}$ and $t.y\in G'_\crestri{x}$ violates the \hyperref[graph : unicity of positions]{unicity of positions} condition.\\
    We have proven $\sp{\I_{A_xG}} \subseteq \sp{\I_{G}}$. Then by Rmk.~\ref{rem:boundaries are stable}, locality and \hyperref[graph : unicity of positions]{unicity of positions} implies ${\sp{\I_{A_x}\setminus \I_G} }\subseteq \restri{x}^-$.
\end{proof}
%

}{}

Summarising, a local rule cannot create new positions: it can only update the time tags of vertices in the interior of the graph (Def.~\ref{def : boundary}). It must also preserve border vertices and border edges, and it must maintain the global acyclicity of the graph. 
We conclude this section with Fig.~\ref{fig : Local operator}, which presents a more generic example of a local rule that illustrates all of these limitations.

\begin{figure*}[t]
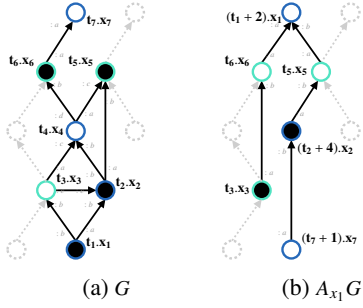

\begin{minipage}[c]{.52\textwidth}
\begin{subfigure}{0.48\textwidth} 
    \captionsetup{justification=centering}
    \resizebox{0.7\textwidth}{!}{\input{figs/Local_operator_1.tex}}
    \caption{$G$}
    \label{fig : Local operator 1}
\end{subfigure}
\hfil
\begin{subfigure}{0.48\textwidth}
    \captionsetup{justification=centering}
    \resizebox{0.7\textwidth}{!}{\input{figs/Local_operator_2.tex}}
    \caption{$A_{x_1}G$}
    \label{fig : Local operator 2}
\end{subfigure}
\end{minipage}
\hfil
\begin{minipage}[c]{.47\textwidth}
\caption{{\em Action of a local rule $A_{(-)}$} centered on a vertex $u_1 = t_1.x_1$. It affects the vertices at positions $\restr{x}{G}$ that are circled dark blue \& cyan. Dark blue vertices (e.g. $u_1,u_4$) can be almost arbitrarily modified whereas cyan vertices must have their names and external edges preserved. Internal states $\Sigma = \{0,1\}$ are represented by white and black.
One can set $u_1'$, $u_2'$, and $u_7'$ to $3.u_1$, $1.u_2$, and $4.u_7$ respectively for instance. } 
\label{fig : Local operator}
\end{minipage}
\end{figure*}

\section{Sequential dynamics and expressivity}

Having defined the action of a local rule for a single step, we now turn to sequential dynamics. Starting from an initial graph, successive local updates produce a computation, recorded by a valid sequence of positions. We introduce space-time diagrams to collect the possible evolutions generated by successive local updates.

Before studying the theoretical properties of sequential composition, we illustrate the generality of the dynamics captured by the model. We first show how the local rule developed throughout the previous section can be generalised to simulate any one-dimensional cellular automaton. We then present the time-dilation example, in which different regions advance at different rates under the same local rule, illustrating dynamics without a natural global-clock interpretation.

\subsection{Valid sequences and space-time diagrams}\label{sec: sequences and ST diag}

Having defined local rules, we now turn to the kinds of ``temporal'' evolutions that can emerge from them. Given an initial graph $G$, the rule $A_{(-)}$ can first evolve any past vertex $t_0.x_0$ such that $G \in \graphvalid{\restri{}}{x_0}$. It can then evolve any vertex $t_1.x_1$ such that $A_x G \in \graphvalid{\restri{}}{x'}$, and so on.\\
We call a sequence of such nondeterministic update choices $\omega = x_n \dots x_1 x_0$ a \emph{valid sequence}. This is a classical object in rewriting theory, sometimes referred to as a rewriting strategy or schedule.

\begin{definition}[Valid sequences]
    From now on $A_{yx}$ will stand for $A_y A_x$.
    We say that $\omega \in \mathcal{X}^*$ is a valid sequence of $G$ iff, for all $\omega_1,\omega_2\in \mathcal{X}^*$ such that $\omega = \omega_2 x \omega_1$, we have $A_{\omega_1}G\in \graphvalidexplicit{\restri{}}{x}$. 
    We denote $\valid{G}{A} \subseteq \mathcal{X}^*$ 
    the set of valid sequences of $G$.
\end{definition}

\begin{example}{Particle system---Sequences}{}
    In the graph $G$ below, $53$ is a valid sequence: $A_{53}G$ is obtained by first applying $A_3$ and then $A_5$. Note that $453$ is also a valid sequence, since $A_{53}G \in \graphvalid{\restri{}}{4}$. However, the sequence $653$ is not valid. Although $6 \in \Past(A_{53}G)$, we do not have $A_{53}G \in \graphvalid{\restri{}}{6}$ because $2.7 \in \B_G$ (see $\graphvalid{\restri{}}{x}$ defined in Ex.~\ref{ex:PS---NS}).\\
    \begin{tabular}{ccc}
        \begin{minipage}{0.31\linewidth}
            \centering
            \includegraphics[width=\linewidth]{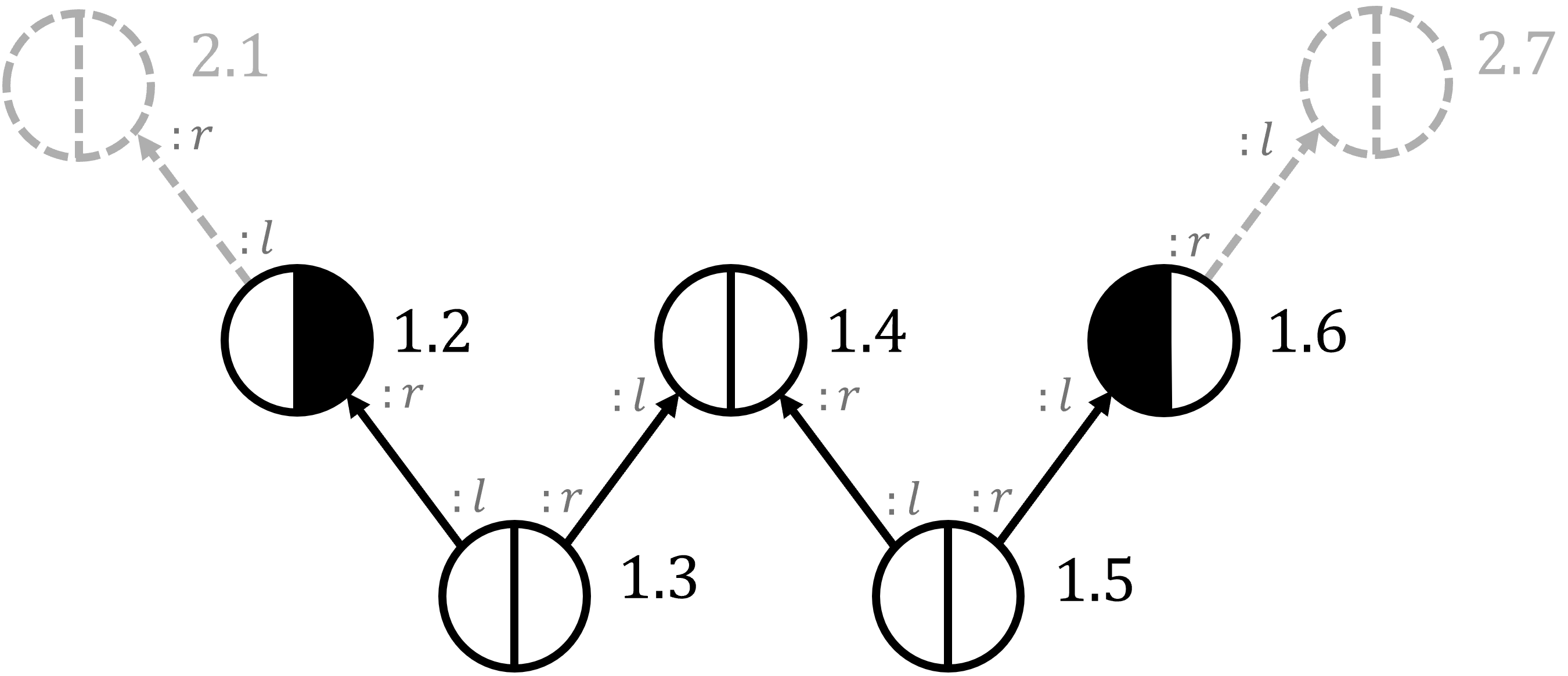}
            \captionof*{figure}{$G$.}
        \end{minipage}
         &
         \begin{minipage}{0.31\linewidth}
            \centering
            \includegraphics[width=\linewidth]{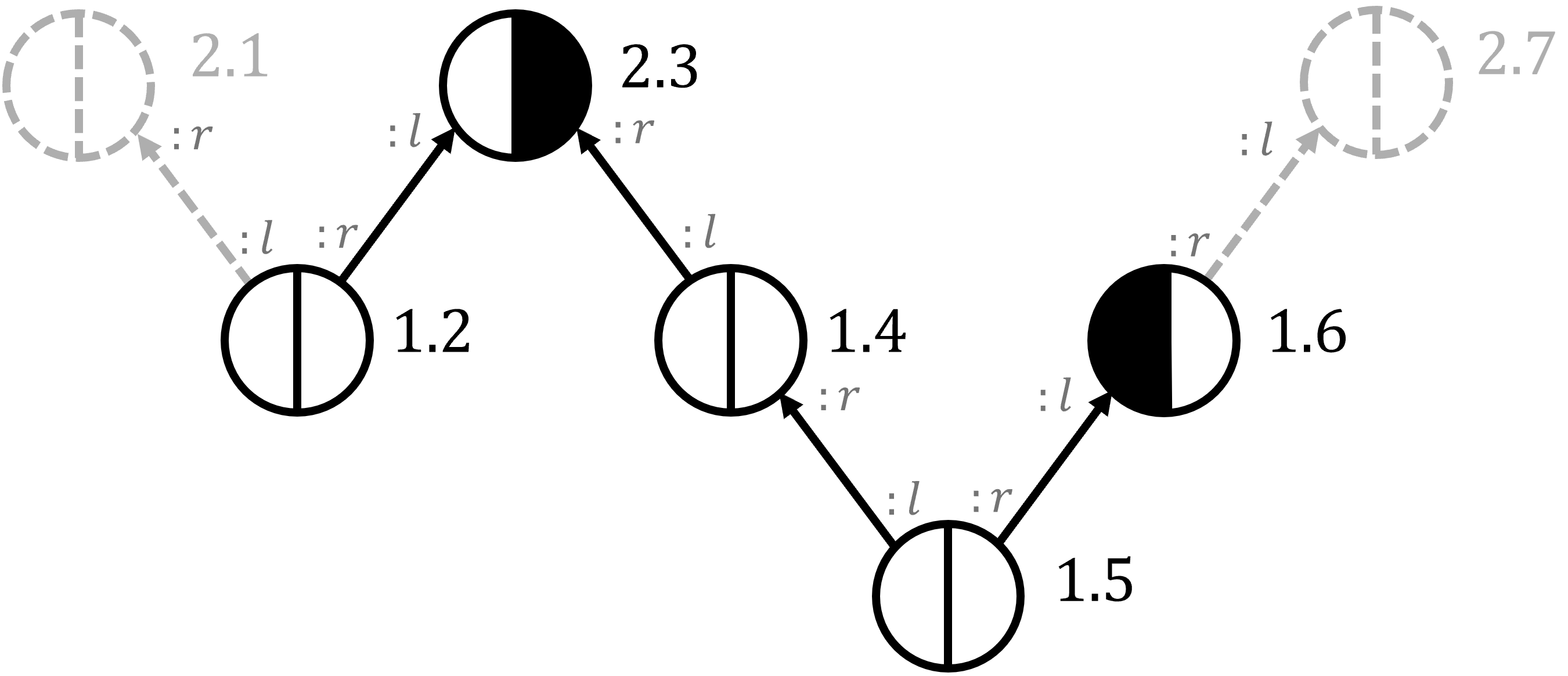}
            \captionof*{figure}{$A_3 G$.} 
        \end{minipage}
         &
         \begin{minipage}{0.31\linewidth}
            \centering
            \includegraphics[width=\linewidth]{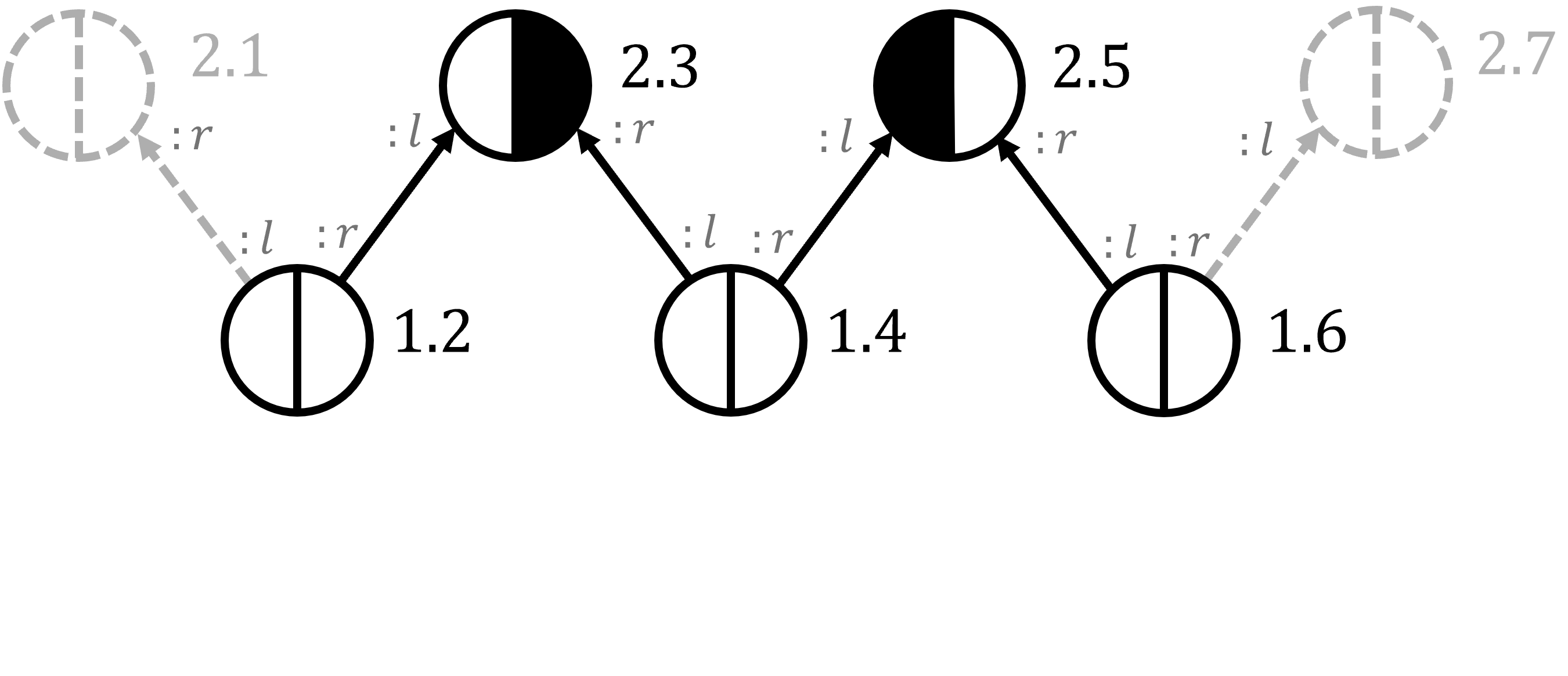}
            \captionof*{figure}{$A_{53} G$.} 
        \end{minipage}
    \end{tabular}
\end{example}



We now define the \emph{space-time diagram} associated to an initial graph and a local rule, as the collection of all valid evolutions of this graph.

\begin{definition}[Space-time diagram]
    Given a graph $G$ and a local rule $A_{(-)}$, the \emph{space-time diagram} $\mathcal{M}(G) := \{\, A_{\omega} G\mid \omega \in \valid{G}{A} \,\}$ is the set of all generated graphs.
    We sometimes omit $G$.
\end{definition}
To visualise how the graphs in $\mathcal{M}$ share common vertices and edges, a \emph{space-time background} is depicted (Fig.~\ref{fig:spacetimediagram}).
These space-time backgrounds are ``pseudo-graphs'', i.e. graphs in the usual sense, without the port constraints or the non-overlapping-position condition of Def.~\ref{def : graphs}. Each pseudo-graph $M$ is defined by:
\begin{align*}
    \V_M &=\bigcup_{G\in \mathcal{M}} \V_G & \E_M &=\bigcup_{G\in \mathcal{M}} \E_G.
\end{align*}

\begin{figure*}[h]
    \centering
    \resizebox{0.65\textwidth}{!}{\begin{tikzpicture}


\foreach \i in {1,2,...,6}{
    \foreach \j in {1,2,3}{
    	\FPeval{\posx}{clip(\i*8)}
        \FPeval{\posy}{clip(\j*8)}
        \ifnum \j=2
            {\vertex{\posx}{\posy}{border}{0}{\j}{\i\j}}
        \else
            {\vertex{\posx}{\posy}{myblack}{0}{\j}{\i\j}}
        \fi
    }
}

\foreach \i in {1,2,...,5}{
    \foreach \j in {1,2,3}{
    	\FPeval{\posx}{clip(\i*8+4)}
        \FPeval{\posy}{clip(\j*8+4)}
        \ifnum \j=2
            {\ifnum \i=3
                {\vertex{\posx}{\posy}{border}{0}{u}{S\i\j}}
            \else
                {\vertex{\posx}{\posy}{border}{0}{\j}{S\i\j}}
            \fi
            }
        \else
            {\vertex{\posx}{\posy}{myblack}{0}{\j}{S\i\j}}
        \fi
    }
}


\foreach \i in {1,2,...,5}{
    \FPeval{\k}{clip(\i+1)}
    \edge{\i1}{S\i1}{myblack}
    \edge{\k1}{S\i1}{myblack}
    \edge{S\i1}{\i2}{border}
    \edge{S\i1}{\k2}{border}
    \edge{\i2}{S\i2}{border}
    \edge{\k2}{S\i2}{border}
    \edge{S\i2}{\i3}{border}
    \edge{S\i2}{\k3}{border}
    \edge{\i3}{S\i3}{myblack}
    \edge{\k3}{S\i3}{myblack}
}

\rightmoover{20}{12}{myblack}{first}
\rightmoover{36}{28}{myblack}{second}
\leftmoover{36}{12}{myblack}{third}
\leftmoover{20}{28}{myblack}{fourth}

\draw (7.5,29) node[right]{{\Huge \scalefont{2.5} $G'$}};
\draw (7.5,12) node[right]{{\Huge \scalefont{2.5} $G$}};

\draw (20.8,10.1) node[color=port]{\Huge $:r$};
\draw (23.4,9.6) node[color=port]{\Huge $:l$};
\draw (24.75,9.6) node[color=port]{\Huge $:r$};
\draw (27.2,10.1) node[color=port]{\Huge $:l$};
\draw (28.8,10.1) node[color=port]{\Huge $:r$};
\draw (31.4,9.6) node[color=port]{\Huge $:l$};
\draw (32.75,9.6) node[color=port]{\Huge $:r$};
\draw (35.2,10.1) node[color=port]{\Huge $:l$};
    
\end{tikzpicture}}
    \caption{{\em Particle system example---Space-time diagram.} In black we highlight just the graphs $G$ and $G'$ belonging to the space-time diagram $\mathcal{M}(G)$. The local rule here moves the left-side particle towards the right and the right-side particle towards the left. 
    }
    \label{fig:spacetimediagram}
\end{figure*}
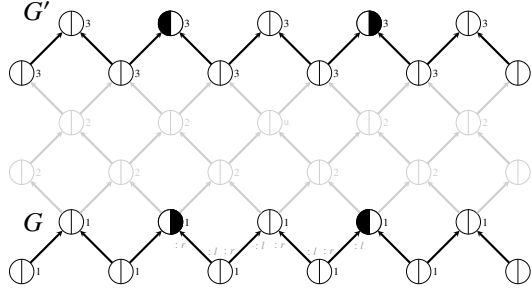

At this stage, the space-time diagram is simply a collection of graphs generated by all possible schedulings. A natural question is whether the outcome of an event in space-time depends on the particular scheduling chosen, or whether the space-time is \emph{well-determined} in the sense that any two schedulings agree on the state of every event they both describe. We call this property \emph{space-time determinism}; it is formalised in~\cite{ArrighiSpaceTimeDetWRLA}, and completing the remaining proofs---including local sufficient conditions---is left to future work.

\subsection{Simulating synchronous cellular automata}\label{sec: Simulation}

Beyond the above particle system example, we now show that we can simulate any one-dimensional cellular automaton (CA for short) in spite of its synchronicity. 
The construction borrows from the well-studied `marching soldiers' scheme \cite{WeakConsistencyGacs}, as best formalised by \cite{MarchingNehaniv}
, but relies on the DAG of dependency rather than extra states as its mechanism for local, relative synchronisation.
Without loss of generality \cite{IbarraJiang}, we simulate radius half CA.

\begin{figure*}[h]
\centering
\begin{subfigure}{.33\textwidth}
    \captionsetup{justification=centering}
    \resizebox{\textwidth}{!}{\tikzset{every picture/.style={line width=0.75pt}} 

\begin{tikzpicture}[x=0.75pt,y=0.75pt,yscale=-1,xscale=1]

\draw [color={rgb, 255:red, 0; green, 0; blue, 0 }  ,draw opacity=1 ][line width=1.5]    (92.03,147.46) -- (212.35,147.46) ;
\draw [color={rgb, 255:red, 0; green, 0; blue, 0 }  ,draw opacity=1 ][line width=1.5]    (405,146.73) -- (212.35,147.46) ;
\draw  [color={rgb, 255:red, 0; green, 0; blue, 0 }  ,draw opacity=1 ][fill={rgb, 255:red, 255; green, 255; blue, 255 }  ,fill opacity=1 ][dash pattern={on 5.63pt off 4.5pt}][line width=1.5]  (80,147.46) .. controls (80,133.9) and (90.99,122.9) .. (104.55,122.9) .. controls (118.12,122.9) and (129.11,133.9) .. (129.11,147.46) .. controls (129.11,161.02) and (118.12,172.01) .. (104.55,172.01) .. controls (90.99,172.01) and (80,161.02) .. (80,147.46) -- cycle ;
\draw  [color={rgb, 255:red, 0; green, 0; blue, 0 }  ,draw opacity=1 ][fill={rgb, 255:red, 255; green, 255; blue, 255 }  ,fill opacity=1 ][dash pattern={on 5.63pt off 4.5pt}][line width=1.5]  (380,145) .. controls (380,131.19) and (391.19,120) .. (405,120) .. controls (418.81,120) and (430,131.19) .. (430,145) .. controls (430,158.81) and (418.81,170) .. (405,170) .. controls (391.19,170) and (380,158.81) .. (380,145) -- cycle ;
\draw [color={rgb, 255:red, 0; green, 0; blue, 0 }  ,draw opacity=1 ][line width=1.5]    (87.83,449.46) -- (208.15,449.46) ;
\draw [color={rgb, 255:red, 0; green, 0; blue, 0 }  ,draw opacity=1 ][line width=1.5]    (328.47,449.46) -- (208.15,449.46) ;
\draw [color={rgb, 255:red, 0; green, 0; blue, 0 }  ,draw opacity=1 ][line width=1.5]    (328.47,449.46) -- (420,450) ;
\draw  [color={rgb, 255:red, 0; green, 0; blue, 0 }  ,draw opacity=1 ][fill={rgb, 255:red, 255; green, 255; blue, 255 }  ,fill opacity=1 ][dash pattern={on 5.63pt off 4.5pt}][line width=1.5]  (380,445.41) .. controls (380,431.83) and (391.19,420.82) .. (405,420.82) .. controls (418.81,420.82) and (430,431.83) .. (430,445.41) .. controls (430,458.99) and (418.81,470) .. (405,470) .. controls (391.19,470) and (380,458.99) .. (380,445.41) -- cycle ;
\draw  [draw opacity=0][fill={rgb, 255:red, 208; green, 2; blue, 27 }  ,fill opacity=0.23 ] (54.72,448.39) -- (220.89,145.85) -- (271.18,145.85) -- (277.35,448.39) -- cycle ;
\draw  [color={rgb, 255:red, 208; green, 2; blue, 27 }  ,draw opacity=1 ][fill={rgb, 255:red, 208; green, 2; blue, 27 }  ,fill opacity=0.2 ] (160,301.19) .. controls (160,285.28) and (182.39,272.38) .. (210,272.38) .. controls (237.61,272.38) and (260,285.28) .. (260,301.19) .. controls (260,317.1) and (237.61,330) .. (210,330) .. controls (182.39,330) and (160,317.1) .. (160,301.19) -- cycle ;

\draw  [color={rgb, 255:red, 0; green, 0; blue, 0 }  ,draw opacity=1 ][fill={rgb, 255:red, 255; green, 255; blue, 255 }  ,fill opacity=1 ][line width=1.5]  (220.89,145.85) .. controls (220.89,132.51) and (231.88,121.7) .. (245.45,121.7) .. controls (259.01,121.7) and (270,132.51) .. (270,145.85) .. controls (270,159.19) and (259.01,170) .. (245.45,170) .. controls (231.88,170) and (220.89,159.19) .. (220.89,145.85) -- cycle ;
\draw  [color={rgb, 255:red, 0; green, 0; blue, 0 }  ,draw opacity=1 ][fill={rgb, 255:red, 255; green, 255; blue, 255 }  ,fill opacity=1 ][line width=1.5]  (75.8,449.46) .. controls (75.8,435.9) and (86.79,424.9) .. (100.35,424.9) .. controls (113.91,424.9) and (124.91,435.9) .. (124.91,449.46) .. controls (124.91,463.02) and (113.91,474.01) .. (100.35,474.01) .. controls (86.79,474.01) and (75.8,463.02) .. (75.8,449.46) -- cycle ;
\draw  [color={rgb, 255:red, 0; green, 0; blue, 0 }  ,draw opacity=1 ][fill={rgb, 255:red, 255; green, 255; blue, 255 }  ,fill opacity=1 ][line width=1.5]  (221.6,445.85) .. controls (221.6,432.51) and (232.6,421.7) .. (246.16,421.7) .. controls (259.72,421.7) and (270.71,432.51) .. (270.71,445.85) .. controls (270.71,459.19) and (259.72,470) .. (246.16,470) .. controls (232.6,470) and (221.6,459.19) .. (221.6,445.85) -- cycle ;

\draw (234,127.4) node [anchor=north west][inner sep=0.75pt]  [font=\sizehalfCA]  {$\sigma_{1}^{1}$};
\draw (87.6,432.4) node [anchor=north west][inner sep=0.75pt]  [font=\sizehalfCA]  {$\sigma_{0}^{0}$};
\draw (232.6,430.4) node [anchor=north west][inner sep=0.75pt]  [font=\sizehalfCA]  {$\sigma_{1}^{0}$};
\draw (198.14,279.61) node [anchor=north west][inner sep=0.75pt]  [font=\sizehalfCAb,color={rgb, 255:red, 208; green, 2; blue, 27 }  ,opacity=1 ]  {$f$};

\end{tikzpicture}}
    \caption{Radius one half locality \dots}
    \label{fig :one_half_radius_CA_1}
\end{subfigure}
\hfill
\begin{subfigure}{.495\textwidth}
    \captionsetup{justification=centering}
    \resizebox{\textwidth}{!}{\input{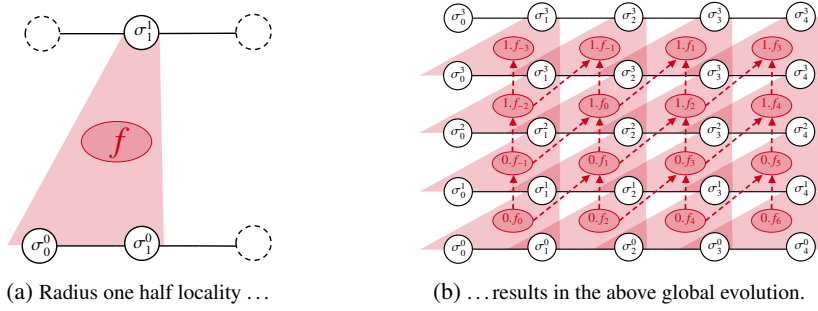}}
    \caption{\dots results in the above global evolution.}
    \label{fig:one_half_radius_CA_2}
\end{subfigure}
\caption{{\em Radius-one-half local rule (a):} A radius-one-half local rule $f$ takes two states as input $(\sigma_x^t,\sigma_{x+1}^t)\in \Sigma_{CA}^2$ and outputs $\sigma_{x+1}^{t+1} = f(\sigma_x^t,\sigma_{x+1}^t)$. {\em Global evolution (b):} The initial configuration is an infinite one-dimensional array $(\sigma^0_x)_{x\in \mathbb{Z}}$. We consider the cellular automaton that applies $f$ homogeneously in space. The successive configurations computed by this automaton are depicted in black ($\sigma^t_x$ is just short for $f(\sigma^{t-1}_{x-1},\sigma^{t-1}_{x})$). Applications of the local rule $f$ are shown in red. We name each application to facilitate the introduction of the simulation scheme in Fig.~\ref{fig:Simulation 2}. Red dashed arrows depict the dependencies between applications of $f$; for example, since $0.f_1$ needs the output $\sigma_1^1$ of $0.f_0$, we draw a red dashed arrow from $0.f_0$ to $0.f_1$.
}
\label{fig: one_half_radius_CA}
\end{figure*}

First let us recall that a radius one half cellular automaton acting on an alphabet $\Sigma_{CA}$ is a global function $F:{\Sigma_{CA}}^{\mathbb{Z}}\to {\Sigma_{CA}}^{\mathbb{Z}}$ defined by the synchronous application of a local function $f: {\Sigma_{CA}}^2\to {\Sigma_{CA}}$ everywhere at once, cf. Fig.~\ref{fig :one_half_radius_CA_1}.
The dependencies between the different applications of $f$ are shown in Fig.~\ref{fig:one_half_radius_CA_2}, notice the similarity with Fig.~\ref{fig:spacetimediagram}. This suggests that the natural way to encode this cellular automaton in our model is to use vertices (a.k.a events) to represent each application of $f$.

Second we define $\Sigma = (\Sigma_{CA} \cup \epsilon)^2$ and pick the same neighbourhood scheme and ports as in the previous example. We define the local rule as depicted in Ex.~\ref{ex:CAsim---LR}. Its action can be understood as follows: 1/ applied to the vertex named $0.f_0$, it computes $f(\sigma_0^0,\sigma_1^0)$ 2/ and stores two copies of the result $\sigma_1^1$, one in the right part of the vertex $0.f_{-1}$, and the other in the left part of the vertex $0.f_1$. 3/ Finally it creates the vertex $1.f_0$, which is awaiting for the output of the $0.f_{-1}$ and $0.f_1$ applications, as represented by the two incoming edges.

\begin{example}{CA simulation---Local rule}{CAsim---LR}
    \begin{minipage}{0.49\linewidth}
        $A_{f_0}$ acts exactly as $A_x$ in Ex.~\ref{ex:PS---LR}, except for the internal states. If $\sigma(0.f_0) = (\sigma_0^0,\sigma_1^0)$ does not contain $\epsilon$, it computes $\sigma_1^1 = f(\sigma_0^0,\sigma_1^0)$ and stores the result in both neighbours.
    \end{minipage}
    \begin{minipage}{0.49\linewidth}
        \begin{tabular}{cc}
            \begin{minipage}{0.49\linewidth}
                \centering
                \captionsetup{justification=centering}
                \resizebox{0.65\textwidth}{!}{\tikzset{every picture/.style={line width=0.75pt}} 

\begin{tikzpicture}[x=0.75pt,y=0.75pt,yscale=-1,xscale=1]

\draw [color={rgb, 255:red, 0; green, 0; blue, 0 }  ,draw opacity=1 ][line width=2.25]    (226,274) -- (189.91,202) ;
\draw [shift={(187.67,197.53)}, rotate = 63.38] [fill={rgb, 255:red, 0; green, 0; blue, 0 }  ,fill opacity=1 ][line width=0.08]  [draw opacity=0] (14.29,-6.86) -- (0,0) -- (14.29,6.86) -- cycle    ;
\draw [line width=2.25]    (226,274) -- (264.27,204.09) ;
\draw [shift={(266.67,199.7)}, rotate = 118.69] [fill={rgb, 255:red, 0; green, 0; blue, 0 }  ][line width=0.08]  [draw opacity=0] (14.29,-6.86) -- (0,0) -- (14.29,6.86) -- cycle    ;
\draw  [color={rgb, 255:red, 38; green, 105; blue, 185 }  ,draw opacity=1 ][fill={rgb, 255:red, 255; green, 255; blue, 255 }  ,fill opacity=1 ][line width=3]  (200,274) .. controls (200,259.64) and (211.64,248) .. (226,248) .. controls (240.36,248) and (252,259.64) .. (252,274) .. controls (252,288.36) and (240.36,300) .. (226,300) .. controls (211.64,300) and (200,288.36) .. (200,274) -- cycle ;
\draw  [color={rgb, 255:red, 80; green, 227; blue, 194 }  ,draw opacity=1 ][fill={rgb, 255:red, 255; green, 255; blue, 255 }  ,fill opacity=1 ][line width=3]  (152,176) .. controls (152,162.75) and (162.75,152) .. (176,152) .. controls (189.25,152) and (200,162.75) .. (200,176) .. controls (200,189.25) and (189.25,200) .. (176,200) .. controls (162.75,200) and (152,189.25) .. (152,176) -- cycle ;
\draw  [color={rgb, 255:red, 80; green, 227; blue, 194 }  ,draw opacity=1 ][fill={rgb, 255:red, 255; green, 255; blue, 255 }  ,fill opacity=1 ][line width=3]  (252,176) .. controls (252,162.75) and (262.75,152) .. (276,152) .. controls (289.25,152) and (300,162.75) .. (300,176) .. controls (300,189.25) and (289.25,200) .. (276,200) .. controls (262.75,200) and (252,189.25) .. (252,176) -- cycle ;
\draw [color={rgb, 255:red, 80; green, 227; blue, 194 }  ,draw opacity=1 ][line width=3]    (176,152) -- (176,200) ;
\draw [color={rgb, 255:red, 80; green, 227; blue, 194 }  ,draw opacity=1 ][line width=3]    (276,152) -- (276,200) ;
\draw [color={rgb, 255:red, 74; green, 107; blue, 226 }  ,draw opacity=1 ][line width=3]    (226,250) -- (226,298) ;

\draw (258,282.4) node [anchor=north west][inner sep=0.75pt]  [font=\sizesimub]  {${0.f_{0}}$};
\draw (200,201.4) node [anchor=north west][inner sep=0.75pt]  [font=\sizesimup,color={rgb, 255:red, 155; green, 155; blue, 155 }  ,opacity=1 ]  {$:r$};
\draw (230,200.4) node [anchor=north west][inner sep=0.75pt]  [font=\sizesimup,color={rgb, 255:red, 155; green, 155; blue, 155 }  ,opacity=1 ]  {$:l$};
\draw (246,232.4) node [anchor=north west][inner sep=0.75pt]  [font=\sizesimup,color={rgb, 255:red, 155; green, 155; blue, 155 }  ,opacity=1 ]  {$:r$};
\draw (184,232.4) node [anchor=north west][inner sep=0.75pt]  [font=\sizesimup,color={rgb, 255:red, 155; green, 155; blue, 155 }  ,opacity=1 ]  {$:l$};
\draw (204.83,262.4) node [anchor=north west][inner sep=0.75pt]    {\sizesimu $\sigma_{0}^{0}$};
\draw (229.83,264) node [anchor=north west][inner sep=0.75pt]    {\sizesimu $\sigma_{1}^{0}$};
\draw (261,167.4) node [anchor=north west][inner sep=0.75pt]    {\sizesimu $\epsilon $};
\draw (182,166.4) node [anchor=north west][inner sep=0.75pt]    {\sizesimu $\epsilon $};

\end{tikzpicture}}
                \captionof*{figure}{$G_{\restri{f_0}}.$}
                \label{fig :Simulation 1 a}         
            \end{minipage}
            &
            \begin{minipage}{0.49\linewidth}
                \centering
                \captionsetup{justification=centering}
                \resizebox{0.65\textwidth}{!}{\tikzset{every picture/.style={line width=0.75pt}} 

\begin{tikzpicture}[x=0.75pt,y=0.75pt,yscale=-1,xscale=1]

\draw [color={rgb, 255:red, 0; green, 0; blue, 0 }  ,draw opacity=1 ][line width=2.25]    (276,176) -- (239.95,105.88) ;
\draw [shift={(237.67,101.43)}, rotate = 62.79] [fill={rgb, 255:red, 0; green, 0; blue, 0 }  ,fill opacity=1 ][line width=0.08]  [draw opacity=0] (14.29,-6.86) -- (0,0) -- (14.29,6.86) -- cycle    ;
\draw [line width=2.25]    (176,176) -- (212.32,107.85) ;
\draw [shift={(214.67,103.43)}, rotate = 118.05] [fill={rgb, 255:red, 0; green, 0; blue, 0 }  ][line width=0.08]  [draw opacity=0] (14.29,-6.86) -- (0,0) -- (14.29,6.86) -- cycle    ;
\draw  [color={rgb, 255:red, 38; green, 105; blue, 185 }  ,draw opacity=1 ][fill={rgb, 255:red, 255; green, 255; blue, 255 }  ,fill opacity=1 ][line width=3]  (200,76) .. controls (200,61.64) and (211.64,50) .. (226,50) .. controls (240.36,50) and (252,61.64) .. (252,76) .. controls (252,90.36) and (240.36,102) .. (226,102) .. controls (211.64,102) and (200,90.36) .. (200,76) -- cycle ;
\draw  [color={rgb, 255:red, 80; green, 227; blue, 194 }  ,draw opacity=1 ][fill={rgb, 255:red, 255; green, 255; blue, 255 }  ,fill opacity=1 ][line width=3]  (152,176) .. controls (152,162.75) and (162.75,152) .. (176,152) .. controls (189.25,152) and (200,162.75) .. (200,176) .. controls (200,189.25) and (189.25,200) .. (176,200) .. controls (162.75,200) and (152,189.25) .. (152,176) -- cycle ;
\draw  [color={rgb, 255:red, 80; green, 227; blue, 194 }  ,draw opacity=1 ][fill={rgb, 255:red, 255; green, 255; blue, 255 }  ,fill opacity=1 ][line width=3]  (252,176) .. controls (252,162.75) and (262.75,152) .. (276,152) .. controls (289.25,152) and (300,162.75) .. (300,176) .. controls (300,189.25) and (289.25,200) .. (276,200) .. controls (262.75,200) and (252,189.25) .. (252,176) -- cycle ;
\draw [color={rgb, 255:red, 80; green, 227; blue, 194 }  ,draw opacity=1 ][line width=3]    (176,152) -- (176,200) ;
\draw [color={rgb, 255:red, 80; green, 227; blue, 194 }  ,draw opacity=1 ][line width=3]    (276,152) -- (276,200) ;
\draw [color={rgb, 255:red, 74; green, 107; blue, 226 }  ,draw opacity=1 ][line width=3]    (226,50) -- (226,98) ;

\draw (251,112.4) node [anchor=north west][inner sep=0.75pt]  [font=\sizesimup,color={rgb, 255:red, 155; green, 155; blue, 155 }  ,opacity=1 ]  {$:r$};
\draw (179,101.4) node [anchor=north west][inner sep=0.75pt]  [font=\sizesimup,color={rgb, 255:red, 155; green, 155; blue, 155 }  ,opacity=1 ]  {$:l$};
\draw (195,141.4) node [anchor=north west][inner sep=0.75pt]  [font=\sizesimup,color={rgb, 255:red, 155; green, 155; blue, 155 }  ,opacity=1 ]  {$:r$};
\draw (236,141.4) node [anchor=north west][inner sep=0.75pt]  [font=\sizesimup,color={rgb, 255:red, 155; green, 155; blue, 155 }  ,opacity=1 ]  {$:l$};
\draw (210,67.4) node [anchor=north west][inner sep=0.75pt]    {\sizesimu $\epsilon $};
\draw (233,67.4) node [anchor=north west][inner sep=0.75pt]    {\sizesimu $\epsilon $};
\draw (177.83,164) node [anchor=north west][inner sep=0.75pt]    {\sizesimu $\sigma_{1}^{1}$};
\draw (255.83,164) node [anchor=north west][inner sep=0.75pt]    {\sizesimu $\sigma_{1}^{1}$};
\draw (261,72.4) node [anchor=north west][inner sep=0.75pt]  [font=\sizesimub]  {$1{.f_{0}}$};

\end{tikzpicture}}
                \captionof*{figure}{$A_{f_0}G_{\restri{f_0}}.$}
                \label{fig :Simulation 1 b}
            \end{minipage}
        \end{tabular}
    \end{minipage}
\end{example}

\begin{figure*}[h!]
    \centering
    \captionsetup{justification=centering}   
    \resizebox{0.8\textwidth}{!}{\begin{tikzpicture}


\foreach \i in {1,2,...,6}{
    \foreach \j in {1,2}{
    	\FPeval{\posx}{clip(\i*8)}
        \FPeval{\posy}{clip(\j*8)}
        \ifnum \j=2
            {\ifnum \i=1
                {}
            \else
                \ifnum \i=6
                    {}
                \else
                    {\vertex{\posx}{\posy}{myblack}{0}{}{\i\j}}
                \fi
            \fi
            }
        \else
            {\ifnum \i<3
                {}
            \else
                {\vertex{\posx}{\posy}{myblack}{0}{}{\i\j}}
            \fi
            }
        \fi
    }
}

\foreach \i in {1,2,...,5}{
    \foreach \j in {1,2}{
    	\FPeval{\posx}{clip(\i*8+4)}
        \FPeval{\posy}{clip(\j*8+4)}
        \ifnum \j=1
            {\ifnum \i=1
                {}
            \else
                {\vertex{\posx}{\posy}{myblack}{0}{}{S\i\j}}
            \fi
            }
        \else
            \ifnum \i=5
                {}
            \else
                {\vertex{\posx}{\posy}{myblack}{0}{}{S\i\j}}
            \fi
        \fi
    }
}


\foreach \i in {1,2,...,5}{
    \FPeval{\k}{clip(\i+1)}
    \ifnum \i=1
        {\edge{\k2}{S\i2}{myblack}}
    \else
        \ifnum \i=2
            {\edge{\k1}{S\i1}{myblack}
            \edge{S\i1}{\i2}{border}
            \edge{S\i1}{\k2}{border}
            \edge{\i2}{S\i2}{myblack}
            \edge{\k2}{S\i2}{myblack}}
        \else
            \ifnum \i=5
                {\edge{\i1}{S\i1}{myblack}
                \edge{\k1}{S\i1}{myblack}
                \edge{S\i1}{\i2}{border}}
            \else
                \edge{\i1}{S\i1}{myblack}
                \edge{\k1}{S\i1}{myblack}
                \edge{S\i1}{\i2}{border}
                \edge{S\i1}{\k2}{border}
                \edge{\i2}{S\i2}{myblack}
                \edge{\k2}{S\i2}{myblack}
            \fi
        \fi
    \fi
}


\draw (26.3,8) node{{\sizebigsimub $0.f_0$}};
\draw (34.3,8) node{{\sizebigsimub $0.f_2$}};
\draw (42.3,8) node{{\sizebigsimub $0.f_4$}};
\draw (50.3,8) node{{\sizebigsimub $0.f_6$}};
\draw (22.3,12) node{{\sizebigsimub $0.f_{-1}$}};
\draw (30.3,12) node{{\sizebigsimub $0.f_{1}$}};
\draw (38.3,12) node{{\sizebigsimub $0.f_{3}$}};
\draw (46.3,12) node{{\sizebigsimub $0.f_{5}$}};
\draw (18.3,16) node{{\sizebigsimub $1.f_{-2}$}};
\draw (26.3,16) node{{\sizebigsimub $1.f_{0}$}};
\draw (34.3,16) node{{\sizebigsimub $1.f_{2}$}};
\draw (42.3,16) node{{\sizebigsimub $1.f_{4}$}};
\draw (14.3,20) node{{\sizebigsimub $1.f_{-3}$}};
\draw (22.3,20) node{{\sizebigsimub $1.f_{-1}$}};
\draw (30.3,20) node{{\sizebigsimub $1.f_{1}$}};
\draw (38.3,20) node{{\sizebigsimub $1.f_{3}$}};


\draw (24.4,8) node{{\sizebigsimu $\sigma_1^0$}};
\draw (32.4,8) node{{\sizebigsimu $\sigma_2^0$}};
\draw (40.4,8) node{{\sizebigsimu $\sigma_3^0$}};
\draw (48.4,8) node{{\sizebigsimu $\sigma_4^0$}};
\draw (20.4,12) node{{\sizebigsimu $\epsilon$}};
\draw (28.4,12) node{{\sizebigsimu $\epsilon$}};
\draw (36.4,12) node{{\sizebigsimu $\epsilon$}};
\draw (44.4,12) node{{\sizebigsimu $\epsilon$}};
\draw (16.4,16) node{{\sizebigsimu $\sigma_1^2$}};
\draw (24.4,16) node{{\sizebigsimu $\sigma_2^2$}};
\draw (32.4,16) node{{\sizebigsimu $\sigma_3^2$}};
\draw (40.4,16) node{{\sizebigsimu $\sigma_4^2$}};
\draw (12.4,20) node{{\sizebigsimu $\epsilon$}};
\draw (20.4,20) node{{\sizebigsimu $\epsilon$}};
\draw (28.4,20) node{{\sizebigsimu $\epsilon$}};
\draw (36.4,20) node{{\sizebigsimu $\epsilon$}};

\draw (23.5,8) node{{\sizebigsimu $\sigma_0^0$}};
\draw (31.5,8) node{{\sizebigsimu $\sigma_1^0$}};
\draw (39.5,8) node{{\sizebigsimu $\sigma_2^0$}};
\draw (47.5,8) node{{\sizebigsimu $\sigma_3^0$}};
\draw (19.5,12) node{{\sizebigsimu $\epsilon$}};
\draw (27.5,12) node{{\sizebigsimu $\epsilon$}};
\draw (35.5,12) node{{\sizebigsimu $\epsilon$}};
\draw (43.5,12) node{{\sizebigsimu $\epsilon$}};
\draw (15.5,16) node{{\sizebigsimu $\sigma_0^2$}};
\draw (23.5,16) node{{\sizebigsimu $\sigma_1^2$}};
\draw (31.5,16) node{{\sizebigsimu $\sigma_2^2$}};
\draw (39.5,16) node{{\sizebigsimu $\sigma_3^2$}};
\draw (11.5,20) node{{\sizebigsimu $\epsilon$}};
\draw (19.5,20) node{{\sizebigsimu $\epsilon$}};
\draw (27.5,20) node{{\sizebigsimu $\epsilon$}};
\draw (35.5,20) node{{\sizebigsimu $\epsilon$}};

\draw (10.5,17) node[right]{{\Huge \scalefont{2.5} $G'$}};
\draw (18.5,9) node[right]{{\Huge \scalefont{2.5} $G$}};
    
\end{tikzpicture}}
    \caption{{\em Space-time diagram of the simulation} of the cellular automaton in Fig.~\ref{fig:one_half_radius_CA_2}. The highlighted $G$ encodes the initial cellular automaton configuration $\sigma^0$ of the CA, whilst $G'$ is locally equal to $\sigma^2$ the configuration obtained after two time steps of the cellular automaton.
    }
    \label{fig:Simulation 2}
\end{figure*}

We encode the initial configuration $\sigma^0$ by a graph $G$ (see Fig.~\ref{fig:Simulation 2}) containing two copies of each state of $\sigma^0$. The simulation is local: we do not require a graph containing an entire configuration $\sigma^t$, which would involve infinitely many updates. Instead, each state $\sigma_x^t$ is recovered as soon as the corresponding vertex has received the required information. The following theorem makes this correspondence precise.

The space-time diagram of the simulation is similar to that of Fig.~\ref{fig:spacetimediagram}. The construction can easily be generalised to $d$-dimensional CA, 
and is likely to work for any local synchronous evolution in 
a broad sense ; the idea being to use vertices to encode 
applications of the local rule, and edges to represent the 
causal relationships of these applications.

\begin{Theorem}{Simulation of synchronous cellular automata}{Simulation of synchronous cellular automata}
    Let $F$ be a radius-one-half cellular automaton with initial configuration $\sigma^0$, and let $G$ be its encoding defined above. For every $t\in\mathbb{N}$, $x\in\mathbb{Z}$, and $H\in\mathcal{M}(G)$:
    \begin{enumerate}
        \item if $u=\lfloor t/2\rfloor.f_{2x-t}$ belongs to $H$ and has no incoming edge through its left port, then $\sigma_H^l(u)=\sigma_x^t$;
        \item if $v=\lfloor t/2\rfloor.f_{2(x-1)-t}$ belongs to $H$ and has no incoming edge through its right port, then $\sigma_H^r(v)=\sigma_x^t$.
    \end{enumerate}
    Moreover, for every $t$ and $x$, there exists a graph $H\in\mathcal{M}(G)$ in which the hypotheses of both items hold.
\end{Theorem}
\begin{proof}
    We proceed by induction on $t$. For $t=0$, both statements follow directly from the encoding of $\sigma^0$: the state $\sigma_x^0$ is stored in the left component of $0.f_{2x}$ and in the right component of $0.f_{2(x-1)}$.

    Suppose that both statements hold at time $t$. The update producing $\sigma_x^{t+1}$ is applied to the vertex
    \[
        w=\lfloor t/2\rfloor.f_{2(x-1)-t},
    \]
    whose left and right components contain, by the induction hypothesis, $\sigma_{x-1}^t$ and $\sigma_x^t$, respectively. By definition of the simulating rule, this update computes
    \[
        f(\sigma_{x-1}^t,\sigma_x^t)=\sigma_x^{t+1}
    \]
    and stores the result in both neighbouring vertices: in the left component of $\lfloor (t+1)/2\rfloor.f_{2x-(t+1)}$ and in the right component of $\lfloor (t+1)/2\rfloor.f_{2(x-1)-(t+1)}$. These components have no incoming edge on the corresponding side precisely once this information has been received. The claim follows for $t+1$.

    Finally, the dependency cone required to compute any fixed $\sigma_x^t$ is finite, thus there exists a finite valid sequence whose application produces both completed copies described above, proving the reachability statement.
\end{proof}

\subsection{Beyond synchronous simulation}\label{sec:exV2}

The previous examples are simulations of synchronous systems. We now consider dynamics that are not naturally interpreted through the lens of a global clock.
A paradigmatic, physics-inspired example is the following time-dilation dynamics.
The time-dilation example extends the internal state space $\Sigma$ used in Ex.~\ref{ex:PS---sshift} with two ``dilation'' states: the \textrm{green} and \textrm{red} states. It also extends the port set with the additional port \(l'\), so that \(\pi=\{l,l',r\}\). The same local rule is extended so that if one such green/red particle is found at position $x$, it will stay there, oscillating between both colours and altering the very texture of space-time, cf. Ex.~\ref{ex:TD---LR}. This results in time dilation as can be seen from Fig.~\ref{fig : Space time diagram example 2}.

\begin{example}{Time dilation---Local rule}{TD---LR}
        We define here the behaviour of $A_{x}$ when $u=t.x$ is a dilation vertex---i.e. $\sigma(u)= green $ or $ \sigma(u) =red$. In both cases particles get destroyed when reaching the dilation vertex and we flip the colour in $x$. On a green vertex (see $H_{\restri{x}}$), $A_x$ behaves as we are used to. On a red vertex (see $G_{\restri{x}}$) $A_x$ creates an anomaly. The edge between $(t+1).x$ and $w$ is reversed, thus we will be forced to apply $A_x$ again before updating $w$.\\
    \begin{tabular}{cccc}
        \begin{minipage}{0.22\linewidth}
            \captionsetup{justification=centering}
            \resizebox{0.7\textwidth}{!}{\begin{tikzpicture}


\redcellinternal{5}{5}{u};
\emptycell{1}{9}{setBorder}{v};
\emptycell{9}{13}{setBorder}{w};


\edge{u}{v}{myblack};
\edge{u}{w}{myblack};


\draw (6,4) node[right]{{\sizeFsix $t.x$}};
\draw (2.25,9) node[right]{{\sizeFsix $v$}};
\draw (6.5,13.5) node[right]{{\sizeFsix $w$}};

\draw (1.5,9) node {{\sizeFsixs $i$}};
\draw (8.5,13) node {{\sizeFsixs $j$}};

\draw (3.25,5.5) node[color = darkport] {{\sizeFsixp $:l$}};
\draw (6.75,6.5) node[color = darkport]  {{\sizeFsixp $:r$}};
\draw (3,8) node[color = darkport]  {{\sizeFsixp $:r$}};
\draw (7,11) node[color = darkport] {{\sizeFsixp $:l$}};

\end{tikzpicture}}
            \captionof*{figure}{$G_{\restri{x}}$.}
            \label{fig : Dynamic example 2; G} 
        \end{minipage}
         &
        \begin{minipage}{0.22\linewidth}
            \captionsetup{justification=centering}
            \resizebox{0.7\textwidth}{!}{\begin{tikzpicture}


\greencellinternal{5}{13}{u};
\emptycell{1}{9}{setBorder}{v};
\emptycell{9}{13}{setBorder}{w};


\edge{v}{u}{myblack};
\edge{u}{w}{myblack};


\draw (1,15) node[right]{{\sizeFsix $(t+1).x$}};
\draw (1.75,8) node[right]{{\sizeFsix $v$}};
\draw (8.5,11) node[right]{{\sizeFsix $w$}};


\draw (3,10.25) node[color = darkport] {{\sizeFsixp $:l$}};
\draw (6.5,13.5) node[color = darkport]  {{\sizeFsixp $:r$}};
\draw (3,12) node[color = darkport]  {{\sizeFsixp $:r$}};
\draw (7.5,12.25) node[color = darkport] {{\sizeFsixp $:l'$}};

\end{tikzpicture}}
            \captionof*{figure}{$A_{x}(G_{\restri{x}})$.}
            \label{fig : Dynamic example 2; A_x G}
        \end{minipage}
        &
        \begin{minipage}{0.22\linewidth}
            \captionsetup{justification=centering}
            \resizebox{0.7\textwidth}{!}{\begin{tikzpicture}


\greencellinternal{5}{5}{u};
\emptycell{1}{9}{setBorder}{v};
\emptycell{9}{9}{setBorder}{w};


\edge{u}{v}{myblack};
\edge{u}{w}{myblack};


\draw (6,4) node[right]{{\sizeFsix $u=t.x$}};
\draw (2.25,9) node[right]{{\sizeFsix $v$}};
\draw (6.25,9) node[right]{{\sizeFsix $w$}};

\draw (1.5,9) node {{\sizeFsixs $i$}};
\draw (8.5,9) node {{\sizeFsixs $j$}};

\draw (3.25,5.5) node[color = darkport] {{\sizeFsixp $:l$}};
\draw (6.5,5.5) node[color = darkport]  {{\sizeFsixp $:r$}};
\draw (3,8) node[color = darkport]  {{\sizeFsixp $:r$}};
\draw (6.75,8) node[color = darkport] {{\sizeFsixp $:l$}};

\end{tikzpicture}}
            \captionof*{figure}{$H_{\restri{x}}$.}
            \label{fig : Dynamic example 2; G green}     
        \end{minipage}
        &
        \begin{minipage}{0.22\linewidth}
            \centering
            \captionsetup{justification=centering}
            \resizebox{0.7\textwidth}{!}{\begin{tikzpicture}


\redcellinternal{5}{13}{u};
\emptycell{1}{9}{setBorder}{v};
\emptycell{9}{9}{setBorder}{w};


\edge{v}{u}{myblack};
\edge{w}{u}{myblack};


\draw (6,13.5) node[right]{{\sizeFsix $(t+1).x$}};
\draw (2.25,9) node[right]{{\sizeFsix $v$}};
\draw (6.25,9) node[right]{{\sizeFsix $w$}};

\draw (1.5,9) node {{\sizeFsixs $i$}};
\draw (8.5,9) node {{\sizeFsixs $j$}};

\draw (3,10.25) node[color = darkport] {{\sizeFsixp $:r$}};
\draw (6.75,10.25) node[color = darkport]  {{\sizeFsixp $:l$}};
\draw (3,12) node[color = darkport]  {{\sizeFsixp $:l$}};
\draw (6.75,12) node[color = darkport] {{\sizeFsixp $:r$}};

\end{tikzpicture}}
            \captionof*{figure}{$A_{x}(H_{\restri{x}})$.}
            \label{fig : Dynamic example 2; A_x G green}
        \end{minipage}
    \end{tabular}
\end{example}

For instance, if two identically made clocks were modelled out of a signal oscillating from left to right and right to left between two neighbouring nodes, the clock lying on the right-hand side of the dilation vertex would tick half as fast as the one lying on the left-hand side. Yet, the very same local rule is being applied left and right of the green/red particle. Such a phenomenon is reminiscent of general relativity, e.g. time flows slightly slower on Earth than it does in the stratosphere, as measured by identically made atomic clocks. Yet, the same laws of Physics apply in the stratosphere and on Earth. How did we get there? 

We argue that, to some extent, the construction of this model of computation mimics some of the key steps of the derivation of general relativity theory from physical symmetries---as found in standard textbooks \cite{Gravitation}. The following analogy can safely be skipped by the reader with lesser interest in Physics. Indeed, let us remind the reader that the textbook derivation of GR proceeds by: 1/ Assuming the existence of a well-determined space-time. 2/ Requiring covariance, i.e. invariance under changes of coordinates, which implies a form of asynchronism as one can choose coordinates whereby one region of a space-like cut will evolve (large time lapse), but not the other (small time lapse). 3/ Concluding that in order to obtain covariance, one needs to provide extra causality structure at each point, namely the metric field. 4/ Assuming background-invariance, namely enabling the possibility that space-time be curved by the presence of this newly allowed metric. 5/ Providing a dynamic upon the metric itself. \\
Here, in the discrete case, we: 1/ represent possible evolutions through a space-time diagram, 2/ generated by an asynchronous evaluation strategy, 3/ structured by an explicit causal DAG of dependencies. 4/ We then allow this DAG itself to vary, and 5/ consider rules that manipulate it.

\begin{figure}[h]
\centering
\resizebox{0.75\textwidth}{!}{\begin{tikzpicture}



\foreach \i in {1,2,3}{
    \foreach \j in {1,2,3}{
    	\FPeval{\posx}{clip(\i*8+4)}
        \FPeval{\posy}{clip(\j*8)}
        \ifnum \j=2
            {\vertex{\posx}{\posy}{border}{0}{\j}{\i\j}}
        \else
            {\vertex{\posx}{\posy}{myblack}{0}{\j}{\i\j}}
        \fi
    }
}

\foreach \i in {1,2,3,4}{
    \foreach \j in {1,2,3}{
    	\FPeval{\posx}{clip(\i*8)}
        \FPeval{\posy}{clip(\j*8+4)}
        \ifnum \j=2
            {\ifnum \i=4
            {\vertex{\posx}{\posy}{border}{0}{}{S\i\j}}
            \else
            \vertex{\posx}{\posy}{border}{0}{\j}{S\i\j}
            \fi}
        \else
            {\vertex{\posx}{\posy}{myblack}{0}{\j}{S\i\j}}
        \fi
    }
}


\foreach \i in {1,2,3}{
    \foreach \j in {1,2}{
    	\FPeval{\posx}{clip(\i*8+4*8-4+1.5)}
        \FPeval{\posy}{clip(\j*16-12)}
        {\vertex{\posx}{\posy}{myblack}{0}{\j}{Big\i\j}}
    }
}

\foreach \i in {1,2,3}{
    \foreach \j in {1,2}{
    	\FPeval{\posx}{clip(\i*8+4*8+1.5)}
        \FPeval{\posy}{clip(\j*16-4)}
        {\vertex{\posx}{\posy}{myblack}{0}{\j}{BigS\i\j}}
    }
}


\foreach \i in {1,2,3}{
    \FPeval{\k}{clip(\i+1)}
    \edge{\i1}{S\i1}{myblack}
    \edge{\i1}{S\k1}{myblack}
    
    \edge{S\i1}{\i2}{border}
    \edge{S\k1}{\i2}{border}
    \edge{\i2}{S\i2}{border}
    \edge{\i2}{S\k2}{border}
    
    \edge{S\i2}{\i3}{border}
    \edge{S\k2}{\i3}{border}
    \edge{\i3}{S\i3}{myblack}
    \edge{\i3}{S\k3}{myblack}
}

\foreach \i in {2,3}{
    \FPeval{\k}{clip(\i-1)}
    \edge{Big\i1}{BigS\i1}{myblack}
    \edge{Big\i1}{BigS\k1}{myblack}
    
    \edge{BigS\i1}{Big\i2}{border}
    \edge{BigS\k1}{Big\i2}{border}
    \edge{Big\i2}{BigS\i2}{myblack}
    \edge{Big\i2}{BigS\k2}{myblack}
}
\edge{Big11}{BigS11}{myblack}
\edge{Big11}{S41}{myblack}

\edge{BigS11}{Big12}{border}
\edge{S41}{Big12}{border}
\edge{S42}{Big12}{border}
\edge{Big12}{BigS12}{myblack}
\edge{Big12}{S43}{myblack}

\redcell{32}{12}{first}
\greencellgray{32}{20}{second}
\redcell{32}{28}{third}
\rightmoover{8}{12}{myblack}{P1A}
\rightmoover{24}{28}{myblack}{P1B}
\rightmoover{41.5}{12}{myblack}{P2A}
\rightmoover{49.5}{28}{myblack}{P2B}

\draw (7.5,24.5) node[right]{{\Huge \scalefont{2.5} $H'$}};
\draw (7.5,8.5) node[right]{{\Huge \scalefont{2.5} $H$}};

\draw (33.7,20.6) node[color=port]{\Huge $:r$};
\draw (35.7,20.6) node[color=port]{\Huge $:l'$};
\draw (33.7,13) node[color=port]{\Huge $:r$};
\draw (35.4,18.6) node[color=port]{\Huge $:l$};
\draw (39.4,18.6) node[color=port]{\Huge $:r$};
\draw (40,13) node[color=port]{\Huge $:l$};
\draw (39.1,21.6) node[color=port]{\Huge $:r$};
\draw (35.4,21.6) node[color=port]{\Huge $:l$};
\draw (33.7,27) node[color=port]{\Huge $:r$};
\draw (40,27) node[color=port]{\Huge $:l$};

\end{tikzpicture}}
\caption{{\em Time-dilation example.} In black we highlight two graphs $H$ and $H'$ belonging to the space-time diagram. We start with two particles, one on the left and the other on the right of the dilation vertex. Notice that, although the same local rule is applied everywhere, time flows twice as fast for the particle on the left.
\label{fig : Space time diagram example 2}}
\end{figure}
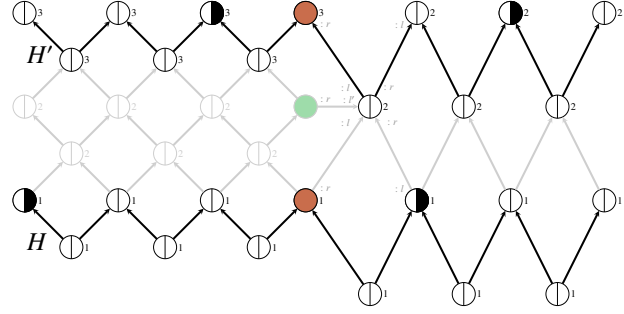

\section{Locality under sequential composition}

Having introduced valid sequences and illustrated the dynamics they generate, we now study whether locality is preserved under sequential composition. 
In other words: if each individual step is local, does the entire computation remain local?

We show that this is indeed the case, provided a suitable notion of locality for sequences. This leads to the notion of $*$-locality, which captures the behaviour of arbitrary finite sequences of rule applications.

This result has important consequences. In Sec.~\ref{sec:local rule}, we established structural properties of local rules, such as position preservation (Lem.~\ref{lem:easy name preservation}) and boundary preservation (Lem.~\ref{lem:border edges are preserved}) that emerge directly from locality. The present section shows that these properties are still true under sequential composition, because of $*$-locality.
Beyond this article, such a property is needed in any proof that reasons about sequences. A typical example is the space-time determinism theorem of~\cite{ArrighiSpaceTimeDetWRLA}: it relies on a privacy condition between disjoint sequences, requiring that the neighbourhoods $\restr{\omega}{G}$ and $\restr{\omega'}{G}$ not intersect except at their boundaries. Privacy is vacuous unless the action of $A_{\omega}$ remains inside $G_{\restri{\omega}}$: without $*$-locality, disjoint neighbourhoods would not imply disjoint supports.


\subsection{Locality for sequences}\label{sec: locality for sequences}

We now extend the notion of neighbourhood from a single position to a sequence of valid positions, in order to study how locality behaves under sequential rule applications: 
\begin{definition}[Neighbourhood of a sequence]\label{def:N_of_omega}
Given a neighbourhood scheme $\restri{}$, $G\in\sshift$ and $\omega\in \valid{G}{}$ , we define 
$$\restr{\omega}{G} := \bigcup_{ \beta x\alpha = \omega} \restr{x}{A_\alpha G}.$$
\RX{Notice that $\restr{\omega}{G}$ may contain new nodes, made out of those of $G$.}{Notice that $\restr{\omega}{G}$ may contain new nodes, made out of those of $G$.}
\todo[inline,color=cyan]{Monotony-1 is used sometimes used on $A_\alpha G$ Lem 4.11-12, so cannot project upon $G$. I'll happen as we induce.}
\end{definition}

\begin{example}{Particle system---Neighbourhood of a sequence}{}
    Consider Ex.~\ref{ex:PS---LR} and the graph $G\in \sshift$ pictured below. For the sequence $\omega = 6745$, we have $\restri{\omega} = \{3,4,5,6,7,8\}$ since it includes the neighbourhoods around positions already valid (such as $5$ and $7$), as well as those that become valid later (such as $4$).\\
    \begin{tabular}{ccc}
        \begin{minipage}{0.49\linewidth}
            \centering
            \includegraphics[width=\linewidth]{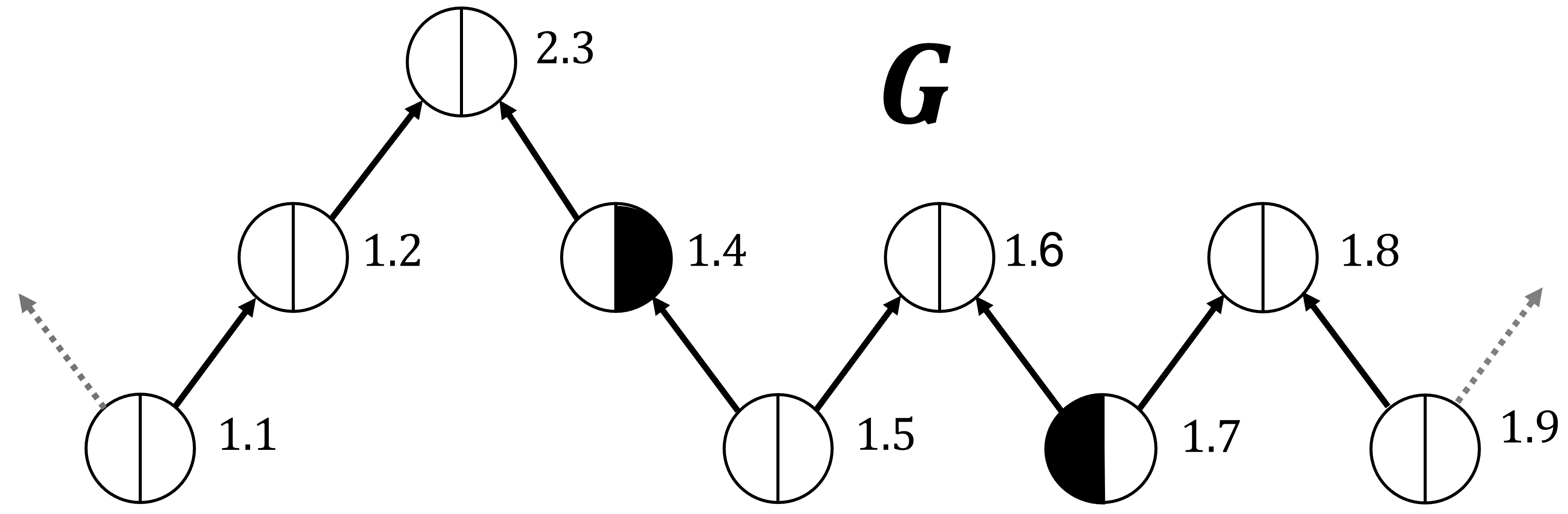}
        \end{minipage}
         &
         \begin{minipage}{0.49\linewidth}
            \centering
            \includegraphics[width=\linewidth]{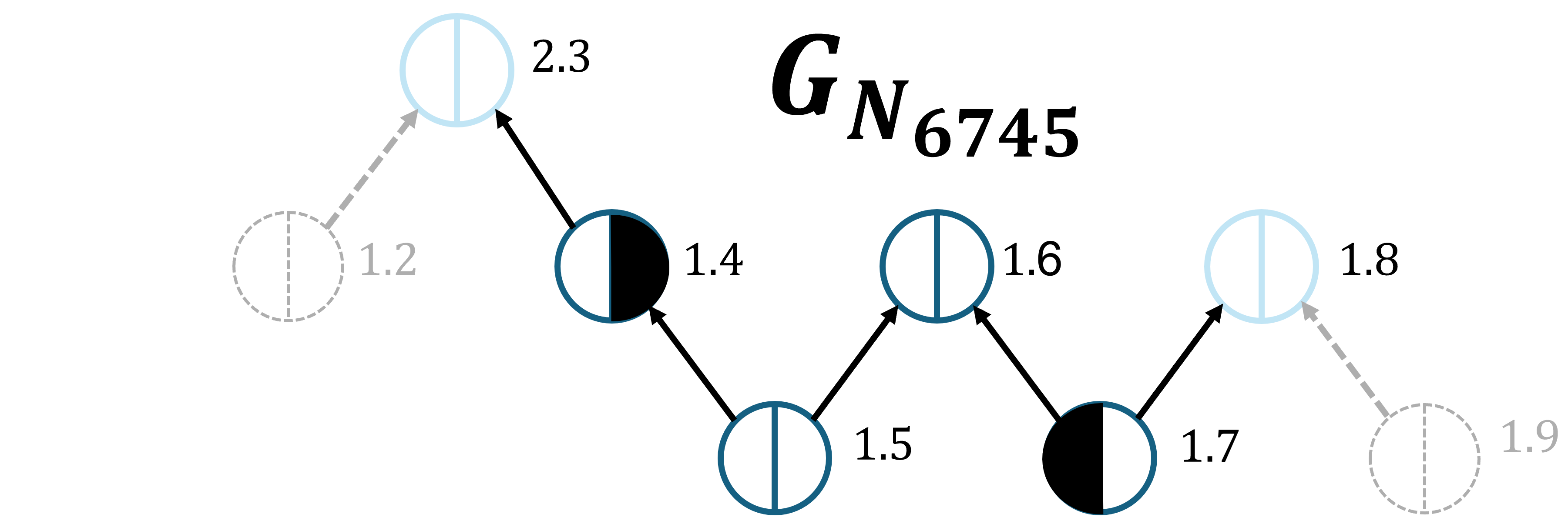}
        \end{minipage}
    \end{tabular}
\end{example}

We first verify that $\restr{\omega}{G}$ is monotone with respect to subsequences:
%
%
\begin{lemma}{Monotony}{monotony}
    All neighbourhood schemes are monotonous---i.e. for every sequence $\omega = \gamma  \beta  \alpha \in \valid{G}{A}$ we have 
    $\restr{\beta}{A_\alpha G} \subseteq \restr{\omega}{G}.$
\todo[inline,color=cyan]{Checked.}
\end{lemma}
\begin{proof}
$$\restr{\beta}{A_\alpha G} = \bigcup_{\omega_2 x \omega_1 = \beta} \restr{x}{A_{\omega_1}A_\alpha G}
         \subseteq \bigcup_{\omega_2 x \omega_1 = \beta\alpha} \restr{x}{A_{\omega_1} G}
         \subseteq \bigcup_{\omega_2 x \omega_1 = \gamma\beta\alpha} \restr{x}{A_{\omega_1} G}
         =\restr{\omega}{G}$$
\end{proof}
Note that this implies $\restr{\beta}{A_\alpha G} \subseteq \restr{\beta  \alpha}{G}$ and  $\restr{\alpha}{G} \subseteq \restr{\omega}{G}$.

We now state the main goal of this section. For size $1$ valid sequences, we have properties ensuring locality: \hyperref[neighbourhood : extensivity]{extensivity} of the neighbourhood scheme and \hyperref[def : locality]{$\restri{}$-locality} of the local rule. Our definition of neighbourhood scheme of a sequence will allow us to extend both properties to sequences of arbitrary size. We first formalise what it means for sequences to satisfy these properties.

\begin{definition}[$*$-locality, $*$-extensivity]
    Consider a neighbourhood scheme $\restri{}$. An operator $A_{(-)}$ is said to be $\omega$-local iff, for all $G$ for which $\omega\in \valid{G}{A}$ we have
    $$A_\omega G = A_\omega (G_\restri{\omega}) \sqcup G_\crestri{\omega}.$$
    An operator is said to be local iff this holds for any $\omega\in \valid{G}{A}$ such that $|\omega|=1$, in which case we call it a local rule. It is said to be $n$-local if this holds for any $\omega$ of length at most $n$, and $*$-local if it holds for any $\omega\in \valid{G}{A}$. \\
    In a similar manner, given a local rule $A_{(-)}$, if the formula $G_{\restri{\omega}} \sqsubseteq H\sqsubseteq G \implies \omega\in \valid{H}{A}\land G_{\restri{\omega}} = H_{\restri{\omega}}$ holds for any $\omega$ of length at most $n$ such that $\omega$ is valid on $G$, we call it $n$-extensive. If it holds for any $\omega$ we call it $*$-extensive.
\end{definition}

Then we prove the properties of $*$-extensivity and $*$-locality by induction. But because they are interdependent, each induction step breaks down into two. First we advance one leg and show $(n+1)$-extensivity. Second we advance the other leg and show $(n+1)$-locality. Combining these yields the properties.

 \begin{lemma}{$n$-locality implies $(n+1)$-extensivity}{*-extensivity}
     Let $n\geq 0$. Consider a neighbourhood scheme $\restri{}$, an $n$-local rule $A_{(-)}$, a graph $G$ and $\omega \in \valid{G}{A}$ of length at most $n+1$. Then $G_{\restri{\omega}} \sqsubseteq H\sqsubseteq G \text{ implies $\omega\in \valid{H}{A}$ and } G_{\restri{\omega}} = H_{\restri{\omega}}$.
     \todo[inline,color=cyan]{Checked.}
 \end{lemma}
 \begin{proof}
    We prove $(n+1)$-extensivity by finite induction. The basis case is $1$-extensivity, which comes immediately from extensivity of $\restri{}$. Our induction hypothesis is $k$-extensivity, with $1<k<n+1$.
    
    Consider $y\omega\in \valid{G}{A}$ with $|\omega|=k$. 
    Suppose that $G_{\restri{y\omega}} \sqsubseteq H\sqsubseteq G$. Notice the following decomposition :
    $$\restr{y\omega}{G} = \bigcup_{\alpha x\beta = y \omega}\restr{x}{A_\beta G} = \restr{y}{A_\omega G}\cup \bigcup_{\substack{\alpha x\beta = y \omega\\ |\alpha |>0}}\restr{x}{A_\beta G} = \restr{y}{A_\omega G}\cup \restr{\omega}{G}$$
    As long as $y\omega$ is valid on $H$, we can derive the same way $\restr{y\omega}{H} =\restr{y}{A_\omega H}\cup \restr{\omega}{H}$.
    
    Thus to prove $y\omega\in\valid{H}{A}$ and $\restr{y\omega}{G} = \restr{y\omega}{H}$, we only have to prove, on the one hand $\omega\in\valid{H}{A}$ and $\restr{\omega}{G} = \restr{\omega}{H}$, and on the other hand $A_\omega H\in \graphvalid{\restri{}}{y}$ and $\restr{y}{A_\omega G} = \restr{y}{A_\omega H}$.

    Since by monotony (Lem.~\ref{lem:monotony}) we have $G_{\restri{\omega}}\sqsubseteq G_{\restri{y\omega}} \sqsubseteq H\sqsubseteq G$
    the $k$-extensivity induction hypothesis readily gives $\omega\in\valid{H}{A}$ and $\restr{\omega}{G} = \restr{\omega}{H}$. 

    Then we prove $\restr{y}{A_\omega G} = \restr{y}{A_\omega H}$. Notice that by \hyperref[neighbourhood : stability]{stability}\ 
    and \hyperref[neighbourhood : extensivity]{extensivity} of $\restri{}$ it is enough to prove $$(A_\omega G)_{\restr{y}{A_\omega G}}\sqsubseteq A_\omega H \sqsubseteq A_\omega G$$
    We start with the second inclusion :
    \begin{align*}
        A_\omega G &= A_\omega G_{\restr{\omega}{G}} \sqcup G_{\crestr{\omega}{G}} \tag{$n$-locality}\\
        &= A_\omega H_{\restr{\omega}{G}} \sqcup G_{\crestr{\omega}{G}} \tag{$G_{\restr{\omega}{G}} \sqsubseteq H$}\\
        &= A_\omega H_{\restr{\omega}{H}} \sqcup G_{\crestr{\omega}{H}} \tag{$\restr{\omega}{G} = \restr{\omega}{H}$}\\
        &\sqsupseteq A_\omega H_{\restr{\omega}{H}} \sqcup H_{\crestr{\omega}{H}} \tag{$H\sqsubseteq G$}\\
        &= A_\omega H
    \end{align*}

    Finally we prove the last inclusion. First we notice that by monotony (Lem.~\ref{lem:monotony}) we have $G_{\restr{y}{A_\omega G}}\sqsubseteq G_{\restr{y\omega}{G}}\sqsubseteq H$. This notably implies that $(G_X)_{\restr{y}{A_\omega G}}\sqsubseteq (H_X)_{\restr{y}{A_\omega G}}$ for any set of positions $X$. Then we derive :
    \begin{align*}
        (A_\omega G)_{\restr{y}{A_\omega G}}&= (A_\omega H_{\restr{\omega}{H}} \sqcup G_{\crestr{\omega}{H}})_{\restr{y}{A_\omega G}}\\
        &= (A_\omega H_{\restr{\omega}{H}})_{\restr{y}{A_\omega G}} \sqcup (G_{\crestr{\omega}{H}})_{\restr{y}{A_\omega G}}\\
        &\sqsubseteq (A_\omega H_{\restr{\omega}{H}})_{\restr{y}{A_\omega G}} \sqcup (H_{\crestr{\omega}{H}})_{\restr{y}{A_\omega G}}\\
        &\sqsubseteq A_\omega H
    \end{align*}
 \end{proof}

Before we proceed to the other leg of the induction, we need a lemma showing that if one is local to a neighbourhood, one is also local to a bigger region. 
\begin{lemma}{$\restri{\omega}\subseteq X$-locality implies $X$-locality}{petit local implique grand local version générale}
    Let $n\geq 0$. 
    Consider an $n$-extensive neighbourhood scheme $\restri{}$, a graph $G$ and $\omega \in \mathcal{X}^*$ of length at most $n$, and any set $X\subseteq \mathcal{X}$ such that $\restr{\omega}{G}\subseteq X$. If $A_{(-)}$ is $\omega$-local then we have:
    $$A_\omega G = (A_\omega G_{X})\sqcup G_{\overline{X}}$$
\todo[inline,color=cyan]{Checked.}
\end{lemma}
\begin{proof}
In order to lighten the notations of this proof we temporarily write $\restri{}$ instead of $\restri{\omega}$.
By $n$-extensivity of $\restri{}$, $G_{\restr{}{G}}\sqsubseteq G_{X} \sqsubseteq G$ implies $\restr{}{G}=\restr{}{G_{X}}$.
\begin{align*}
    A_\omega G_{X}&= A_\omega G_{X{}\restr{}{G_{X}}}\sqcup G_{X{}\crestr{}{G_{X}}}\\
    &= A_\omega G_{X{}\restr{}{G}}\sqcup G_{X\crestr{}{G}}\textrm{ by extensivity}\\
    &= A_\omega G_{\restr{}{G}}\sqcup G_{X\cap\crestr{}{G}}\textrm{ since }\restr{}{G}\subseteq X\\
    A_\omega G_{X}\sqcup G_{\overline{X}}&=A_\omega G_{\restr{}{G}}\sqcup G_{X\cap\crestr{}{G}} \sqcup G_{\overline{X}}\\
    &=A_\omega G_{\restr{}{G}}\sqcup G_{X\cap\crestr{}{G}} \sqcup G_{\overline{X}\cap\crestr{}{G}} \textrm{ since } \overline{X}\subseteq\crestr{}{G}\\
    &=A_\omega G_{\restri{}(G)}\sqcup G_{\crestri{}(G)}\\
    &= A_\omega G
\end{align*}
\end{proof}

Now we prove that any local operator for a $(n+1)$-extensive neighbourhood scheme is also $(n+1)$-local.
\begin{lemma}{$(n+1)$-extensivity implies $(n+1)$-locality}{extloc}
    Let $n\geq 0$. 
    If an operator $A_{(-)}$ is local and $(n+1)$-extensive, then it is also $(n+1)$-local.
\todo[inline,color=cyan]{Checked.}
\end{lemma}
\begin{proof}
    By recurrence on the size of $\omega$. The base case $|\omega|=1$ is locality. The induction hypothesis is $n$-locality. Consider $|\omega|=n+1$ and let $\omega = \beta \alpha$ with $\beta,\alpha$ non-empty. By monotony (Lem.~\ref{lem:monotony}) we have that for all $G$, $\restr{\alpha}{G}\subseteq \restr{\omega}{G}$. We can thus use Lem.~\ref{lem:petit local implique grand local version générale} to get:  
\begin{align*}
A_\alpha G &= (A_\alpha G_{\restri{\omega}})\sqcup G_{\crestri{\omega}} \\
A_\beta A_\alpha G &= A_\beta ((A_\alpha G_{\restri{\omega}})\sqcup G_{\crestri{\omega}})
\end{align*}
By monotony we also have that $\restr{\beta}{A_\alpha G} \subseteq \restr{\omega}{G}$. We can thus use Lem.~\ref{lem:petit local implique grand local version générale} again with $X = \restr{\omega}{G}$: 
\begin{align*}
A_\beta A_\alpha G &= \left(A_\beta \left(((A_\alpha G_{\restri{\omega}}) \sqcup G_{\crestri{\omega}})_{ \restr{\omega}{G}}\right)\right) \sqcup ((A_\alpha G_{\restri{\omega}})\sqcup G_{\crestri{\omega}})_{\overline{ \restri{\omega}}(G)} \\
&= A_\beta (A_\alpha G_{\restri{\omega}} 
)_{ \restr{\omega}{G}} \sqcup ( 
G_{\crestri{\omega}})_{\overline{ \restri{\omega}}(G)} \\
&=(A_\beta A_\alpha G_{\restri{\omega}}) \sqcup  G_{\crestri{\omega}}
\end{align*}
The second equality holds because repeated applications of position preservation (Lem.~\ref{lem:easy name preservation}) imply $\sp{\I_{A_\alpha G_{\restri{\omega}}}}\subseteq \restri{\omega}$.
\end{proof}

Using Lem~\ref{lem:extloc} and Lem~\ref{lem:*-extensivity} we can then immediately deduce by induction both $*$-extensivity and $*$-locality for any local rule. We illustrate in Ex.~\ref{ex:PS---*L} how $*$-locality holds for the particle system example.

\begin{Theorem}{$*$-locality and $*$-extensivity}{ star ext and star loc}
    If an operator $A_{(-)}$ is local, then it is also $*$-local and $*$-extensive.
\end{Theorem}

\begin{example}{Particle system---$*$-locality}{PS---*L}
    The action of $A_{6745}$ on $G$ can be computed using only the information contained in $G_{\restri{6745}}$. In other words, we have $*$-locality, that is $A_{6745}G = A_{6745}G_{\restri{6745}}\sqcup G_{\overline{\restri{6745}}}$.\\
    \begin{tabular}{ccc}
        \begin{minipage}{0.49\linewidth}
            \centering
            \includegraphics[width=0.9\linewidth]{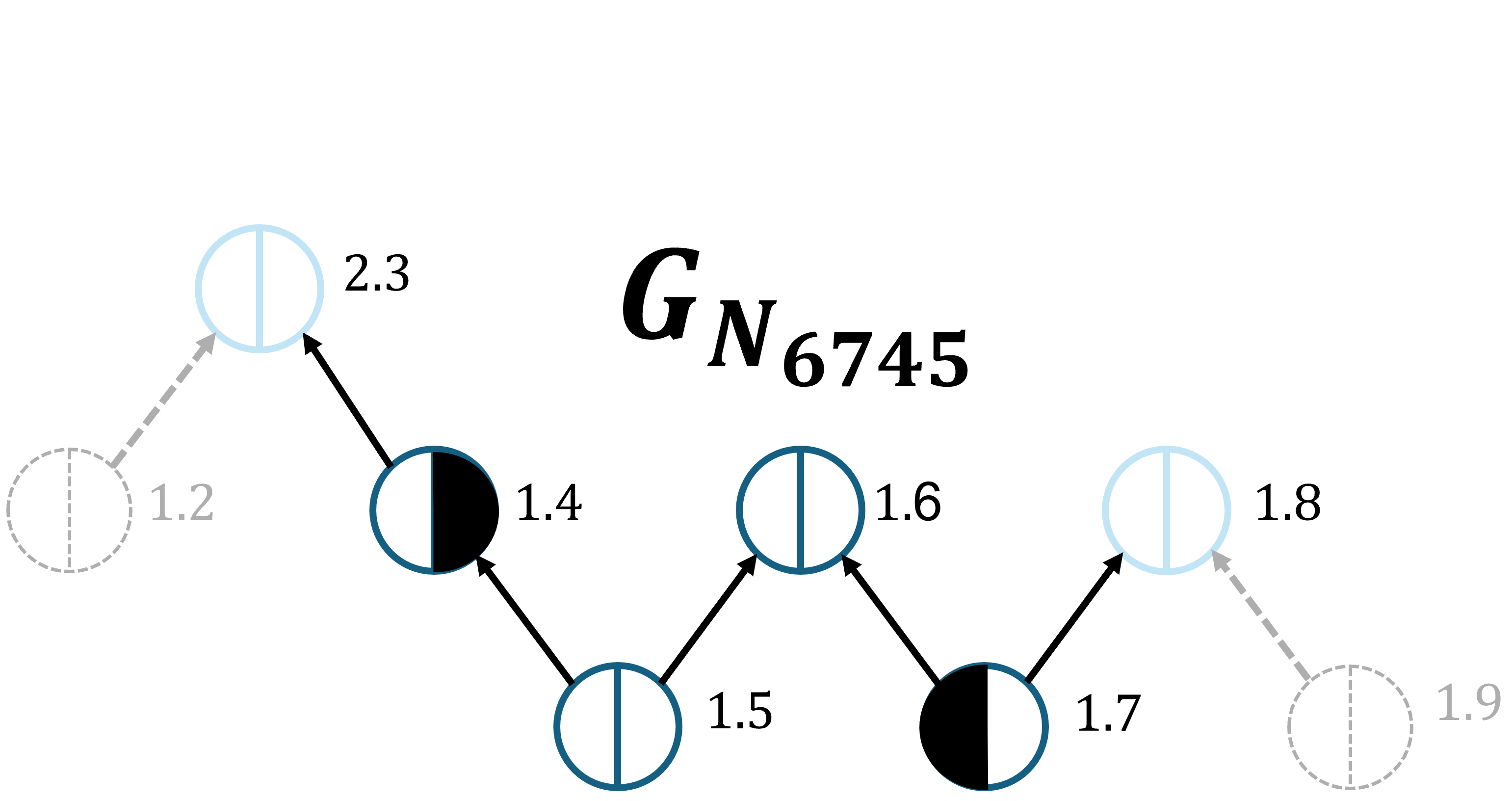}
        \end{minipage}
         &
         \begin{minipage}{0.49\linewidth}
            \centering
            \includegraphics[width=0.9\linewidth]{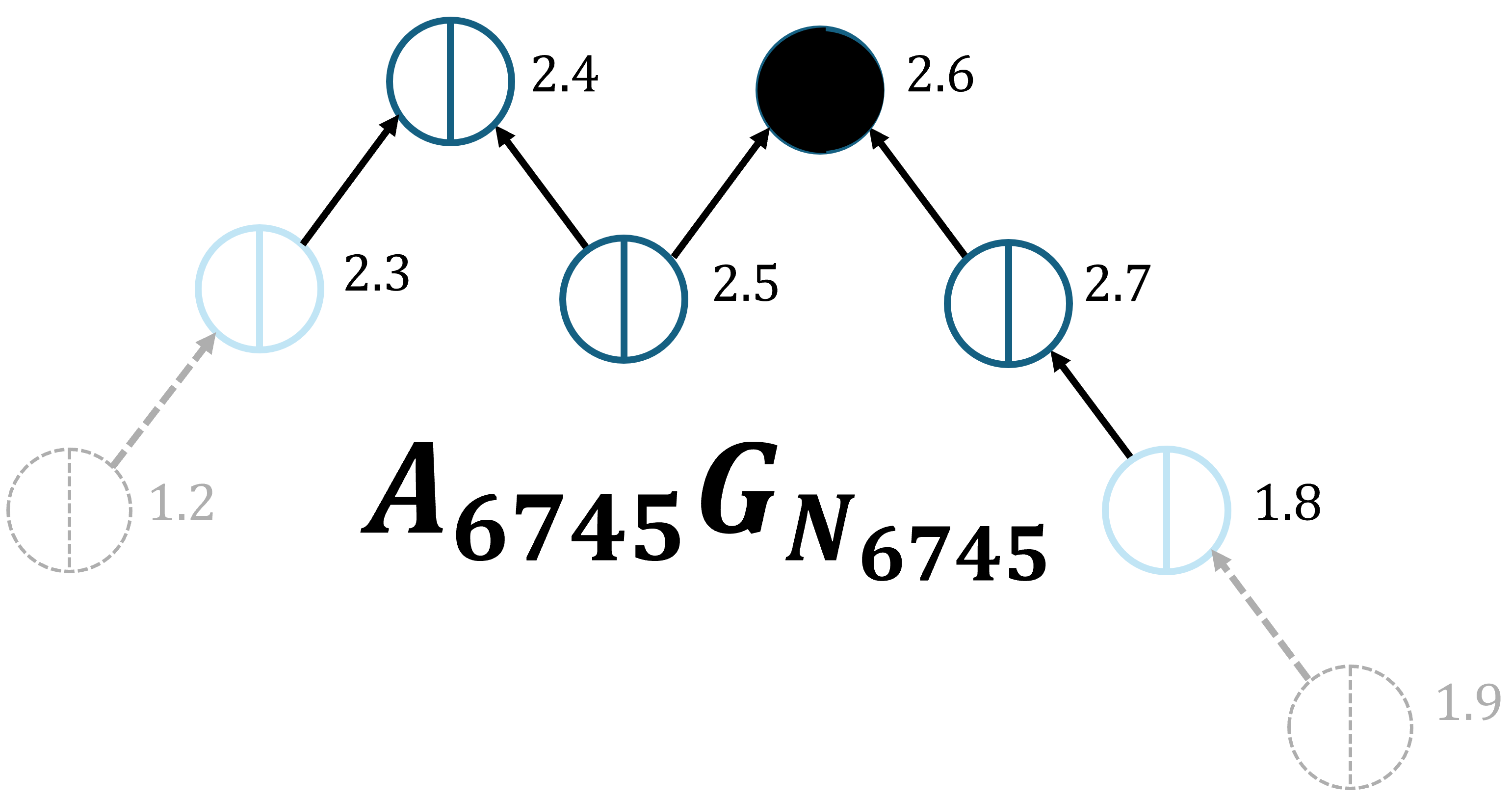}
        \end{minipage}
    \end{tabular}
\end{example}

Having proved $*$-locality, we can now derive in a straightforward way the same interface properties as before, from the well-definedness of the same $\sqcup$ decomposition.
Namely, we recover position preservation (Lem.~\ref{lem:easy name preservation}) and boundary preservation (Lem.~\ref{lem:border edges are preserved}), which in turn imply that only names of interior vertices are modified (Rem.~\ref{rem:boundaries are stable}).


\begin{corollary}{$*$-Position and Boundary preservation}{* easy name preservation}
  Let $A_{(-)}$ be a local rule, let $H\in\sshift$ and $\omega\in\valid{H}{A}$, and set $G=H_{\restri{\omega}}$. If $\V_G\subseteq \I_H$, then
  \begin{align*}
      \sp{\I_{A_\omega G}}&\subseteq\sp{\I_G},\\
      \forall u\in \I_G,\ \sp{u}\in\restri{\omega}\setminus\restri{\omega}^-&\implies u\in \I_{A_\omega G},\\
      \dangling{G}&=\dangling{A_\omega G}.
  \end{align*}
  Consequently, all new names are at positions of $\restri{\omega}^-$, i.e. $\sp{\I_{A_\omega G}\setminus \I_G}\subseteq\restri{\omega}^-$, where the interior $\restri{\omega}^-$ is computed in $H$.
\end{corollary}
\begin{proof}
    Renaming invariance is preserved by composition. Moreover, each step replaces a disk while preserving its border and leaving its context unchanged; hence, by induction on $|\omega|$, $A_\omega$ preserves the border of the disk $G_{\restri{\omega}}$. Together with $*$-locality (Th.~\ref{th: star ext and star loc}), these are precisely the properties used in the proofs of position preservation (Lem.~\ref{lem:easy name preservation}) and boundary preservation (Lem.~\ref{lem:border edges are preserved}). Applying the same arguments with $x$ replaced by $\omega$ gives the three stated properties. The final inclusion then follows as in Rem.~\ref{rem:boundaries are stable}.
\end{proof}

\section{Conclusion}\label{sec:conclusion}

\paragraph*{Summary of contributions.}

We introduced graphs that can be thought of as space-like cuts of space-time diagrams. The vertices have names of the form $u=t.x$ and can be thought of as events; the edges encode dependencies between events, i.e. if $u=t.x$ points to $v=t'.y$, then $v$ is ahead in time of $u$, awaiting information from $u$. The action of a local rule $A_x$ on $u$ is non-trivial only if $u$ is minimal: it disposes of it by communicating its information to its dependants, and creates a fresh vertex $u'=(t+1).x$ in some provisional state. We also characterised the structural effect of such an update: local rules preserve positions (Lem.~\ref{lem:easy name preservation}) and, under the extendability hypothesis, preserve boundary vertices and dangling edges (Lem.~\ref{lem:border edges are preserved}). Consequently, names can be modified only in the interior of the selected neighbourhood (Rem.~\ref{rem:boundaries are stable}).

We developed the particle-system and time-dilation examples introduced in~\cite{ArrighiSpaceTimeDetWRLA}, giving precise local rules and space-time diagrams. We also formalised the cellular-automaton simulation sketched there: Th.~\ref{th:Simulation of synchronous cellular automata} proves that any one-dimensional cellular automaton can be simulated via a marching-soldiers encoding, in which the DAG of dependencies tracks the relative advancement of computations across space. The time-dilation rule (Sec.~\ref{sec:exV2}), in turn, illustrates fundamentally asynchronous behaviour without a natural global-clock interpretation.

Our main technical result is that locality is preserved under sequential composition of rule applications: any local rule is automatically $*$-local and $*$-extensive (Th.~\ref{th: star ext and star loc}), meaning that the neighbourhood and locality properties that hold for a single rule application extend canonically to arbitrary valid sequences. The proof proceeds by an inductive bootstrapping between $n$-extensivity and $n$-locality (Lem.~\ref{lem:*-extensivity}, \ref{lem:extloc}), relying on the monotony of neighbourhood schemes (Lem.~\ref{lem:monotony}) and yielding as corollaries the preservation of boundary and position also along sequences (Cor.~\ref{cor:* easy name preservation}).


\paragraph*{Open questions.}

We proved that causal graph rewriting can simulate any one-dimensional cellular automaton. A natural question is whether the simulation preserves the global space-time structure, not just the local behaviour: do the space-times of local rules that simulate cellular automata correspond to expansive graph subshifts~\cite{ArrighiSubshifts}, in the same way that space-times of cellular automata correspond to expansive tilings? One difficulty is that the naming conventions we adopt for vertices systematically prevent space-times from being cyclic, even when the dynamics described is periodic---whereas the corresponding graph subshift will in fact be cyclic.

A more foundational question concerns our notion of validity. We require that each updated position be past at the time of its update. One may ask whether a more permissive notion---allowing updates to positions that are not yet past but whose neighbourhood is already determined---could yield a richer class of dynamics whilst preserving $*$-locality.

\paragraph*{Perspectives.}

\noindent{\em Lightweight synchronisation for parallel computation.} An immediate practical motivation for asynchronous models is that clock synchronisation is an expensive overhead in parallel simulation of dynamical systems and in distributed computation. The causal graph rewriting framework provides a principled alternative: by encoding computational dependencies directly in the DAG structure, it allows different regions to advance at their own pace without global coordination, while the dependency edges ensure that no event is processed before its inputs are ready.

\noindent{\em Space-time determinism.} A deeper theoretical question is whether the non-determinism introduced by asynchrony is entirely harmless: do any two schedulings agree on the state of every event they both describe? This property is formalised as the main result of the short companion paper~\cite{ArrighiSpaceTimeDetWRLA}. The $*$-locality and $*$-extensivity results of the present article are natural prerequisites for that development: they reduce the comparison of two schedulings $\omega$ and $\omega'$ to a purely local question---two schedulings agree on the state of a vertex $v$ whenever their respective neighbourhoods $\restr{\omega}{G}$ and $\restr{\omega'}{G}$ overlap at $v$ in a compatible way---so that space-time determinism follows from a local condition on the rule ensuring that such overlaps are always compatible. A separate article will unravel the additional properties of the model needed for space-time determinism and provide a self-contained presentation of the proof material underlying~\cite{ArrighiSpaceTimeDetWRLA}.

\noindent{\em Reversibility and quantum extension.} Both cellular automata and causal graph dynamics have followed the trajectory from classical to reversible to quantum~\cite{KariBlock,ArrighiRCGD,ArrighiOverview,ArrighiQNT}. Causal graph rewriting is set to follow the same path: the reversible subfamily is the natural next step, and constitutes a prerequisite for the quantum extension.

\paragraph*{Acknowledgements}~
This project was partially funded by the European Union through the MSCA SE project QCOMICAL.
It was also funded by the French National Research Agency (ANR): project TaQC ANR-22-CE47-0012 and within the framework of `Plan France 2030', under the research projects EPIQ ANR-22-PETQ-0007, OQULUS ANR-23-PETQ-0013, HQI-Acquisition ANR-22-PNCQ-0001 and HQI-R\&D ANR-22-PNCQ-0002. 
The project was also supported by the WOST, WithOut SpaceTime project (\url{https://withoutspacetime.org}), grant ID\# 63683 from the John Templeton Foundation (JTF). The opinions expressed in this work are those of the author(s) and do not necessarily reflect the views of the John Templeton Foundation.
Finally, this work was supported by the F.R.S.-FNRS under project CHEQS within the Excellence of Science (EOS) program.

\bibliography{biblio}

\end{document}